\documentclass[11pt,reqno]{article}
\usepackage{amssymb,amsmath,amsfonts}
\usepackage{amsthm}
\usepackage{graphicx}
\usepackage{epstopdf}
\usepackage{bm}
\usepackage{paralist}
\usepackage{color}
\usepackage{fullpage}
\usepackage{algorithm}
\usepackage[noend]{algorithmic}
\usepackage[font={small}]{caption} 
\usepackage{subcaption}

\usepackage[usenames,dvipsnames,svgnames,table]{xcolor}
\usepackage[pagebackref,letterpaper=true,colorlinks=true,pdfpagemode=none,citecolor=OliveGreen,linkcolor=black,urlcolor=BrickRed,pdfstartview=FitH]{hyperref}

\newcommand{\bea}{\begin{eqnarray}}
\newcommand{\eea}{\end{eqnarray}}
\newcommand{\<}{\langle}
\renewcommand{\>}{\rangle}
\newcommand{\E}{{\mathbb E}}

\def\Unif{{\sf Unif}}
\def\BI{{\rm I}}

\def\eps{{\varepsilon}}
\def\normeq{\cong}
\def\rE{{\rm E}}
\def\id{{\rm I}}
\def\bh{\boldsymbol{h}}
\def\Psp{{\mathfrak{P}}}

\def\hx{\hat{x}}

\def\cond{{\sf c}}
\def\bcond{\underline{\sf c}}
\def\ro{\overline{{\sf c}}}

\def\btau{{\boldsymbol{\tau}}}
\def\bsigma{{\boldsymbol{\sigma}}}
\def\bsm{{\boldsymbol{\sigma}}_{\rm M}}
\def\bxi{{\boldsymbol{\xi}}}
\def\bfe{{\boldsymbol{e}}}
\def\bap{\overline{p}}

\def\bZ{{\boldsymbol{Z}}}

\def\hgamma{\widehat{\gamma}}

\def\dc{{\Delta{\sf c}}}
\def\bC{{\boldsymbol{C}}}
\def\bQ{{\boldsymbol{Q}}}
\def\bM{{\boldsymbol{M}}}
\def\hbQ{{\boldsymbol{\widehat{Q}}}}
\def\hbu{{\boldsymbol{\widehat{u}}}}
\def\bmu{{\boldsymbol{\mu}}}
\def\bg{{\boldsymbol{g}}}
\def\hnu{\widehat{\nu}}
\def\bzero{{\mathbf 0}}
\def\bone{{\mathbf 1}}
\def\cF{{\mathcal F}}
\def\cG{{\mathcal G}}
\def\cT{{\mathcal T}}
\def\cC{{\mathcal C}}

\def\cO{{\mathcal O}}

\def\Z{{\mathbb Z}}
\def\RS{\mbox{\tiny RS}}

\def\field{{\mathbb F}}
\def\reals{{\mathbb R}}
\def\complex{{\mathbb C}}
\def\integers{{\mathbb Z}}
\def\Z{{\mathbb Z}}

\def\sg{s_{\mathfrak G}}
\def\normal{{\sf N}}
\def\cnormal{{\sf CN}}
\def\sT{{\sf T}}
\def\Group{{\mathfrak G}}
\def\sg{s_{{\mathfrak G}}}

\def\oh{\overline{h}}
\def\bv{{\boldsymbol{v}}}
\def\bz{{\boldsymbol{z}}}
\def\bx{{\boldsymbol{x}}}
\def\ba{{\boldsymbol{a}}}
\def\bb{{\boldsymbol{b}}}
\def\bA{\boldsymbol{A}}
\def\bB{\boldsymbol{B}}
\def\bxz{\boldsymbol{x_0}}
\def\hbx{\boldsymbol{\hat{x}}}

\def\lBayes{\lambda_{c}^{\mbox{\tiny{Bayes}}}}
\def\lml{\lambda_{c}^{\mbox{\tiny{ML}}}}
\def\lsdp{\lambda_{c}^{\mbox{\tiny{SDP}}}}
\def\tlsdp{\tilde{\lambda}_{c}^{\mbox{\tiny{SDP}}}}
\def\sopt{{\mbox{\tiny{opt}}}}

\def\OPT{{\rm OPT}}
\def\de{{\rm d}}

\def\csdp{c^{\mbox{\tiny{SDP}}}}
\def\cpca{c^{\mbox{\tiny{PCA}}}}
\def\xpca{\boldsymbol{\hat{x}}^{\mbox{\tiny{PCA}}}}
\def\xml{\boldsymbol{\hat{x}}^{\mbox{\tiny{ML}}}}
\def\xsdp{\boldsymbol{\hat{x}}^{\mbox{\tiny{SDP}}}}
\def\xBayes{\boldsymbol{\hat{x}}^{\mbox{\tiny{Bayes}}}}
\def\bX{\boldsymbol{X}}
\def\bY{\boldsymbol{Y}}
\def\bW{\boldsymbol{W}}
\def\prob{{\mathbb P}}
\def\E{{\mathbb E}}
\def\MSE{{\rm MSE}}
\def\Ove{{\rm Overlap}}
\def\<{\langle}
\def\>{\rangle}
\def\Tr{{\sf Tr}}
\def\T{{\sf T}}
\def\argmax{{\arg\!\max}}
\def\sign{{\rm sign}}

\def\diag{{\rm diag}}
\def\ed{\stackrel{{\rm d}}{=}}
\def\Poisson{{\rm Poisson}}
\def\cM{{\cal M}}
\def\usigma{{\underline{\bsigma}}}
\def\Pp{{\sf P}^{\perp}}
\def\obv{{\boldsymbol{\overline{v}}}}
\def\di{{\partial i}}
\def\conv{{\rm conv}}

\def\by{{\boldsymbol{y}}}
\def\bw{{\boldsymbol{w}}}
\def\P{{\sf P}}

\def\cov{{\rm Cov}}
\def\tol {{\tt tol}}

\def\be{{\boldsymbol{e}}}

\def\DeltaMax{\Delta_\text{max}}

\def\bphi{{\boldsymbol{\varphi}}}

\def\hG{\widehat{G}}
\def\hF{\widehat{F}}

\def\bhSigma{{\bf\widehat{\Sigma}}} 
\def\Bind{{\sf Bind}}
\def\Var{{\rm Var}}

\newtheorem{claim}{Claim}[section]
\newtheorem{lemma}[claim]{Lemma}

\newtheorem{proposition}[claim]{Proposition}
\newtheorem{corollary}[claim]{Corollary}

\newtheorem{remark}[claim]{Remark}

\title{
Phase Transitions in Semidefinite Relaxations}

\author{Adel Javanmard\footnote{USC Marshall School of Business, University of Southern California},\;\;
Andrea~Montanari\footnote{Department of Electrical
    Engineering and Department of Statistics, Stanford University}
\, and\, Federico Ricci-Tersenghi\footnote{Dipartimento di Fisica,
    Universit\'a di Roma, La Sapienza}}

\date{\today}                                           % Activate to display a given date or no date

\begin{document}
\maketitle

\begin{abstract}
Statistical inference problems arising within signal processing, data
mining, and machine learning naturally give rise to hard combinatorial
optimization problems. These problems become intractable when the 
dimensionality of the data is large, as is often  the case for modern
datasets. A popular idea is to construct convex
relaxations of these combinatorial problems, which can be solved
efficiently for large scale datasets.

Semidefinite programming (SDP) relaxations are among the most powerful
methods in this family, and are surprisingly well-suited for a broad range of
problems where data take the form of matrices or graphs.
It has been observed several times that, when the `statistical noise'
is small enough, SDP relaxations correctly detect the underlying
combinatorial structures. 

In this paper we develop asymptotic predictions for several
`detection thresholds,' as well as for the estimation error above
these thresholds.
We study some classical SDP relaxations for statistical problems
motivated by graph synchronization and community detection in
networks. We map these optimization problems to statistical mechanics
models with vector spins, and use non-rigorous techniques from
statistical mechanics to characterize the corresponding phase
transitions. Our results clarify the effectiveness of SDP relaxations
in solving high-dimensional statistical problems. 
\end{abstract}

\tableofcontents

\section{Introduction}

Many information processing tasks can be formulated 
as optimization problems. This idea has been central to data analysis
and statistics at least since Gauss and Legendre's invention of the
least-squares method in the early 19th century \cite{gauss1809theoria}. 

Modern datasets pose new challenges to this 
centuries' old framework. On the one hand, high-dimensional applications
require to estimate simultaneously  millions of
parameters. Examples span genomics \cite{ben1999clustering}, imaging \cite{plaza2009recent},
web-services \cite{KBV09}, and so on. On the other hand, the unknown object to
be estimated has often a combinatorial structure: In
clustering we aim at estimating a partition of the data points \cite{von2007tutorial}.
Network analysis  tasks usually require to identify a
discrete subset of nodes in a  graph \cite{girvan2002community,krzakala2013spectral}. 
Parsimonious data explanations are sought by imposing
combinatorial sparsity constraints \cite{wasserman2000bayesian}.

There is an obvious tension between the above requirements. 
While efficient algorithms are needed to estimate a
large number of parameters, the maximum likelihood method often
requires to solve  NP-hard
combinatorial optimizations. A flourishing  line of work
addresses this conundrum by designing effective convex relaxations of
these combinatorial problems \cite{Tibs96,chen1998atomic,candes2010power}. 

Unfortunately, the statistical 
properties of such convex relaxations are well understood only in
a few cases (compressed sensing being the most important success
story \cite{DoTa05,Dantzig,DMM09,amelunxen2014living}). In
this paper we use tools from statistical mechanics
to develop a precise picture of the behavior of a class of semidefinite programming relaxations.
Relaxations of this type appear to be surprisingly effective in a
variety of problems ranging from clustering to graph synchronization.
For the sake of concreteness we will focus on three specific problems:

\vspace{0.25cm}

\noindent {\bf $\Z_2$-synchronization.} In the general synchronization 
problem, we aim at estimating $x_{0,1}, x_{0,2}, \dots, x_{0,n}$ which
are unknown elements of a known  group $\Group$.  This is done 
using data that consists of noisy observations of `relative positions'
$Y_{ij}= x_{0,i}^{-1}x_{0,j}+{\sf noise}$. A large number of practical
problems can be modeled in this framework. For instance, the case
$\Group = SO(3)$ (the orthogonal group in three dimensions) is
relevant for camera registration, and molecule structure
reconstruction in electron microscopy \cite{singer2011three}.

The $\Z_2$-synchronization problem is arguably the simplest problem in
this class, and corresponds to $\Group = \Z_2$ (the group of integers
modulo $2$). Without loss of generality, we will identify this with
the group $(\{+1,-1\}, \,\cdot\,)$ (elements of the group are $+1$,
$-1$ and the group operation is ordinary multiplication).  We assume
observations to be distorted by Gaussian noise, namely for each $i<j$ we
observe $Y_{ij} = (\lambda/n)\, x_{0,i}x_{0,j} + W_{ij}$, where
$W_{ij}\sim\normal(0,1/n)$ are independent standard normal random
variables. This fits the general definition since $x_{0,i}^{-1} = x_{0,i}$ for
$x_{0,i}\in \{+1,-1\}$.

In matrix notation, we observe a symmetric matrix $\bY=\bY^*\in
\reals^{n\times n}$ given by
\begin{align}
\bY = \frac{\lambda}{n}\, \bxz\, \bxz^{*} + \bW\, . \label{eq:MainModel}
\end{align} 
(Note that entries on the diagonal carry no information.)
Here $\bxz\in\{+1,-1\}^n$, $\bxz^*$ denote the transpose of $\bxz$, and
$\bW = (W_{ij})_{i,j\le n}$ is a 
random matrix from the Gaussian Orthogonal Ensemble (GOE), i.e. a
symmetric matrix with independent entries (up to symmetry)
$(W_{ij})_{1\le i<j\le n}\sim_{i.i.d.}\normal(0,1/n)$ and
$(W_{ii})_{1\le i\le n}\sim_{i.i.d.}\normal(0,2/n)$.

A solution of the $\Z_2$ synchronization problem can be interpreted 
as a bi-partition of the set $\{1,\dots,n\}$, and hence this has been
used as a model for partitioning signed networks
\cite{cucuringu2015synchronization,abbe2014decoding}.

\vspace{0.25cm}

\noindent {\bf $U(1)$-synchronization.} 
This is again an instance of the synchronization problem. However,
 we take $\Group = U(1)$.
This is the group of complex number of modulus one, with the operation
of complex multiplication $\Group = (\{x\in\complex:\, |x|=1\},\,\cdot\,)$

As in the previous case, we assume
observations to be distorted by Gaussian noise, i.e. for each $i<j$ we
observe $Y_{ij} = (\lambda/n)\, x_{0,i}\overline{x}_{0,j} + W_{ij}$, where
$\overline{z}$ denotes complex conjugation\footnote{Here and below
  $\cnormal(\mu,\sigma^2)$, with $\mu = \mu_1+i\, \mu_2$  and
  $\sigma^2\in\reals_{\ge 0}$ denotes the
  complex normal distribution. Namely, $X\sim\cnormal(\mu,\sigma^2)$
  if $X = X_1+i\, X_2$ with $X_1\sim\normal(\mu_1,\sigma^2/2)$,
  $X_2\sim\normal(\mu_2,\sigma^2/2)$ independent Gaussian random variables.} 
and $W_{ij}\sim \cnormal(0,1/n)$.

In matrix notations, this model takes the same form
\eqref{eq:MainModel}, provided we interpret $\bxz^*$ as the conjugate
transpose of vector $\bxz\in \complex^{n}$, with components $x_{0,i}$,
$|x_{0,i}|=1$. We will follows this convention throughout.

$U(1)$ synchronization has been used as a model for clocks
synchronization over networks
\cite{singer2011angular,bandeira2014tightness}.
It is also closely related to the phase-retrieval problem in signal
processing \cite{candes2015phase,waldspurger2015phase,alexeev2014phase}.
An important qualitative difference with respect to the previous
example ($\Z_2$ synchronization) lies in the fact that $U(1)$ is a
continuous group. We regard this as a prototype of
synchronization problems over compact Lie groups (e.g. $SO(3)$)

\vspace{0.25cm}
\noindent {\bf Hidden partition (a.k.a. community detection).}
The hidden (or planted) partition model is a statistical model for 
the problem of finding clusters in large network datasets (see \cite{krzakala2013spectral,massoulie2014community,mossel2013proof}
and references therein for earlier work). 
The data consist of graph $G = (V, E)$ over vertex set $V=
[n]\equiv\{1,2,\dots,n\}$ generated as follows. We partition $V =
V_+\cup V_-$ by setting $i\in V_+$  or $i\in V_-$ independently across
vertices with $\prob(i\in V_+) = \prob(i\in V_-)=1/2$. Conditional on
the partition, edges are independent with 
\begin{equation}
\prob\big\{(i,j)\in E\,\big|\, V_+,V_-\big\} = 
\begin{cases}
a/n & \mbox{ if $\{i,j\}\subseteq V_+$ or $\{i,j\}\subseteq V_-$,}\\
b/n & \mbox{ otherwise.}\\
\end{cases}\label{eq:SBMDefinition}
\end{equation}
\begin{figure}[t!]
\centerline{\includegraphics[width=.55\textwidth]{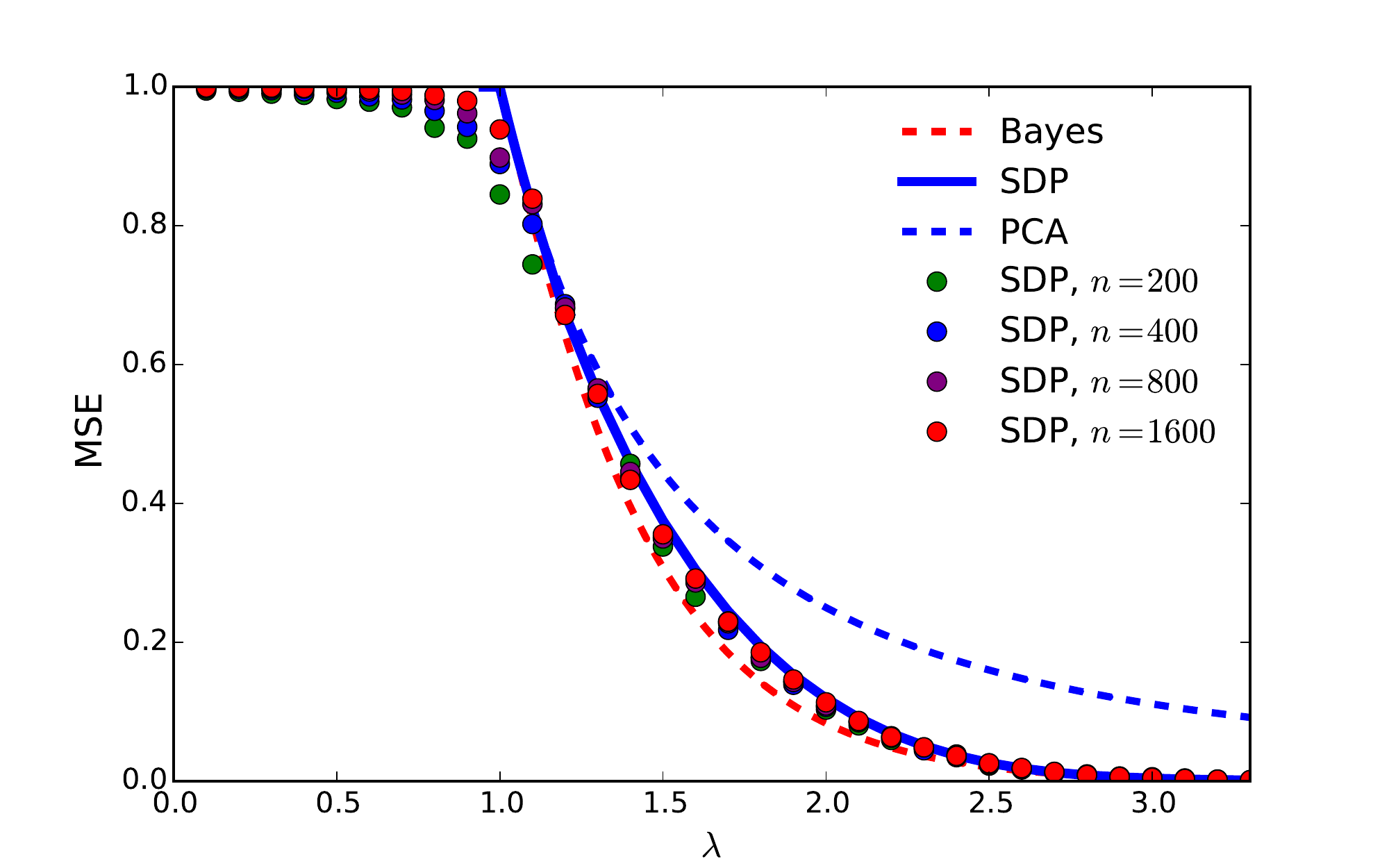}}
\caption{Estimating $\bxz\in\{+1,-1\}^n$ under the noisy $\Z_2$
synchronization model of \eqref{eq:MainModel}. Curves correspond to
(asymptotic) analytical predictions, and dots to numerical
simulations (averaged over $100$ realizations). \label{fig:sdpZ2}}
\end{figure}
Here $a>b>0$ are model parameters that will be kept of order one as 
$n\to\infty$. This corresponds to a random graph with bounded average
degree $d=(a+b)/2$, and a cluster (a.k.a. `block' or `community') structure
corresponding to the partition $V_+\cup V_-$. Given a realization of
such a graph, we are interested in estimating the underlying
partition.

We can encode the partition $V_+$, $V_-$ by a vector $\bxz\in \{+1,-1\}^n$,
letting $x_{0,i} = +1$ if $i\in V_+$ and $x_{0,i} = -1$ if $i\in
V_-$. An  important insight --that we will develop below-- is that this problem is analogous to 
$\Z_2$-synchronization, with signal strength $\lambda =
(a-b)/\sqrt{2(a+b)}$.  The parameters' correspondence is obtained, at
a heuristics level, by noting that, if $\bA_G$ is the adjacency matrix
of $G$, then $\E\<\bxz,\bA_G\bxz\>/(n\E\|\bA_G\|^2_F)^{1/2}\approx (a-b)/\sqrt{2(a+b)}$.
(Here and below $\<\ba,\bb\> = \sum_ia_ib_i$ denotes the standard scalar product
between vectors.)

A generalization of this problem to the case of more than two blocks
has been studied since the eighties as a model for social network structure
\cite{holland},  under the name of `stochastic block
model.'  For the sake of simplicity, we will focus here on the
two-blocks case.
 
\section{Illustrations}

As a first preview of our results, Figure \ref{fig:sdpZ2} reports our
analytical predictions for the estimation error in the $\Z_2$
synchronization problem, comparing them with numerical simulations
using SDP. 
An estimator is a map $\hbx:\reals^{n\times
  n}\to \reals^n$, $\bY\mapsto \hbx(\bY)$. We compare various
estimators in terms of their per-coordinate mean square error:
\begin{align}
\MSE_n(\hbx)\equiv\frac{1}{n}\E\big\{\min_{s\in\{+1,-1\}}\|\hbx(\bY)-s\,
  \bxz\big\|_2^2\big\}\, ,
\end{align}
where expectation is with respect to the noise model
\eqref{eq:MainModel} and $\bxz\in\{+1,-1\}^n$ uniformly random. 
Note the minimization with respect to the sign $s\in \{+1,-1\}$ inside
the  expectation: because of the symmetry of 
\eqref{eq:MainModel},
the vector $\bxz$ can  only be estimated  up to a global sign.
We will be interested in the high-dimensional limit $n\to\infty$
and will omit the subscript $n$ --thus writing $\MSE(\hbx)$-- to
denote  this limit.
Note that trivial estimator that always returns $0$ has error
 $\MSE_n({\mathbf{0}}) =1$, and hence for every other method
we should achieve $\MSE(\hbx) \in [0,1]$,

Classical statistical theory suggests two natural reference
estimators: the Bayes-optimal  and the maximum likelihood estimators.
We will discuss these methods first, in order to set the stage for SDP relaxations.

\vspace{0.2cm}

\noindent{\bf  Bayes-optimal estimator} (a.k.a. minimum MSE). This  provides a
lower bound on the performance of any other approach.
It takes the conditional expectation of the unknown signal given the observations:
\begin{align}
\xBayes(\bY) = \E\big\{\, \bx\, \big|\, (\lambda/n)\bx\bx^*+\bW=
  \bY\big\}\, .
\end{align}
Explicit formulas are given in Supplementary Information (SI). We note
that $\xBayes(\bY)$ assumes knowledge of the prior distribution.
The red-dashed curve in Fig. \ref{fig:sdpZ2} presents our analytical
prediction for the  asymptotic
MSE for $\xBayes(\,\cdot\,)$. Notice that $\MSE(\xBayes)
= 1$ for all $\lambda\le  1$ and $\MSE(\xBayes)
< 1$ strictly for all $\lambda> 1$, with  $\MSE(\xBayes)\to 0$ quickly
as $\lambda\to\infty$. The point $\lBayes=1$ corresponds to a phase
transition for optimal estimation, and no method can have non-trivial $\MSE$ for
$\lambda\le \lBayes$. 

\vspace{0.2cm}

\noindent{\bf Maximum likelihood (MLE).} The estimator $\xml(Y)$ is given by
the solution of
\begin{align}
\xml(\bY) = c(\lambda)\, \argmax_{\bx\in \{+1,-1\}^n}\; 
\<\bx,\bY\bx\>\;\, . \label{eq:ML}
\end{align}
Here 
$c(\lambda)$ is a scaling factor\footnote{In practical
applications, $\lambda$ might not be known. We are not
concerned by this at the moment, since maximum likelihood is used as a
idealized benchmark here.

Note that, strictly speaking, this is a `scaled' maximum likelihood
estimator. We prefer to scale $\xml(\bY)$ in order to keep
$\MSE(\xml)\in [0,1]$.} 
that is chosen according to
the asymptotic theory as to minimize the MSE.
As for the Bayes-optimal curve, we obtain $\MSE(\xml) = 1$ for
$\lambda\le \lml = 1$ and $\MSE(\xml) < 1$ (and rapidly decaying to 0) for
$\lambda> \lml$ .
(We refer to the SI for this result.)

\vspace{0.2cm}

\noindent{\bf  SDP.} Neither the Bayes, nor the maximum likelihood approaches can be
implemented efficiently. In particular, solving the combinatorial
optimization problem in Eq. \eqref{eq:ML} is a prototypical NP-complete
problem.  Even worse, approximating the optimum value within 
a sub-logarithmic factor is computationally hard \cite{arora2005non} (from a worst case perspective). 
SDP relaxations allow to obtain tractable
approximations.
Specifically --and following a standard `lifting' idea--  we replace the problem \eqref{eq:ML} by
the following semidefinite program over the symmetric matrix
$\bX\in\reals^{n\times n}$: 
\begin{align}
\mbox{maximize} & \;\;\;\<\bX,\bY\>\, ,\label{eq:SDP}\\
\mbox{subject to} & \;\;\; \bX\succeq 0\, ,\; \;\;\; X_{ii}=1 \;\; \forall i\in [n]\, .\nonumber
\end{align}
We use $\<\, \cdot\,,\cdot\,\>$ to denote the scalar product between
matrices,
namely $\<\bA,\bB\> \equiv \Tr(\bA^*\bB)$, and $\bA\succeq 0$ to
indicate that $\bA$ is positive semidefinite\footnote{Recall that a
  symmetric matrix $\bA$ is said to be PSD if all of its eigenvalues
  are non-negative.} (PSD). If we assume $\bX = \bx\bx^*$, the SDP
\eqref{eq:SDP} reduces to the maximum-likelihood problem \eqref{eq:ML}.
By dropping this condition, we obtain a convex optimization problem
that is solvable in polynomial time. Given an optimizer $\bX_\sopt= \bX_\sopt(\bY)$ of
this convex problem, we need to produce a vector estimate. We follow a
different strategy from standard `rounding' methods in computer
science, which is motivated by our analysis below.
We compute the eigenvalue
decomposition
$\bX_\sopt = \sum_{i=1}^n\xi_i\, \bv_i\bv_i^{*}$,
with eigenvalues $\xi_1\ge \xi_2\ge \dots\ge \xi_n\ge
0$, and eigenvectors $\bv_i = \bv_i(\bX_\sopt(\bY))$, with  $\|\bv_i\|_2 = 1$. We
then return the estimate
\begin{align}
\xsdp(\bY)= \sqrt{n}\, \csdp(\lambda)\, \bv_1(\bX_\sopt(\bY))\, ,
\end{align}
with $\csdp(\lambda)$ a certain scaling factor, see SI.

Our analytical prediction for $\MSE(\xsdp)$ is plotted as
 blue solid line in Fig.~\ref{fig:sdpZ2}. Dots report the results of numerical simulations
with this relaxation for increasing problem dimensions. The asymptotic  theory appears to
capture very well these data already for $n = 200$. 
For further comparison, alongside the above estimators, we report the  asymptotic prediction for
$\MSE(\xpca)$, the  mean square error of principal component 
analysis. This method simply returns the principal eigenvector of
$\bY$, suitably rescaled (see SI).

Figure \ref{fig:sdpZ2} reveals several interesting features:
\begin{enumerate}
\item\emph{Phase transition for optimal estimation.}
Bayes-optimal estimation achieves non-trivial accuracy as soon as $
\lambda>\lBayes =1$. The same is achieved by a
method as simple as PCA (blue-dashed curve). On the other hand, for $\lambda<1$
no method can achieve a mean square error that is asymptotically smaller than one 
(the latter can be achieved trivially by returning $\hbx = 0$.)
\item\emph{Suboptimality of PCA at large signal strength.}
PCA can be implemented efficiently, but does not exploit the
information $x_{0,i}\in\{+1,-1\}$. As a consequence, its estimation
error is significantly sub-optimal at large $\lambda$ (see SI).
\item\emph{Near-optimality of  SDP relaxations.}
The tractable estimator $\xsdp(\bY)$ achieves the best of both
worlds. Its phase transition coincides with the Bayes-optimal one
$\lBayes=1$, and $\MSE(\xsdp)$ decays exponentially at
large $\lambda$, staying close to $\MSE(\xBayes)$ and strictly smaller
than $\MSE(\xpca)$, for $\lambda \ge1$.
\end{enumerate}

We believe that the above features are generic: as shown in the
SI,  $U(1)$ synchronization confirms this expectation.

\begin{figure}[t!]
\centerline{\includegraphics[width=.65\textwidth]{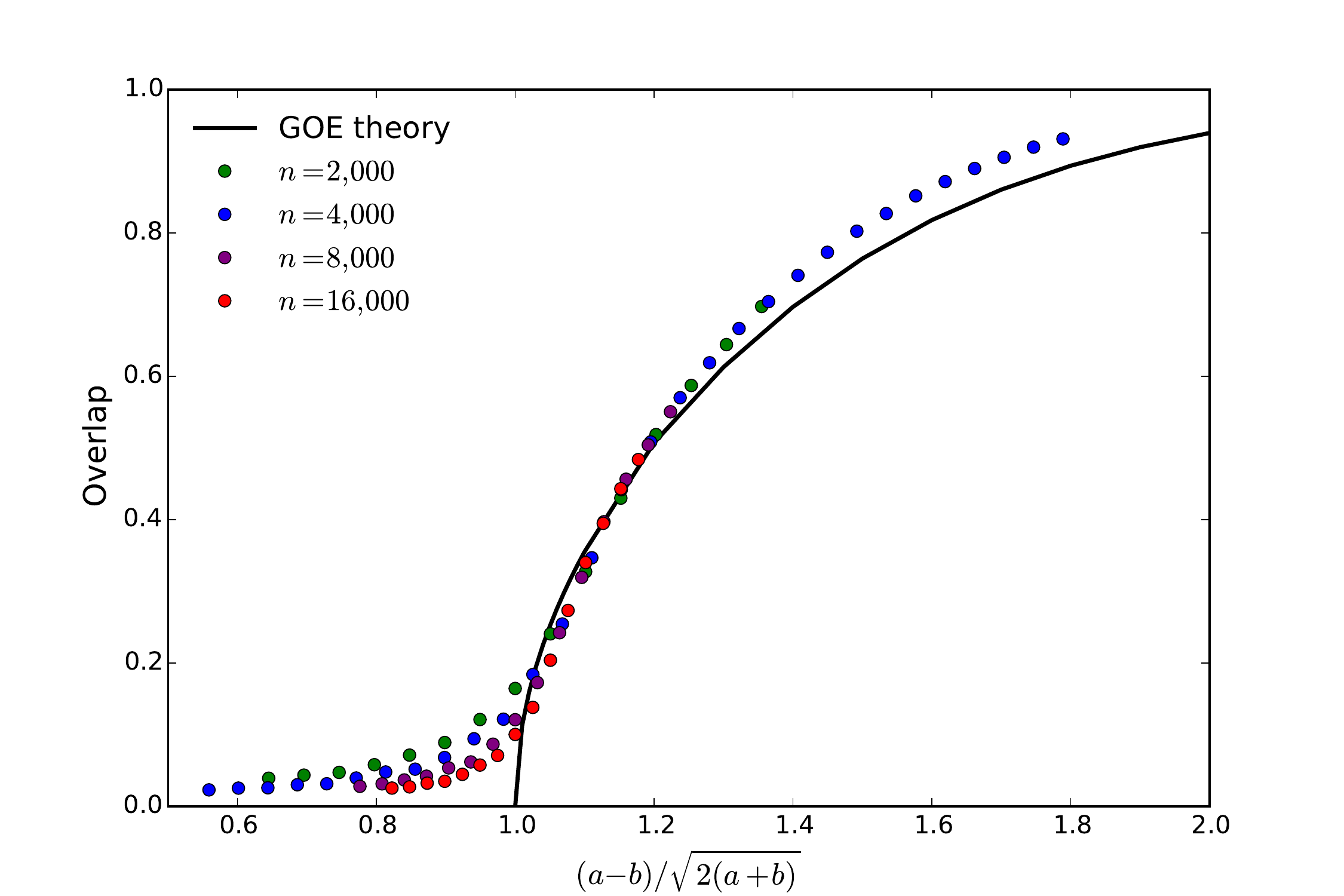}}
\caption{Community detection under the hidden partition model of
  Eq. \eqref{eq:SBMDefinition},
for average degree $(a+b)/2 = 5$. Dots corresponds to the
performance of the SDP reconstruction method (averaged over $500$
realizations). The continuous curve is the asymptotic analytical
prediction for the Gaussian model (which captures the large-degree behavior).
\label{fig:sdpSBM}}
\end{figure}

Figures \ref{fig:sdpSBM} illustrates our results for the community
detection problem under the hidden partition model of
Eq. \eqref{eq:SBMDefinition}.
Recall that we encode the ground truth by a vector $\bxz\in \{+1,-1\}^n$.
In the present context, an  estimator is  required to return a
partition of the vertices of the graph.
Formally, it is a function on the space of graphs with $n$ vertices
$\cG_n$, namely $\hbx:\cG_n\to \{+1,-1\}^n$, $G\mapsto\hbx(G)$. 
We will measure the performances of such an estimator through the overlap:
\begin{align}
\Ove_n(\hbx)= \frac{1}{n}\E\big\{\big|\<\hbx(G),\bxz\>\big|\big\}\, .
\end{align}
and its asymptotic $n\to\infty$ limit (for which we omit the subscript).
In order to motivate the SDP relaxation we note that the maximum
likelihood estimator partitions $V$ in two sets of equal size as to
minimize the number of edges across the partition (the minimum
bisection problem). Formally
\begin{align}
\xml(G) \equiv \arg\!\!\max_{\bx \in \{+1,-1\}^n}
\Big\{ \sum_{(i,j)\in E}x_ix_j:\;\;\; \<\bx,\bone\>=0\Big\}\, ,
\end{align}
where $\bone = (1,1,\dots, 1)$ is the all-ones vector. 
Once more, this problem is hard to approximate \cite{khot2006ruling}, which motivates the
following SDP relaxation:
\begin{align}
\mbox{maximize} & \;\;\;\sum_{(i,j)\in E}X_{ij}\, ,\label{eq:SDPpart}\\
\mbox{subject to} & \;\;\; \bX\succeq 0\, ,
\; \;\;\; \bX\bone = \bzero\, ,\;\;\;\; X_{ii}=1 \;\; \forall i\in [n]\, .\nonumber
\end{align}
Given an optimizer $\bX_\sopt= \bX_\sopt(G)$, we extract a partition of the
vertices $V$ as follows. As for the $\Z_2$ synchronization problem, we
compute the principal eigenvector $\bv_1(\bX_\sopt)$. We then partition
$V$ according to the sign of $\bv_1(\bX_{\sopt})$. Formally
\begin{align}
\xsdp(G) = \sign\big(\bv_1(\bX_\sopt(G))\big)\, .
\end{align}
Let us emphasize a few features of Figure \ref{fig:sdpSBM}:
\begin{enumerate}
\item\emph{Accuracy of GOE theory.} The continuous curve of Figure \ref{fig:sdpSBM} reports
  the analytical prediction within the $\Z_2$ synchronization model,
  with Gaussian noise (the GOE theory). This can be shown to capture
  the large degree limit: $d=(a+b)/2\to\infty$, with
  $\lambda=(a-b)/\sqrt{2(a+b)}$ fixed. However, it 
  describes well the empirical
  results for sparse graphs of average degree as small as $d=5$.
\item\emph{Superiority of SDP to PCA.} A sequence of recent papers
 (see \cite{krzakala2013spectral} and references therein) demonstrate that classical spectral methods --such as PCA--
  fail to detect the hidden partition in graphs with bounded average
  degree. 
In contrast,  Figure \ref{fig:sdpSBM} shows that a standard SDP
relaxation does not break down in the sparse regime. See \cite{montanari2015semidefinite}
for rigorous evidence towards the same conclusion.
\item\emph{Near optimality of SDP.} As proven in \cite{mossel2012stochastic}, no
  estimator can achieve $\Ove_n(\hbx)\ge \delta>0$ as $n\to \infty$,
  if $\lambda=(a-b)/\sqrt{2(a+b)}<1$. 
\end{enumerate}
\begin{figure}[t!]
\hspace{3cm}\includegraphics[width=.65\textwidth]{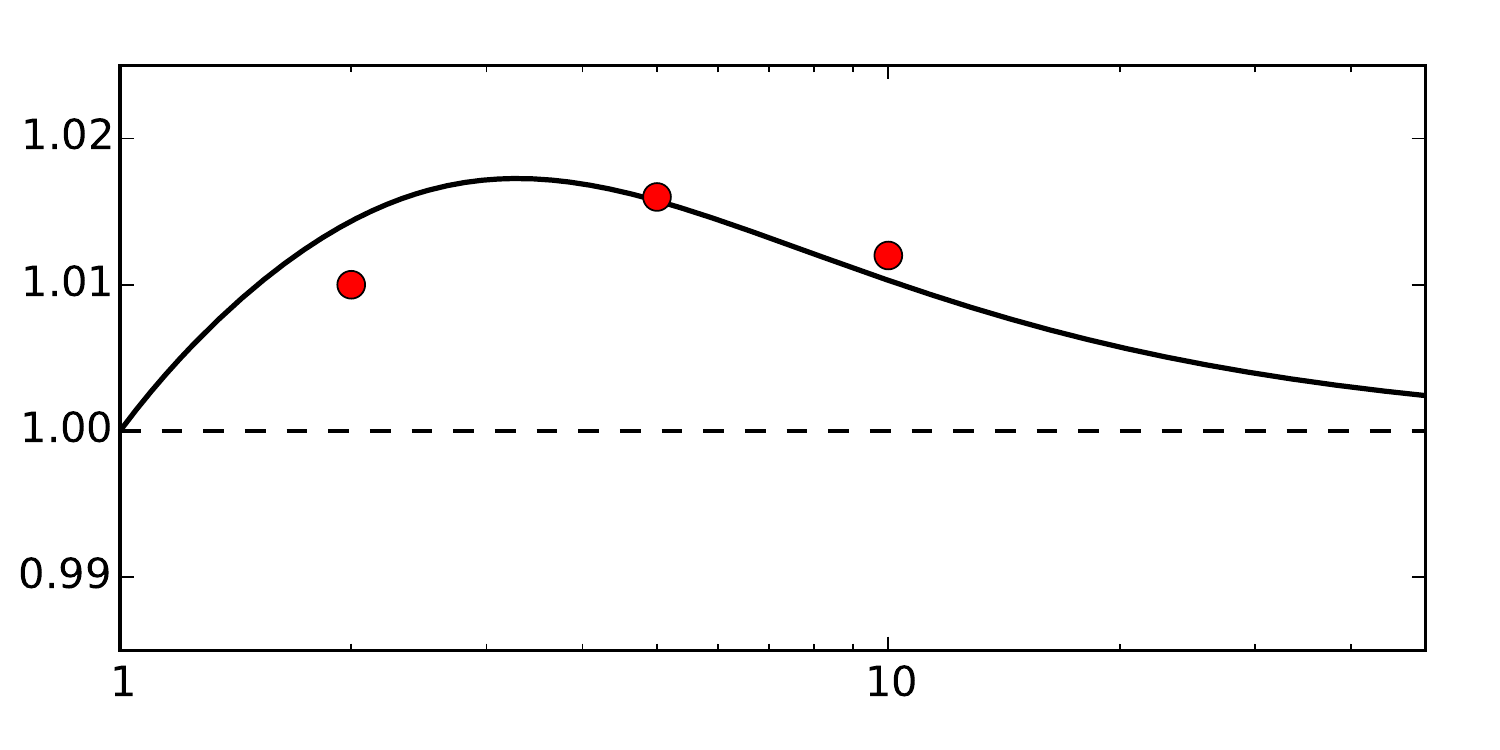}
\put(-307,55){{\small $\lsdp$}}
\put(-150,0){{\small $d$}}
\caption{Continuous line:  prediction for the critical point of the SDP
  estimator (in terms of the rescaled parameter
  $\lambda=(a-b)/\sqrt{2(a+b)}$). Dashed line: ideal phase transition
  $\lambda=1$. Red circles: numerical estimates of the phase
  transition location for $d=2$, $5$, $10$. \label{fig:sdp_sparse}}
\end{figure}
Figure \ref{fig:sdpSBM} (and the theory developed in the next section) 
suggests that SDP has a phase transition threshold. Namely, there exists
$\lsdp = \lsdp(d)$ such that, if
\begin{align}
\lambda = \frac{a-b}{\sqrt{2(a+b)}}\ge \lsdp\big(d = (a+b)/2\big)\, . \label{eq:lsdp}
\end{align}
then SDP achieves overlap bounded away from zero: $\Ove(\xsdp)>0$. 
The figure also suggests $\lsdp(5) \approx \lBayes = 1$, i.e. SDP is nearly
optimal.

Below we will derive an accurate approximation  for the critical
point $\lsdp(d)$. The factor $\lsdp(d)$ measures the sub-optimality of SDP for
graphs of average degree $d$. 

Figure \ref{fig:sdp_sparse} plots our prediction for the function 
$\lsdp(d)$, together with empirically determined values for 
this threshold, obtained through Monte Carlo experiments for $d\in
\{2,5,10\}$ (red circles). These were obtained by running the SDP
estimator on randomly generated graphs with size up to
$n=64,000$ (total CPU time was about $10$ years). 
In particular, we obtain $\lsdp(d)>1$ strictly, but
the gap $\lsdp(d)-1$  is very small (at most of the order of $2\%$) for
all $d$. This confirms in a precise quantitative way the conclusion
that SDP is nearly optimal for the hidden partition problem.

Simulations results are in broad agreement with our predictions,
but present small discrepancies (below $0.5\%$). These discrepancies
might be due to the extrapolation form finite-$n$ simulations to
$n\to\infty$, or to the inaccuracy of our analytical calculation.

%
%************************************************************
%
\section{Analytical results}

Our analysis is based on a connection with statistical mechanics. 
The models arising from this connection are spin
models in the so-called `large-$N$' limit, a topic of intense study
across statistical mechanics and quantum field theory \cite{brezin1993large}.
Here we exploit this connection to apply non-rigorous but sophisticated tools from 
the theory of mean field spin glasses \cite{MezardMontanari,SpinGlass}. 
The paper \cite{montanari2015semidefinite} provides partial rigorous
evidence towards the predictions developed here. 

We will first  focus on the  simpler problem of synchronization
under Gaussian noise, treating together the $\Z_2$ and $U(1)$ 
case. We will then discuss the new features arising within the sparse
hidden partition problem. Most technical derivations are presented in
the SI. 
In order to treat the real ($\Z_2$) and complex ($U(1)$) cases
jointly, we will use $\field$ to denote any of the fields of reals or
complex numbers, i.e. either $\reals$ or $\complex$.

\subsection{Gibbs measures, vector spin models}

We start by recalling that a matrix $\bX\in \field^{n\times n}$ is
PSD if and only if it can be written as $\bX = \bsigma\bsigma^*$
for some $\bsigma\in \field^{n\times m}$.
Indeed, without loss of generality, one can take $m=n$,  and  any
$m\ge n$ is equivalent.

 Letting $\bsigma_1$, \dots
$\bsigma_n\in\field^m$ be the rows of $\bsigma$, the SDP \eqref{eq:SDP}
can be rewritten as
\begin{align}
\mbox{maximize} & \;\;\;\sum_{(i,j)}Y_{ij}\<\bsigma_i,\bsigma_j\>\, ,\label{eq:NonConvex}\\
\mbox{subject to} & \;\;\; \bsigma_i\in S^{m-1}\;\; \forall i\in [n]\, ,\nonumber
\end{align}
with $S^{m-1} = \{\bz\in \field^m:\; \|\bz\|_2 = 1\}$ the unit sphere
in $m$ dimensions. The SDP relaxation corresponds to any case $m\ge n$ or
--following the physics parlance-- $m=\infty$. Note however that cases
with bounded (small) $m$ are of independent interest. In particular,
for $m=1$ we have $\bsigma_i\in\{-1,+1\}$ (for the real case) or
$\bsigma_i\in U(1)\subset \complex$ (for the complex case). Hence we
recover the maximum-likelihood estimator setting $m=1$.
It is also known that, for $m> \sqrt{2n}$, the problem
\eqref{eq:NonConvex} has no local optima except the global ones \cite{burer2003nonlinear}.

A crucial question is how the solution of \eqref{eq:NonConvex}
depends on the spin dimensionality $m$, for $m\ll n$. Denote by $\OPT(\bY;m)$ the
optimum value when the dimension is $m$ (in particular $\OPT(\bY;m)$
is also the value of \eqref{eq:SDP} for $m\ge n$). It was proven in
\cite{montanari2015semidefinite} that there exists a constant $C$ independent of $m$ and $n$
such that
\begin{align}
\Big(1-\frac{C}{m}\Big)\, \OPT(\bY;\infty)\le \OPT(\bY;m)\le
  \OPT(\bY;\infty)\, ,\label{eq:Gro}
\end{align}
with probability converging to one as $n\to\infty$ (whereby $\bY$ is chosen with any
of the distributions studied in the present paper).
The upper bound in Eq.~\eqref{eq:Gro} follows immediately from the
definition. The lower bound  is a generalization of the celebrated
Grothendieck inequality from functional analysis \cite{khot2012grothendieck}.

The above inequalities imply that
we can obtain  information about the SDP \eqref{eq:SDP} in the
$n\to\infty$ limit, by taking $m\to\infty$ \emph{after} $n\to\infty$.
This is the asymptotic regime usually studied in physics under the term
`large-$N$ limit.'

Finally, we can associate to the problem \eqref{eq:NonConvex} 
a finite-temperature Gibbs measure as follows:
\begin{align}
p_{\beta,m}(\de \bsigma) = \frac{1}{Z}\exp\Big\{2m\beta \sum_{i<j} \Re(Y_{ij}
  \<\bsigma_i,\bsigma_j\>)\Big\}\, \prod_{i=1}^np_0(\de \bsigma_i)\, . \label{eq:GibbsModel}
\end{align}
where $p_0(\de \bsigma_i)$ is the uniform measure over the
$m$-dimensional sphere $S^{m-1}$, and $\Re(z)$ denotes the real part of $z$.
This allows to treat in a unified framework all of the estimators
introduced above. 
The optimization problem \eqref{eq:NonConvex}  is recovered by taking
the limit $\beta\to\infty$ (with maximum likelihood for $m=1$ and SDP 
for $m\to\infty$). The Bayes-optimal estimator is recovered by setting
$m=1$
and $\beta=\lambda/2$ (in the real case) or $\beta = \lambda$ (in the
complex case).

\subsection{Cavity method: $\Z_2$ and $U(1)$ synchronization}

The cavity method from spin-glass theory can be used to
analyze the asymptotic structure of the Gibbs measure
\eqref{eq:GibbsModel}
as $n\to\infty$. Below we will state the
predictions of our approach for the SDP estimator $\xsdp$.

Here we list the main steps of our analysis for the expert reader,
deferring a complete derivation to the SI: 
\begin{itemize}
\item[$(i)$] We use the cavity
method to derive the `replica symmetric' predictions for the model
\eqref{eq:GibbsModel} in the limit $n\to\infty$. 
\item[$(ii)$] By setting $m=1$,
$\beta=\lambda/2$ (in the real case) 
or $\beta = \lambda$ (in the
complex case) we obtain the Bayes-optimal error $\MSE(\xBayes)$: on
the basis of \cite{deshpande2015asymptotic}, we expect the replica symmetric assumption to
hold, and these predictions to be exact. (See also
\cite{lesieur2015mmse} for related work.)
\item[$(iii)$] By setting $m=1$ and
$\beta\to\infty$ we obtain a prediction for the error of maximum
likelihood estimation $\MSE(\xml)$. While this prediction is not
expected to be exact (because of replica symmetry breaking),
it should be nevertheless rather accurate, especially for large $\lambda$.
\item[$(iv)$] By setting $m\to\infty$ and $\beta\to\infty$, we obtain the SDP
estimation error $\MSE(\xsdp)$, which is our main object of interest. Notice that the inversion of limits
$m\to\infty$ and $n\to\infty$ is justified (at the level of objective
value) by Grothendieck inequality. Further, since the $m=\infty$ case
is equivalent to a convex program, we expect the replica symmetric
prediction to be exact in this case.
\end{itemize}

The properties of the SDP estimator are given in terms of the solution
of a set 3 non-linear equations for the 3 scalar parameters 
$\mu$, $q$, $b\in\reals$, that we state next. Let $Z\sim \normal(0,1)$
(in the real case) or $Z\sim \cnormal(0,1)$ (in the complex case).
Define $\rho = \rho(Z;\mu,q,r)$ as the only non-negative
solution of the following equation in $(0,\infty)$:
\begin{align}
1=\frac{|\mu+\sqrt{q}Z|^2}{(\rho+r)^2} +
  \frac{1-q}{\rho^2}\, .
\end{align}
Then $\mu$, $q$, $r$ satisfy
\begin{align}
\mu & = \lambda \E\left\{\frac{\mu+\sqrt{q}\, \Re(Z)}{\rho+r}\right\}\, ,\;\;\;\;\;\;\;
q  = \E\left\{\frac{|\mu+\sqrt{q}\,
    Z|^2}{(\rho+r)^2}\right\}\, ,\label{eq:SDP_FP1}\\
r & =\E\Big\{\frac{1}{\rho}-
      \frac{\mu}{\sqrt{q}}\,\frac{\Re(Z)}{\rho+r}
      -\frac{|Z|^2}{\rho+r}\Big\}\, . \label{eq:SDP_FP3}
\end{align}
These equations can be solved by iteration, after approximating the
expectations  on the right-hand side numerically.
The properties of the SDP estimator can be derived from this
solution. Concretely, we have
\begin{align}
\MSE(\xsdp) = 1-\frac{\mu(\lambda)^2}{\lambda^2 q(\lambda)}\, .
\end{align}
The corresponding curve is reported in Figure \ref{fig:sdpZ2} for the
real case $\Group = \Z_2$. We can also obtain the asymptotic overlap
from the solution of these equations. The cavity prediction is
\begin{align}
\Ove(\xsdp)= 1-2\Phi\Big(-\frac{\mu(\lambda)}{\sqrt{q(\lambda)}}\Big)\, .
\end{align}
The corresponding curve is plotted in Figure \ref{fig:sdpSBM}.

More generally, 
for any dimension $m$, and inverse temperature $\beta$, we obtain
equations that are analogous to Eqs.~\eqref{eq:SDP_FP1}-\eqref{eq:SDP_FP3}.
The parameters $\mu, q, b$ characterize the asymptotic
structure of the probability measure $p_{\beta,m}(\de\bsigma)$ defined
in Eq.~\eqref{eq:GibbsModel}, as follows. We assume, for simplicity
$\bxz = (+1,\dots,+1)$. 
Define the following probability measure on 
unit sphere $S^{m-1}$, parametrized by $\bxi\in\reals^m$, $r\in\reals$
\begin{align}
\nu_{\bxi,r} (\de\bsigma)=
\frac{1}{z(\bxi,r)}\,
\exp\Big\{2\beta\, m\Re\<\bxi,\bsigma\> - \beta\,
  m\, r |\sigma_1|^2 \Big\}\, p_0(\de\bsigma)\, ,\label{eq:Generalnu}
\end{align}
For $\nu$ a probability measure on $S^{m-1}$ and $R$ an orthogonal (or
unitary) matrix, let $\nu^R$ be the measure obtained
by\footnote{Formally, $\nu^R(\bsigma\in A) \equiv \nu(R^{-1}\bsigma
  \in A)$.} `rotating' $\nu$. Finally, let $p^{(m,\beta)}_{i(1),\dots,i(k)}$ denote
the joint distribution of $\bsigma_{i(1)},\cdots,\bsigma_{i(k)}$ under $p_{m,\beta}$.
Then, for any fixed $k$, and any sequence of $k$-uples $(i(1),\dots,i(k))_n\in [n]$, we
have
\begin{align}
p^{(m,\beta)}_{i(1),\dots,i(k)}\Rightarrow 
\int \nu^R_{\bxi_1,r} (\,\cdot\,)\times  \cdots\times \nu^R_{\bxi_k,r} (\,\cdot\,)\,
  \de R
\end{align}
Here $\de R$ denotes the uniform (Haar) measure on the orthogonal
group, $\Rightarrow$ denotes convergence in distribution (note that
$p^{(m,\beta)}_{i(1),\dots,i(k)}$ is a random variable), and 
$\bxi_1,\dots,\bxi_k\sim_{iid} \normal(\mu\, \bfe_1,\bQ)$ with 
$\bQ = \diag(q,q_0\dots,q_0)$, $q_0 = (1-q)/(m-1)$.

\subsection{Cavity method: Community detection in sparse graphs}

We next consider the hidden partition model, defined
by Eq.~\eqref{eq:SBMDefinition}. As above, we denote by $d=(a+b)/2$ the
asymptotic average degree of the graph $G$,
and by $\lambda = (a-b)/\sqrt{2(a+b)}$ the `signal-to-noise' ratio.
As illustrated by  Figure \ref{fig:sdpSBM} (and further simulations presented in SI),
$\Z_2$ synchronization appears to be a very accurate approximation for
the hidden partition model already
at moderate $d$. 

The main change with respect to the dense case is that the
phase transition at $\lambda = 1$, is slightly shifted, as
per Eq.~\eqref{eq:lsdp}. 
Namely, SDP can detect the hidden partition with high probability 
if and only if $\lambda\ge \lsdp(d)$, for some $\lsdp(d)>1$.

Our prediction for the curve $\lsdp(d)$ will be denoted by $\tlsdp(d)$
and is plotted in Figure \ref{fig:sdp_sparse}. It is obtained by
finding an approximate solution of the RS cavity equations.
We see that $\tlsdp(d)$  approaches very quickly the ideal value
$\lambda=1$ for $d\to\infty$. Indeed our prediction implies
\begin{align}
\tlsdp(d) = 1+\frac{1}{8d} + O(d^{-2})\, .
\end{align}
Also, $\tlsdp(d)\to 1$ as $d\to 1$. This is to be expected because the
constraints $a\ge b\ge 0$ imply $(a-b)/2\le d$, with $b=0$ at $(a-b)/2
=d$.
Hence the community detection problem becomes trivial at $(a-b)/2=d$:
it is sufficient to identify the connected components in $G$. This
implies the bound $\lsdp(d)\le \sqrt{d}$.

More interestingly, $\tlsdp(d)$ admits a characterization in terms of 
a distributional recursion, that can be evaluated numerically, and is
plotted as a continuous line in Figure \ref{fig:sdp_sparse}. Surprisingly, \emph{the SDP
detection threshold appears to be sub-optimal at most by  $2\,\%$.}
In order to state this characterization, consider first the recursive
distributional equation (RDE)
\begin{align}
\cond \ed \sum_{i=1}^L\frac{\cond_i}{1+\cond_i}\, .\label{eq:firstRDE}
\end{align}
Here $\ed$ denotes equality in distribution, $L\sim \Poisson(d)$ and
$\cond_1,\dots,\cond_L$ are i.i.d.  copies of $\cond$. This has to be read as an
equation for the law of the random variable $\cond$ (see, e.g., \cite{aldous2005survey}
for further background on RDEs). 
We are interested in a specific solution of this equation, which can
be constructed as follows. Set $c^{0} = \infty$ almost surely, and,
for $\ell\ge 0$, let $\cond^{\ell+1} \ed
\sum_{i=1}^L\cond_i^{\ell}/(1+\cond^{\ell}_i)$. It is proved in \cite{lyons1997unsolved}
 that the resulting sequence of random variables converges in
distribution to a solution of Eq.~(\ref{eq:firstRDE}): $\cond^{\ell}
\stackrel{{\rm d}}{\Rightarrow} \cond_*$.

The quantity $\cond_*$ has a beautiful interpretation. Consider a
(rooted) Galton-Watson tree with offspring distribution $d$, and
imagine each edge to be a conductor with conductance equal to
one. Then
$\cond_*$ is the total conductance between the root, and the boundary
of the tree `at infinity.' In particular, $\cond_*=0$ almost surely
for $d\le 1$, and $\cond_*>0$ with positive probability if $d>1$ (see
\cite{lyons2013probability} and SI). 

Next consider the distributional recursion 
\begin{align}
\big(\cond^{\ell+1};h^{\ell+1}\big)\ed
  \Big(\sum_{i=1}^{L_++L_-}\frac{\cond^{\ell}_i}{1+\cond^{\ell}_i}
  ;\sum_{i=1}^{L_++L_-}\frac{s_ih^{\ell}_i}{\sqrt{1+\cond^{\ell}_i}}
\Big)\, ,\label{eq:Stability}
\end{align}
where $s_1,\dots, s_{L_+} = +1$,
$s_{L_++1},\dots, s_{L_++L_-} = -1$, and we use
 initialization $(\cond^{0},h^{0}) = (+\infty,1)$. This recursion determines
sequentially the distribution of $(\cond^{\ell+1},h^{\ell+1})$ from the
distribution of $(\cond^{\ell},h^{\ell})$. Here $L_+\sim \Poisson((d+\lambda)/2)$,
$L_-\sim\Poisson((d-\lambda)/2)$, and
$(\cond^{\ell}_1,h^{\ell}_1),\dots,(\cond^{\ell}_{L},h^{\ell}_{L})$ are i.i.d. copies
of $(\cond^{\ell},h^{\ell})$, independent of $L_+,L_-$. Notice that since
$L_++L_-\sim\Poisson(d)$, we have $\cond^{\ell}\stackrel{{\rm d}}{\Rightarrow} \cond_*$. The threshold
$\tlsdp(d)$ is defined as the smallest $\lambda$ such that the $h^t$
`diverges exponentially':
\begin{align}
\tlsdp(d) \equiv \inf\Big\{\lambda\in [0,\sqrt{d}] :\,
  \lim\inf_{t\to\infty}\frac{1}{t}\log \E(|h^{t}|^2) > 0\Big\}\, .
\end{align}
This value can be computed numerically, for instance by sampling the
recursion \eqref{eq:Stability}. The results of such an
evaluation  are plotted as a continuous line in Figure \ref{fig:sdp_sparse}.
%
%************************************************************
%
\section{Final algorithmic considerations}

We have shown that ideas from statistical mechanics can be used
to precisely locate phase transitions in SDP relaxations for
high-dimensional statistical problems. In the problems investigated
here, we find that SDP relaxations have optimal thresholds (in $\Z_2$
and $U(1)$ synchronization) or nearly-optimal thresholds (in community
detection under the hidden partition model). Here `near-optimality' is
to be interpreted in a precise quantitative sense: SDP's threshold is
sub-optimal --at most-- by a $2\,\%$ factor. 
As such SDPs provide a very useful tool for designing computationally
efficient
algorithms, that are also statistically efficient.

Let us emphasize that other polynomial-time algorithms can be used
for the specific problems studied here. In the synchronization problem, 
naive PCA achieves the optimal threshold $\lambda=1$.
In the community detection problem, several authors
recently developed ingenious spectral algorithms that achieve the
information theoretically optimal threshold $(a-b)/\sqrt{2(a+b)} = 1$,
see e.g.
\cite{decelle2011asymptotic,krzakala2013spectral,massoulie2014community,mossel2013proof,saade2014spectral}.

However, SDP relaxations have the important feature of being robust to
model  miss-specifications (see also \cite{moitra2015robust} for an
independent investigation of robustness issues). In order to illustrate this point, we
perturbed the hidden partition model as follows. For a perturbation
level $\alpha\in [0,1]$, we draw $n\alpha$ vertices
$i_1,\dots,i_{n\alpha}$ uniformly at
random in $G$. For each such vertex $i_{\ell}$ we connect by edges all
the neighbors of $i_{\ell}$. In our case, this results in adding
$O(nd^2\alpha)$ edges. 

In the SI, we compare the behavior of SDP and the Bethe Hessian
algorithm of \cite{saade2014spectral} for this perturbed model: while SDP
appears to be rather insensitive to the perturbation, the performance
of Bethe Hessian are severely degraded by it. We expect a similar
fragility to arise in other spectral algorithms. 

%
%*********
%
\section*{Acknowledgments}
A.J. and A.M. were partially supported by NSF grants CCF-1319979 and DMS-1106627 and the
AFOSR grant FA9550-13-1-0036.

\newpage

\bibliographystyle{amsalpha}

\newcommand{\etalchar}[1]{$^{#1}$}
\providecommand{\bysame}{\leavevmode\hbox to3em{\hrulefill}\thinspace}
\providecommand{\MR}{\relax\ifhmode\unskip\space\fi MR }
\providecommand{\MRhref}[2]{%
  \href{http://www.ams.org/mathscinet-getitem?mr=#1}{#2}
}
\providecommand{\href}[2]{#2}

\newpage

\phantom{a}

\vspace{0.5cm}

\begin{center}
{\Large \bf Supplementary Information}
\end{center}

\vspace{1cm}

Most of the derivations in this documents are based on non-rigorous method from
statistical physics. All the results that are rigorously proved
will be stated as lemmas, propositions, and so on.

\section{Notations}

\subsection{General notations}

%
% Algebra
%
We will often treat $\Z_2$ and $U(1)$ synchronization
simultaneously.
Throughout $\field = \reals$ or $\field = \complex$ 
depending on whether we are treating the real case ($\Z_2$
synchronization) or the complex case ($U(1)$ synchronization).

We let $S^{m-1}$ denotes the radius one sphere in 
$\reals^m$ or $\complex^m$ depending on the context.
Namely $S^{m-1} = \{\bz\in \field^m:\, \|\bz\|_2=1\}$. In particular
$S^0= \{+1,-1\}$ in the real case, and $S^0= \{z\in\complex: |z|=1\}$
in the complex case.

Some of our formulae depends upon the domain that we are  considering (real or
complex). In order to write them in a compact form,
we introduce the notation $s_\Group = 1$ for $\Group = \Z_2$, and
 $s_\Group = 2$ for $\Group = U(1)$. 

%
% Probability
%
%We use special symbols for special probability laws.

We write $X\sim\Poisson(a)$ to indicate that $X$ is a Poisson random variable with mean $a$.
A Gaussian random vector $\bz$  with mean $\ba = \E(\bz)$ and
covariance $\bC = \E((\bz-\ba)(\bz-\ba)^*)$
is denoted by $\bz\sim \normal(\ba,\bC)$. 
Note that in the complex case, this means that  $\bC$ is Hermitian and 
 $\E((\bz-\ba)(\bz-\ba)^{\sT}) = 0$. 
Occasionally, we will write $\bz\sim\cnormal(\ba,\cC)$ for complex
Gaussians, whenever it is useful to emphasize that $\bz$ is complex.

The standard Gaussian density is denoted by $\phi(x) =
e^{-x^2/2}/\sqrt{2\pi}$, and the  Gaussian distribution by $\Phi(x)
=\int_{-\infty}^x \phi(t) \,\de t$.

Given two un-normalized measures  $p$ and $q$ on the same space,
we write $p(s) \normeq q(s)$ if they are equal up to an overall
normalization constant. We use $\doteq$ to denote equality up to
subexponential factors, i.e. $f(n)\doteq g(n)$ if
$\lim_{n\to\infty}n^{-1}\log [f(n)/g(n)] =0$.

\subsection{Estimation metrics}

We recall the definition of some estimation metrics used in the main
text. For the sake of uniformity, we consider estimators
$\hbx:\reals^{n\times n}\to\reals^n$.

It is convenient to define a \emph{scaled \MSE}, with scaling factor
$c\in \field$:
\begin{align}
\MSE_n(\hbx;c) \equiv
  \frac{1}{n}\E\big\{\min_{s\in S^0}\big\|\bxz-s\,
  c\,\hbx(\bY)\big\|^2_2\big\} \, . \label{eq:ScaledMSEdef}
\end{align}
We also define  the overlap as follows in the real case
\begin{align}
\Ove_n(\hx) \equiv
  \frac{1}{n}\E\big\{|\<\sign(\hbx(\bY)),\bxz\>|\big\} \, .
\end{align}
In the complex case, we replace $\sign(z)$ by $z/|z|$ (defined to be
$0$ at $z=0$):
\begin{align}
\Ove_n(\hx) \equiv
  \frac{1}{n}\E\left\{\left|\sum_{i=1}^n\frac{\hx_i(\bY)}{|\hx_i(\bY)|}\,
  x_{0,i}\right|\right\} \, .\label{eq:OverlapDef}
\end{align}
This formula applies to the real case as well.
(Note that, in the main text, we defined the overlap only for
estimators taking values in $\{+1,-1\}^n$, in the real
case. Throughout these notes, we generalize  that definition for the
sake of uniformity.)

We omit the subscript $n$ to refer to the $n\to\infty$ limit of these
quantities.

\section{Preliminary facts}

\subsection{Some estimation identities}

\begin{lemma}\label{lemma:Expectation}
Let $p_0(\,\cdot\,)$ be a probability measure on the real line
$\reals$, symmetric around $0$ (i.e. $p_0((a,b)) = p_0((-b,-a))$ for
any interval $(a,b)$). For $\gamma\ge 0$, define $f:\reals\to\reals$ as
\begin{align}
f(y;\gamma) \equiv \frac{\int \sigma\,
  e^{\sqrt{\gamma}y\sigma-\frac{1}{2}\gamma\sigma^2}\, p_0(\de\sigma)}
{\int  e^{\sqrt{\gamma}y\sigma-\frac{1}{2}\gamma\sigma^2}\,
  p_0(\de\sigma)}\, .
\end{align}
Then we have the identity
\begin{align}
\E\{|\sigma|\, f(\sqrt{\gamma}|\sigma| +Z;\gamma)\} =\E\big\{
  f(\sqrt{\gamma}|\sigma| +Z;\gamma)^2\big\} \, . \label{eq:ExpectationTower}
\end{align}
where the expectation is with respect to the independent random
variables  $Z\sim\normal(0,1)$, and $\sigma\sim p_0(\,\cdot\,)$.

Analogously, let $p_0(\,\cdot\,)$ be a probability measure on
$\complex$, symmetric under rotations (i.e. $p_0(e^{i\theta}R) =
p_0(R)$, $p_0(R) = p_0(R^*)$ for
any Borel set $R\subseteq \complex$ and any $\theta\in (0,2\pi]$).
For $\gamma\ge 0$, define $f:\complex\to\reals$ as
\begin{align}
f(y;\gamma) \equiv \frac{\int \sigma\,
  e^{2\sqrt{\gamma}\Re(y^*\sigma)-\gamma|\sigma|^2}\, p_0(\de\sigma)}
{\int  e^{2\sqrt{\gamma}\Re(y^*\sigma)-\gamma|\sigma|^2}\,\label{eq:ExpectationTowerComplex}
  p_0(\de\sigma)}\, .
\end{align}
Then we have the identity (with $Z\sim \cnormal(0,1)$ a complex normal)
\begin{align}
\E\{|\sigma|\, f(\sqrt{\gamma}|\sigma| +Z;\gamma)\} =\E\big\{
  |f(\sqrt{\gamma}|\sigma| +Z;\gamma)|^2\big\} \, .
\end{align}
\end{lemma}
\begin{proof}
Consider, to be definite, the real case, and define the observation
model
\begin{align}
Y = \sqrt{\gamma}\, \sigma +Z\, ,
\end{align}
where $\sigma\sim p_0(\,\cdot\,)$ independent of the noise
$Z\sim\normal(0,1)$. Then a straightforward calculation shows that
\begin{align}
f(y;\gamma) = \E\{\sigma|Y=y\}\, .
\end{align}
Then, by the tower property of conditional expectation
$\E\{\sigma\, f(Y;\gamma)\} = \E\{f(Y;\gamma)^2\}$ or,
equivalently
\begin{align}
\E\{\sigma\, f(\sqrt{\gamma}\,\sigma+Z;\gamma)\} =
  \E\{f(\sqrt{\gamma}\,\sigma+Z;\gamma)^2\}\, .
\end{align}
The identity (\ref{eq:ExpectationTower}) follows by exploiting the
symmetry of $p_0$, which implies $f(-y;\gamma) = -f(y;\gamma)$.

The proof follows a similar argument in the complex case.
\end{proof}

We apply the above lemma to specific cases that will be of interest to us.
Below, $\BI_k(z)$ denotes the modified Bessel function of the second
kind. Explicitly, for $k$ integer, we have the integral representation
\begin{align}
\BI_k(z) = \frac{1}{2\pi}\int_{0}^{2\pi}e^{z\cos\theta}\cos(k\theta)\,
  \de\theta\, .
\end{align}
\begin{corollary}\label{coro:Identities}
For any $\gamma\ge 0$, we have the identities
\begin{align}
\E\big\{\tanh(\gamma+\sqrt{\gamma}\, Z)\big\} &= 
\E\big\{\tanh(\gamma+\sqrt{\gamma}\, Z)^2\big\}\, ,\label{eq:TanhIdentity}\\
\E\left\{\frac{\gamma+\sqrt{\gamma}\, Z}{|\gamma+\sqrt{\gamma}\, Z|}\,\frac{\BI_1(2|\gamma+\sqrt{\gamma}\, Z|)}{\BI_0(2|\gamma+\sqrt{\gamma}\, Z|)}\right\} &= 
\E\left\{\frac{\BI_1(2|\gamma+\sqrt{\gamma}\, Z|)^2}{\BI_0(2|\gamma+\sqrt{\gamma}\, Z|)^2}\right\}\, ,\label{eq:BesselIdentity}
\end{align}
where the expectation is with respect to $Z\sim\normal(0,1)$ (first line)
or $Z\sim \cnormal(0,1)$ (second line).
\end{corollary}
\begin{proof}
These follows from Lemma \ref{lemma:Expectation}. For the first line
we apply the real case~\eqref{eq:ExpectationTower} with  $p_0 =
(1/2)\delta_{+1}+(1/2)\delta_{-1}$, whence
\begin{align}
f(y;\gamma) = \tanh(\sqrt{\gamma}\, y)\, .
\end{align}

For the second line
we apply the complex case~\eqref{eq:ExpectationTowerComplex} with $p_0$ the uniform measure over the
unit circle. Consider the change of variables $y = |y| e^{j\phi}$ and $\sigma = e^{j(\phi+\theta)}$.
Computing the curve integral, we have

\begin{align}
f(y;\gamma) 
&= \frac{\int_0^{2\pi} \, e^{j(\phi+\theta)} e^{2\sqrt{\gamma}|y|\cos(\theta)} \de \theta}
{\int_0^{2\pi}\, e^{2\sqrt{\gamma}y \cos(\theta)} \de \theta}\nonumber\\
&= \frac{e^{j\phi}\int_0^{2\pi} \,  e^{2\sqrt{\gamma}|y|\cos(\theta)}  \cos(\theta) \de \theta}
{\int_0^{2\pi}\, e^{2\sqrt{\gamma}|y| \cos(\theta)} \de \theta}
= \frac{y}{|y|}\frac{\BI_1(2\sqrt{\gamma}|y|)}{\BI_0(2\sqrt{\gamma}|y|)}\, ,
\end{align}
where in the second equality we used the fact that $\int_{0}^{2\pi} e^{2\sqrt{\gamma}|y| \cos(\theta)} \sin(\theta) \de \theta = 0$.
\end{proof}

\section{Analytical results for $\Z_2$ and and $U(1)$ synchronization}\label{sec:DenseGeneral}

As explained in the main text, we are interested in the following
probability measure over $\bsigma =
(\bsigma_1,\bsigma_2,\dots,\bsigma_n)$,
where $\bsigma_i\in S^{m-1}$:
\begin{align}
p_{\beta,m}(\de \bsigma)= \frac{1}{Z_{n,\bY}(\beta,m)}\exp\Big\{2m\beta \sum_{i<j} \Re(Y_{ij}
  \<\bsigma_i,\bsigma_j\>)\Big\}\, \prod_{i=1}^np_0(\de \bsigma_i)\, . \label{eq:GibbsModel}
\end{align}
Here $p_0(\de\bsigma_i)$ is the uniform measure over $\bsigma_i\in
S^{m-1}$. 

%By symmetry we can assume that $\bxz = (+1,+1,\dots,+1)$ and hence 
%observations are distributed according to
%%
%\begin{align}
%%
%Y_{ij} = \frac{\lambda}{n} +W_{ij}\, ,
%%
%\end{align}
%%
%where $W_{ij}\sim \normal(0,1/n)$. \aj{Where do you use this simplifying assumption?}

We define a general $m,\beta$ estimator as follows.
\begin{enumerate}
\item In order to  break the $\cO(m)$ symmetry, we add a term
  $\beta\sum_{i=1}^n\<\bh,\bsigma\>$ in the exponent of
  Eq.~(\ref{eq:GibbsModel}), for $\bh$ an arbitrary small vector. It
  is understood throughout that $\|\bh\|_2\to 0$ after $n\to\infty$. 

As is customary in statistical physics, we will not explicitly carry
out calculations with the perturbation $\bh\neq 0$, but only using
this device to select the relevant solution at $\bh=0$.
\item We compute the expectation
\begin{align}
\bz(\bY;\beta,m) = \int \bsigma \;  p_{\beta,m}(\de \bsigma)\, .\label{eq:BetaEstimator}
\end{align}
Note that for $\beta\to\infty$ this amounts to maximizing the
exponent term in equation~\eqref{eq:GibbsModel}. 
\item Compute the empirical covariance
\begin{align}
\hbQ \equiv \frac{1}{n}\sum_{i=1}^n\bz_i\bz_i^*\, .
\end{align}
Let $\hbu$ be its principal eigenvector.
\item Return $\hbx^{(\beta,m)}$
\begin{align}
\hx_i^{(\beta,m)}= c\, \<\hbu,\bz_i\>\, , \label{eq:GeneralEstimator}
\end{align}
where $c = c(\lambda)$ is the optimal scaling predicted by the
asymptotic theory.
\end{enumerate}

\subsection{Derivation of the Gibbs measure}
\label{sec:Gibbs}

The Gibbs measure (\ref{eq:GibbsModel}) encodes several estimators of
interests. Here we briefly describe this connections.

\vspace{0.5cm}

\noindent{\bf Bayes-optimal estimators.} As mentioned in the main
text, this is obtained by setting $m=1$ and $\beta = \lambda/2$ (in
the real case) or $\beta = \lambda$ (in the complex case). To see
this, recall our observation model (for $i<j$)
\begin{align}\label{eq:ObsModel}
Y_{ij} = \frac{\lambda}{n}x_{0,i}x_{0,j}^* + W_{ij}\, ,
\end{align}
with $x_{0,i}\in S^0$ and $W_{ij}\sim \normal(0,1/n)$.
Hence, by an application of Bayes formula, the conditional density of
$\bxz = \bsigma\in\field^n$ given $\bY$ is becomes
\begin{align}
p(\de\bsigma|\bY)  &= \frac{1}{Z'_{\bY}} \,
  \exp\Big\{-\frac{n\sg}{2}\sum_{i<j}\Big|Y_{ij}-\frac{\lambda}{n}\sigma_i\sigma_j^*\Big|^2\Big\}
\,
                      \prod_{i=1}^np_0(\de \sigma_i)\, ,
\end{align}
where we recall that $\sg=1$ for $\Group = \Z_2$ (real case), and $\sg=2$ for $\Group =
U(1)$  (complex case). Further, $p_0(\de\sigma)$ denotes the uniform
measure over $S^0\in\field$ (in particular, this is the uniform
measure over $\{+1,-1\}$ for $\Group = \Z_2$). 
Expanding the square and re-absorbing terms independent of $\bsigma$ 
in the normalization constant, we get
\begin{align}
p(\de\bsigma|\bY)  &=  \frac{1}{Z_{\bY}} \,
  \exp\Big\{\lambda \sg\sum_{i<j}\Re(Y_{ij}\sigma_i^*\sigma_j)\Big\}\,
                      \prod_{i=1}^np_0(\de \sigma_i)\, . 
\end{align}
As claimed, this coincides with
Eq.~(\ref{eq:GibbsModel}) if we set $\beta = \lambda/2$ (in
the real case) or $\beta = \lambda$ (in the complex case).

\vspace{0.5cm}

\noindent{\bf Maximum-likelihood
  and SDP estimators.} By letting
$\beta\to\infty$ in Eq.~(\ref{eq:GibbsModel}), we obtain
that $p_{\beta,m}$ concentrates on
the maximizers of the problem
\begin{align}
\mbox{maximize} & \;\;\;\sum_{i<j}Y_{ij}\<\bsigma_i,\bsigma_j\>\, ,\label{eq:NonConvex}\\
\mbox{subject to} & \;\;\; \bsigma_i\in S^{m-1}\;\; \forall i\in [n]\, ,\nonumber
\end{align}
In the case $m\ge n$ we recover the
SDP relaxation.
In the case $m=1$, this is
equivalent to the maximum likelihood
problem
\begin{align}
\mbox{maximize} & \;\;\;\sum_{i<j}\big|Y_{ij}-\sigma_i\sigma_j^*\big|^2\, ,\\
\mbox{subject to} & \;\;\; \sigma_i\in S^{0}\;\; \forall i\in [n]\, .\nonumber
\end{align}

\subsection{Cavity derivation for $\Z_2$ and $U(1)$ synchronization}

In this section we use the cavity method to derive the asymptotic
properties of the measure (\ref{eq:GibbsModel}).

\subsubsection{General $m$ and $\beta$}

In the replica-symmetric cavity method, we consider adding a single
 variable $\bsigma_0$ to a problem with $n$ variables
 $\bsigma_1,\bsigma_2,\dots,\bsigma_n$. We compute the 
marginal distribution of $\bsigma_0$ in the system with $n+1$
variables, to be denoted by $\nu^{n+1}_0(\de\bsigma_0)$. 
This is expressed in terms of the marginals of
the other variables in the system with $n$ variables
$\nu^{n}_1(\de\bsigma_0)$, \dots $\nu^n_n(\de\bsigma_1)$.
We will finally impose the consistency condition that $\nu^{n+1}_0$
is distributed as any of $\nu^n_1$, \dots $\nu^n_n$ in the
$n\to\infty$ limit.

Assuming that $\bsigma_1$, \dots $\bsigma_n$ are, for this purpose,
approximately independent, we get
\begin{align}
\nu^{n+1}_{0}(\de\bsigma_0)\; \normeq \; p_0(\de\bsigma_0)\prod_{k =1}^n
 \int \exp\Big\{2\beta m\Re(Y_{0k}
  \<\bsigma_0,\bsigma_k\>)\Big\}\; \nu^n_{k}(\de\bsigma_k)\, .\label{eq:FirstCavity}
\end{align}
We will hereafter drop the superscripts $n$, $n+1$ assuming they will
be clear from the range of the subscripts.

Next we consider a fixed $k\in [n]$ and estimate the integral
by expanding the exponential term.
This expansion proceeds slightly different in the real and the complex
cases. We give details for the first one, leaving the second to the reader. Write
\begin{align*}
&\int \exp\Big\{2\beta m Y_{0k}
  \<\bsigma_0,\bsigma_k\>)\Big\}\; \nu_{k}(\de\bsigma_k) = \\
&=1 + 2\beta m Y_{0k} \<\bsigma_0,\rE_{k}(\bsigma_k)\>
+\frac{1}{2}\, 4\beta^2m^2 Y^2_{0k}
  \<\bsigma_0,\rE_{k}(\bsigma_k\bsigma_k^*)\bsigma_0\> +O(n^{-3/2})\\
&=\exp\Big\{2\beta m Y_{0k} \<\bsigma_0,\rE_{k}(\bsigma_k)\>
+\frac{1}{2}\, 4\beta^2m^2 Y^2_{0k}
  \<\bsigma_0,\big(\rE_{k}(\bsigma_k\bsigma_k^*)-
  \rE_{k}(\bsigma_k)\rE_{k}(\bsigma_k^*)\big)\bsigma_0\>
  +O(n^{-3/2})\Big\}\, ,
\end{align*}
where $\rE_{k}(\,\cdot\,) \equiv \int
(\,\cdot\,)\,\nu_{k}(\de\bsigma_k)$ denotes expectation with respect
to $\nu_{k}$. Here, we used the fact that $Y_{ij} = O(1/\sqrt{n})$, as per equation~\eqref{eq:ObsModel}.

Substituting in Eq.~(\ref{eq:FirstCavity}), and neglecting
$O(n^{-1/2})$ terms, we get (both in the real and complex case)
\begin{align}
\nu_{i}(\de\bsigma_i) =
\frac{1}{z_{i}}\, 
\exp\Big\{2\beta\, m\,\Re\<\bxi_{i},\bsigma_i\> + \frac{2\beta m}{ \sg}\,
\<\bsigma_i,\bC_{i}\bsigma_i\>\Big\}\, p_0(\de\bsigma_i)\, \, ,\label{eq:FirstCavityField}
\end{align}
where $\bxi_{i}\in\field^m$, $\bC_{i}\in\field^{m\times
  m}$, with $\bC_i^*=\bC_i$, $z_i$ is a normalization constant, and $\sg=1$ (real case) or $\sg=2$ (complex case). 
This expression holds for $i=0$ and, by the consistency condition,
for all $i\in\{1,\dots, n\}$.

We further have the following equations for $\bxi_{0}$,
$\bC_{0}$:
\begin{align}
\bxi_{0} & = \sum_{k=1}^n Y^*_{0k} \rE_{k}(\bsigma_k)\, ,\label{eq:Cavity_1}\\
\bC_{0} & = \beta m\sum_{k=1}^n |Y_{0k}|^2
\big(\rE_{k}(\bsigma_k\bsigma_k^*)-\rE_{k}(\bsigma_k)\rE_{k}(\bsigma_k^*)\big)\, .\label{eq:Cavity_2}
\end{align}
Notice that the expectations $\rE_k(\,\cdot\,)$ on the right-hand
side are in fact
functions of $\bxi_k$, $\bC_k$ through Eq.~(\ref{eq:FirstCavityField}).

We next pass to studying the distribution of $\{\bxi_{k}\}$
and $\{\bC_{k}\}$. For large $n$, the pairs  $\{(\bxi_{k},\bC_{k})\}$
appearing on the right-hand side of Eqs.~(\ref{eq:Cavity_1}),
(\ref{eq:Cavity_2}) can be treated as independent. 
By the law of large numbers and central limit theorem, we obtain that
\begin{align}
\bxi_i\sim \normal(\bmu,\bQ)\, ,\;\;\;\;\; \bC_i = \bC\, ,
\end{align}
for some deterministic quantities $\bmu$, $\bQ$, $\bC$.
Note that the law of $\bxi_i$ can be equivalently described by
\begin{align}
\bxi_i = \bmu + \bQ^{1/2}\bg\, ,
\end{align} 
where $\bg = (g_1,g_2,\dots,g_m)\sim \normal(0,\id_m)$.

Using these and the consistency condition in Eqs.~(\ref{eq:Cavity_1}),
(\ref{eq:Cavity_2}), we obtain the following equations for 
the unknowns $\bmu, \bQ, \bC$:
\begin{align}
\bmu & = \lambda\E\big\{\rE_{\bxi,\bC}(\bsigma)\big\}\, ,\label{eq:CavityGen_1}\\
\bQ& = \E\big\{\rE_{\bxi,\bC}(\bsigma) \rE_{\bxi,
     \bC}(\bsigma)^*\big\}\, ,\label{eq:CavityGen_2}\\
\bC & = \beta m\, \E\big\{\rE_{\bxi,\bC} (\bsigma\bsigma^*)-
      \rE_{\bxi,\bC}(\bsigma) \rE_{\bxi,\bC}(\bsigma)^*\big\}\, .
\label{eq:CavityGen_3}
\end{align}
Here $\E$ denotes expectation with respect to $\bxi\sim
\normal(\bmu,\bQ)$. Further $\rE_{\bxi,\bC}$ denotes expectation with
respect to the following probability measure on $S^{m-1}$:
\begin{align}
\nu_{\bxi,\bC} (\de\bsigma) =
\frac{1}{z(\bxi,\bC)}\,
\exp\Big\{2\beta\, m\Re\<\bxi,\bsigma\> + \frac{2\beta\,
  m}{\sg}
\<\bsigma,\bC\bsigma\>\Big\}\, p_0(\de\bsigma)\, .\label{eq:Generalnu}
\end{align}

The prediction of the replica-symmetric cavity methods have been summarized in the main text.  We generalize the discussion here.
Assume, for simplicity $\bxz = (+1,\dots,+1)$. 
For  $\nu$ a probability measure on $S^{m-1}$ and $R$ an orthogonal (or
unitary) matrix, let $\nu^R$ be the measure obtained
by `rotating'
$\nu$, i.e. $\nu^R(\bsigma\in A) \equiv \nu(R^{-1}\bsigma
  \in A)$ for any measurable set $A$. Finally, let $p^{(m,\beta)}_{i(1),\dots,i(k)}$ denote
the joint distribution of $\bsigma_{i(1)},\cdots,\bsigma_{i(k)}$ under $p_{m,\beta}$.
Then, for any fixed $k$, and any sequence of $k$-tuples $(i(1),\dots,i(k))_n\in [n]$, we
have
\begin{align}
p^{(m,\beta)}_{i(1),\dots,i(k)}\Rightarrow 
\int \nu^R_{\bxi_1,\bC} (\,\cdot\,)\times  \cdots\times \nu^R_{\bxi_k,\bC} (\,\cdot\,)\,
  \de R
\end{align}
Here $\de R$ denotes the uniform (Haar) measure on the orthogonal
group, `$\Rightarrow$' denotes convergence in distribution, and 
$\bxi_1,\dots,\bxi_k\sim_{iid}
\normal(\mu,\bQ)$ as above. Note that the original measure~\eqref{eq:GibbsModel} 
is unaltered under multiplication by a phase. Specifically, if we add an arbitrary phase $\phi_\ell$
to coordinate $\ell$ of the spins, the measure remains unchanged. 
In order to break this invariance, we consider marginals $\nu_{\bxi,\bC}$
that corresponds to $\bmu$ being real-valued. Henceforth, without loss of generality we stipulate that $\bmu$ is a real-valued
vector.    

While in general this prediction is only a good approximation (because
of replica symmetry breaking) we expect to be asymptotically exact for the
Bayes-optimal, ML and SDP estimator. In the next sections we will discuss special estimators.

\subsubsection{Bayes-optimal: $m=1$ and $\beta \in \{\lambda/2,
  \lambda\}$}

For $m=1$, $\bsigma= \sigma$ is a scalar satisfying $\sigma\sigma^*=
|\sigma|^2=1$. Hence, the term proportional to $\bC$ in
Eq.~(\ref{eq:Generalnu}) is a constant and can be  dropped. 
Also $\bxi=\xi$, $\bmu = \mu$ and $\bQ = q$ are scalar in this case.

The expression (\ref{eq:Generalnu}) thus reduces to
\begin{align}
\nu_{\xi}  (\de\sigma) =
\frac{1}{z(\xi)}\, 
e^{2\beta\, \Re(\xi^*\sigma) }\, p_0(\de\sigma)\, ,\label{eq:NuM=1}
\end{align}
with $p_0(\de\sigma)$ the uniform measure on $\{+1,-1\}$ (in the real
case) or on $\{z\in\complex :\; |z|=1\}$ (in the complex case).
Substituting in Eqs.~(\ref{eq:CavityGen_1}), (\ref{eq:CavityGen_2}),
we get
\begin{align}
\mu & = \lambda\E\big\{\rE_{\xi}(\sigma)\big\}\, ,\label{eq:General_Bayes_1}\\
q& = \E\big\{|\rE_{\xi}(\sigma)|^2\big\}\, ,\label{eq:General_Bayes_2}
\end{align}
with $\rE_{\xi}$ denoting expectation with respect to $\nu_{\xi}$, 
and
$\E$ expectation with respect to $\xi\sim\normal(\mu,q)$.

We will write these equations below in terms of classical functions both in
the real and in the complex cases. Before doing that,
we derive expressions for the estimation error in the
$n\to\infty$ limit.
The estimator $\hbx(\bY)$ is given in this case by
$\hbx(\bY)=\hbx^{\beta,m=1}(\bY)$, cf.
Eq.~(\ref{eq:BetaEstimator}). Therefore, the scaled MSE, cf. 
Eq.~(\ref{eq:ScaledMSEdef}), reads
\begin{align}
\MSE(\hbx;c) &=
\lim_{n\to\infty}\left\{1-\frac{2c}{n}\sum_{i=1}^n\E\big\{\Re(x_{0,i}\hx_{i}(\bY))\big\}
  +\frac{c^2}{n}\sum_{i=1}^n\E\big\{|\hx_{i}(\bY)|^2\big\}\right\}\\
& = 1-2c\E\big\{\Re \rE_{\xi}(\sigma)\big\} +
  c^2\E\big\{|\rE_{\xi}(\sigma)|^2\big\}\\
& = 1-\frac{2c\mu}{\lambda}+c^2 q\, . \label{eq:MSEgeneral}
\end{align}

Note that the optimal scaling is $c=\mu/(\lambda q)$, leading to
minimal error --for the ideally scaled estimator--
\begin{align}
\MSE(\hbx) & = 1-\frac{\mu^2}{\lambda^2 q}\, . 
\end{align}

For the overlap we have, from Eq.~(\ref{eq:OverlapDef}),
\begin{align}
\Ove(\hx)&=\lim_{n\to\infty}
  \frac{1}{n}\E\left\{\left|
\sum_{i=1}^n\frac{\hx_i(\bY)}{|\hx_i(\bY)|}\, x_{0,i}\right|\right\} \\
& = \E\Big\{\frac{\rE_{\xi}(\sigma)}{|\rE_{\xi}(\sigma)|}\Big\}\, . \label{eq:OverlapGeneral}
\end{align}

\vspace{0.5cm}

\noindent {\bf Real case.} In this case $\rE_{\xi}(\sigma) =
\tanh(2\beta\xi)$
and therefore Eqs.~(\ref{eq:General_Bayes_1}),
(\ref{eq:General_Bayes_2}) yield
\begin{align}
\mu & = \lambda\E\big\{\tanh\big(2\beta\mu +2\beta\sqrt{q} Z\big)\big\}\, ,\\
q& = \E\big\{\tanh\big(2\beta\mu +2\beta\sqrt{q} Z\big)^2\big\}\, ,
\end{align}
where expectation is with respect to $Z\sim\normal(0,1)$. 
As discussed in Section \ref{sec:Gibbs}, the Bayes optimal estimator
is recovered by setting $\beta = \lambda/2$ above. 
Using the identity (\ref{eq:TanhIdentity})
in Corollary \ref{coro:Identities}, we obtain the solution
\begin{align}
\mu =\frac{\kappa}{\lambda}\, ,\;\;\;\;\;
q =\frac{\kappa}{\lambda^2}\, .
\end{align}
where $\kappa$ satisfies the fixed point equation 
\begin{align}
\kappa= \lambda^2 \E\big\{\tanh\big(\kappa+\sqrt{\kappa} Z\big)\big\}
\end{align}
We denote by $\kappa_* = \kappa_*(\lambda)$ the largest non-negative solution of
this equation. Using Eqs.~(\ref{eq:MSEgeneral}) and
(\ref{eq:OverlapGeneral}) we obtain the following predictions for the
asymptotic estimation error
\begin{align}
\MSE(\xBayes) &= 1-\frac{\kappa_*(\lambda)}{\lambda^2}\, ,\\
\Ove(\xBayes) &= 1-2\Phi(-\sqrt{\kappa_*(\lambda)})\, .
\end{align}
(Note that in this case, the optimal choice of a scaling is $c=1$.) 

\vspace{0.5cm}

\noindent {\bf Complex case.} In this case 
\begin{align}
\rE_{\xi}(\sigma) =\frac{\xi}{|\xi|}\, \frac{\BI_1(2\beta|\xi|)}{\BI_0(2\beta|\xi|)}\, ,
\end{align}
where, as mentioned above, $\BI_k(z)$ denotes the modified Bessel function of the second
kind. 

The general fixed point equations~(\ref{eq:General_Bayes_1}) and
(\ref{eq:General_Bayes_2}) yield
\begin{align}
\mu & = \lambda\E\Big\{\frac{\mu+\sqrt{q} \Re(Z)}{|\mu+\sqrt{q}
      Z|}\, \frac{\BI_1(2\beta|\mu+\sqrt{q}Z|)}{\BI_0(2\beta|\mu+\sqrt{q}
      Z|)}\Big\}\, ,\\
q & = \lambda\E\Big\{ \frac{\BI_1(2\beta|\mu+\sqrt{q}
      Z|)^2}{\BI_0(2\beta|\mu+\sqrt{q}
      Z|)^2}\Big\}\, .
\end{align}
As discussed in Section \ref{sec:Gibbs}, the Bayes optimal estimator
is recovered by setting $\beta = \lambda$ in these equations.
In this case we can use the identity (\ref{eq:BesselIdentity})
in Corollary \ref{coro:Identities}, to obtain the solution
\begin{align}
\mu =\frac{\kappa}{\lambda}\, ,\;\;\;\;\;
q =\frac{\kappa}{\lambda^2}\, .
\end{align}
where $\kappa$ satisfies the fixed point equation 
\begin{align}
\kappa = \lambda^2 \E\Big\{\frac{\kappa+\sqrt{\kappa} \Re(Z)}{|\kappa+\sqrt{\kappa}
      Z|}\, \frac{\BI_1(2|\kappa+\sqrt{\kappa}Z|)}{\BI_0(2|\kappa+\sqrt{\kappa}
      Z|)}\Big\}\, ,
\end{align}
where the expectation is taken with respect to $Z\sim\cnormal(0,1)$.
We denote by $\kappa_*=\kappa_*(\lambda)$ the largest non-negative solution of
these equations.

Using again Eqs.~(\ref{eq:MSEgeneral}) and
(\ref{eq:OverlapGeneral}) , we obtain
\begin{align}
\MSE(\xBayes) &= 1-\frac{\kappa_*(\lambda)}{\lambda^2}\, ,\\
\Ove(\xBayes) &= \E\Big\{\frac{\kappa_*(\lambda)+\sqrt{\kappa_*(\lambda)}\, \Re(Z)}{|\kappa_*(\lambda)+\sqrt{\kappa_*(\lambda)}
      Z|}\Big\}\, .
\end{align}

\subsubsection{Maximum likelihood: $m=1$ and $\beta \to\infty$}
\label{eq:MLDense}

As discussed in Section \ref{sec:Gibbs}, the maximum likelihood
estimator is recovered by setting $m=1$ and $\beta\to\infty$.
Notice that \emph{in this case our results are only approximate
  because of replica symmetry breaking.} 

We can take the limit $\beta\to\infty$ in Eqs.~(\ref{eq:NuM=1}),
(\ref{eq:General_Bayes_1}), (\ref{eq:General_Bayes_2}). In this limit,
the measure $\nu_{\xi}(\,\cdot\,)$
concentrates on the single point $\sigma^*=\xi/|\xi|\in S^0$. We thus obtain
\begin{align}
\mu & = \lambda\E\Big\{\frac{\Re(\xi)}{|\xi|}\Big\}\, ,\label{eq:ML_RS}\\
q& = 1\,.
\end{align}
We next specialize our discussion to the real and complex cases.

\vspace{0.5cm}

\noindent{\bf Real case.} Specializing Eq.~(\ref{eq:ML_RS}) to the
real case, we get the equation
\begin{align}
\mu & = \lambda\big(1-2\Phi(-\mu)\big)\, . \label{eq:FP_ML}
\end{align}
Taylor expanding near $\mu =0$, this yields
$\mu = 2\phi(0)\lambda\,\mu+O(\mu^2)$ which yields the critical point
(within the replica symmetric approximation)
\begin{align}
\lambda_{c}^{\mbox{\tiny{ML,RS}}}= \sqrt{\frac{\pi}{2}}\approx 1.253314\, .
\end{align}
We denote by $\mu_*= \mu_*(\lambda)$ the largest non-negative solution
of Eq.~(\ref{eq:FP_ML}). The asymptotic estimation metrics (for
optimally scaled estimator) at level $\lambda$ are
given by
\begin{align}
\MSE(\xml) & = 1-\frac{\mu_*(\lambda)^2}{\lambda^2}\, ,\label{eq:MSE_ML}\\
\Ove(\xml) & =
             \frac{\mu_*(\lambda)}{\lambda}\, . \label{eq:Overlap_ML}
\end{align}
It follows immediately from Eq.~(\ref{eq:FP_ML}) that, as
$\lambda\to\infty$,  $\mu_*(\lambda)
=\lambda[1-2\Phi(-\lambda)+O(\Phi(-\lambda)^2)]$, whence
\begin{align}
\MSE(\xml) & = 4\Phi(-\lambda) +O(\Phi(-\lambda)^2)=
\sqrt{\frac{8}{\pi\lambda^2}}\, e^{-\lambda^2/2}\, \big(1+O(\lambda^{-1})\big)\, ,\\
\Ove(\xml) & = 1-2\Phi(-\lambda) +O(\Phi(-\lambda)^2)=
1-\sqrt{\frac{2}{\pi\lambda^2}}\, e^{-\lambda^2/2}\, \big(1+O(\lambda^{-1})\big)\, .
\end{align}

\vspace{0.5cm}

\noindent{\bf Complex case.} Specializing Eq.~(\ref{eq:ML_RS}), we get
\begin{align}
\mu & = \lambda\E\Big\{\frac{\mu+\Re(Z)}{|\mu+Z|}\Big\}\, ,\label{eq:ML_Complex}
\end{align}
where the expectation is with respect to $Z\sim\cnormal(0,1)$. 
Taylor-expanding around $\mu = 0$, we get $\mu =
(\lambda\mu/2)\E\{|Z|^{-1}\}+O(\mu^2)$.
Using $\E\{|Z|^{-1}\} = \sqrt{\pi}$, we obtain the replica-symmetric
estimate for the critical point
\begin{align}
\lambda_{c}^{\mbox{\tiny{ML,RS}}}= \frac{2}{\sqrt{\pi}} \approx 1.128379\, .
\end{align}
Denoting by $\mu_* = \mu_*(\lambda)$ the largest non-negative solution
of Eq.~(\ref{eq:ML_Complex}), the estimation metrics are obtained
again via Eqs.~(\ref{eq:MSE_ML}) and (\ref{eq:Overlap_ML}).

For large $\lambda$, it is easy to get $\mu_*(\lambda)/\lambda =
1-(4\lambda^2)^{-1}+O(\lambda^{-3})$ whence
\begin{align}
\MSE(\xml) & = \frac{1}{2\lambda^2}+O(\lambda^{-3})\, ,\\
\Ove(\xml) & = 1-\frac{1}{4\lambda^2}+O(\lambda^{-3}) \, .
\end{align}

\subsubsection{General $m$ and $\beta \to\infty$}

In the limit $\beta\to\infty$, the measure
$\nu_{\bxi,\bC}(\,\cdot\,)$ of Eq.~(\ref{eq:Generalnu}) concentrates
on the single point $\bsigma_*(\bxi,\bC)$ that maximizes the
exponent. A simple calculation yields
\begin{align}
\bsigma_*(\bxi,\bC) = \big(\rho\id-(2/\sg)\bC\big)^{-1}\bxi\, ,
\end{align}
where $\rho$ is a Lagrange multiplier determined by the normalization
condition $\|\bsigma_*\|_2 = 1$, or
\begin{align}
\<\bxi, \big(\rho\id-(2/\sg)\bC\big)^{-2}\bxi\> = 1\, .
\end{align}
Further $\nu_{\bxi,\bC}(\,\cdot\,)$ has variance of order $1/\beta$
around $\bsigma_*$. 

In order to solve Eqs.~(\ref{eq:CavityGen_1})
to (\ref{eq:CavityGen_3}) we next assume that the $\cO(m)$ symmetry is
--at most-- broken vectorially to $\cO(m-1)$.
Without loss of generality, we can assume that it is broken along the
direction $\bfe_1=(1,0,\dots,0)$. Further, since
$\nu_{\bxi,\bC}(\,\cdot\,)$ is a measure on the unit sphere
$\{\bsigma:\; \|\bsigma\|_2=1\}$,  the matrix $\bC$ is only defined up
to a shift $\bC-c_0\id$. This leads to the following ansatz
for the order parameters.
\begin{align}
\bmu = \left(\begin{matrix}
\mu\\
0\\
\cdot\\
\cdot\\
\cdot\\
0
\end{matrix}\right)\, ,
\;\;\;\;\;\;\;
\bC= \left(\begin{matrix}
-\sg r/2 & & & & &\\
&0& & & &\\
&&\cdot&&&\\
&&&\cdot&&\\
&&&&\cdot&\\
&&&&&0
\end{matrix}\right)\, ,
\;\;\;\;\;\;\;
\bQ= \left(\begin{matrix}
q_1 & & & & &\\
&q_0& & & &\\
&&\cdot&&&\\
&&&\cdot&&\\
&&&&\cdot&\\
&&&&&q_0
\end{matrix}\right)\, ,\label{eq:AnsatzOn}
\end{align}
where it is understood that out-of-diagonal entries vanish. 
Note that the above structure on $(\bmu,\bC,\bQ)$ is the only one that remains invariant under 
rotations in $\field^m$ that map $\be_1$ to itself.
We then can represent $\bxi$ as follows. For 
$ Z_1,\dots,Z_m\sim_{i.i.d.}\normal(0,1)$:
\begin{align}
\bxi = \big(\mu+\sqrt{q_1}\, Z_1,\sqrt{q_0}\, Z_2,\dots,\sqrt{q_0}\,
  Z_m\big)^{\sT}\, ,
\end{align}
and $\bsigma_* (\bxi,\bC)$ reads
\begin{align}
\bsigma_* =\Big(\frac{\mu+\sqrt{q_1}\, Z_1}{\rho+r},
  \frac{\sqrt{q_0}\, Z_2}{\rho},\dots,
\frac{\sqrt{q_0}\,
  Z_m}{\rho}\Big)^{\sT}\, .\label{bsmEq}
\end{align}

Taking the limit $\beta\to\infty$ of Eqs.~(\ref{eq:CavityGen_1})
to (\ref{eq:CavityGen_3}) we obtain the following four equations 
for the four parameters $\mu,r,q_0,q_1$:
\begin{align}
\mu &= \lambda\, \E\left[ \frac{\mu + \sqrt{q_1}\,\Re(Z_1)}{\rho+r}\right]\,,\label{muEq}\\
q_1 &= \E\left[\frac{|\mu + \sqrt{q_1}\,Z_1 |^2}{(\rho+r)^2}\right]\,,\label{q1Eq}\\
q_0 & = q_0 \E\left[\frac{|Z_2|^2}{\rho^2}\right]\,,\label{q0Eq}\\
r &=
\E\left[\frac{|Z_2|^2}{\rho} - \frac{\mu\,\Re(Z_1)}{\sqrt{q_1}(\rho+r)}
-\frac{|Z_1|^2}{\rho+r}\right]\,.\label{c1Eq}
\end{align}

Further, the normalization condition $\Tr(\bQ) =
\E(\|\bsigma_*(\bxi,\bC)\|^2_2) =1$ yields
\begin{align}
q_1+(m-1)q_0 = 1\, .\label{eq:Normaliz}
\end{align}
In the above expressions, expectation is with respect to the Gaussian vector
$\bZ = (Z_1,\dots,Z_m)\sim\normal(0,\id_m)$, and $\rho =
\rho(Z_1,\dots,Z_m)$
is defined as the solution of the equation
\begin{align}
1 &=  \frac{| \mu + \sqrt{q_1}\,Z_1|^2}{(\rho+r)^2} +
\frac{q_0}{\rho^2}\sum_{i=2}^m |Z_i|^2\,. \label{rhoEq}
\end{align}
The simplest derivation of these equations is obtained by
differentiating the ground state energy, for which we defer to Section \ref{sec:FreeEnergy}.

We can then compute the performance of the estimator
$\hbx^{(\beta,m)}$ defined at the beginning of this section. Note that 
$\hbQ\to \bQ$ as $n\to\infty$, and therefore its principal vector is
$\hbu\to \bfe_1$ (within the above ansatz), and therefore, for a test
function $f$, we have 
\begin{align}\label{eq:RULE}
\frac{1}{n}\sum_{i=1}^n f(x_{0,i},\hx_i^{(\infty,m)}) =
\E\Big\{f\Big(X_0,c\,\frac{\mu X_0+\sqrt{q_1}
  Z_1}{\rho+r}\Big)\Big\}\, ,
\end{align}
where $X_0\sim\Unif(S^0)$ independent of $Z_1$.

Applying~\eqref{eq:RULE} and after a simple calculation we obtain 
\begin{align}
\MSE(\hbx^{(\infty,m)}) =
1-\frac{\mu_*(\lambda)^2}{\lambda^2 q_{1,*}(\lambda)}\, ,
\end{align}
where $\mu_*(\lambda)$, $q_{1,*}(\lambda)$ denote the solutions of the above equations.
Also, invoking~\eqref{eq:RULE} the asymptotic overlap is given by
\begin{align}
\Ove(\hbx^{(\infty,m)}) &= \E\Big\{X_0 \frac{\mu_* X_0+\sqrt{q_{1,*}} Z_1}{|\mu_* X_0+\sqrt{q_{1,*}} Z_1|} \Big\}\nonumber\\
&= 1-2\Phi\Big(-\frac{\mu_*(\lambda)}{\sqrt{q_{1,*}(\lambda)}}\Big)\, .
\end{align}

\vspace{0.5cm}

\noindent{\bf Spin-glass phase.}
The spin-glass phase is described by the completely symmetric solution
with $\mu = 0$, $b=0$ and $q_0=q_1 = 1/m$. 
From Eq.~(\ref{rhoEq}) we get
\begin{align}
\rho^2=\frac{1}{m}\, \|\bZ\|_2^2\, .
\end{align}

\vspace{0.5cm}

\noindent{\bf Critical signal-to-noise ratio.}
We next compute the critical value of $\lambda$. We begin by
expanding Eq.~(\ref{rhoEq}). Define 
\begin{align}
F(s) \equiv  \frac{q_1}{(s+r)^2}\,
|Z_1|^2+\frac{q_0}{s^2}\sum_{i=2}^m |Z_i|^2\, ,
\end{align}
and let $\rho_0$ be the solution of the equation $1=F(\rho_0)$. 
Notice that $\rho_0$ is unaltered under sign change $Z_1\to -Z_1$. Further,
comparing with the equation for $\rho$, see Eq.~(\ref{rhoEq}), we
obtain the following perturbative estimate
\begin{align}
\rho = \rho_0 -\frac{1}{F'(\rho_0)}
\frac{2\mu\sqrt{q_1}\Re(Z_1)}{(\rho_0+r)^2} +O(\mu^2)\, .
\end{align}
By the results for the spin glass phase, we have $q_0,q_1 \to (1/m)$
and $\rho, (\rho+r)\to \|\bZ\|_2/\sqrt{m}$ as $\mu\to 0$,
whence
\begin{align}
\rho = \rho_0+\frac{\Re(Z_1)}{\|\bZ\|_2}\, \mu +o(\mu)\, .\label{eq:rhoExpansion}
\end{align}

Now consider  Eq.~(\ref{muEq}). Retaining only $O(\mu)$ terms we get
\begin{align}
\mu &= \lambda \mu \, \E\left[\frac{1}{\rho+r}\right]
+\lambda\sqrt{q_1}\, \E\left[\frac{\Re(Z_1)}{\rho+r}\right]\\
& = \lambda\mu\, \E\left[\frac{\sqrt{m}}{\|\bZ\|_2}\right]
+\lambda\sqrt{q_1}\, \E\left[\frac{\Re(Z_1)}{\rho_0+r}\right]-
\lambda\frac{1}{\sqrt{m}}\,
\E\left[\frac{\Re(Z_1)}{(\rho_0+r)^2}\frac{\Re(Z_1)}{\|\bZ\|_2}\right]\,
\mu+o(\mu)\, .
\end{align}
where in the last step we used Eq.~(\ref{eq:rhoExpansion}) and $q_1 =
1/m +o(1)$ as $\mu\to 0$.  Now recalling that $\rho_0$ is even in
$Z_1$,
the second term vanishes and we obtain 
\begin{align}
\mu = \lambda\sqrt{m}\E\left\{\frac{1}{\|\bZ\|_2}-\frac{\Re(Z_1)^2}{\|\bZ\|_2^3}\right\} \, \mu +o(\mu)\, .
\end{align}
We therefore get the critical point $\lambda_{\rm c}(m)$ by setting to
$1$ the coefficient of  $\mu$ above. In the real case, we get
\begin{align}
\lambda^{\RS}_{\rm c}(m)^{-1} &= \sqrt{m}\Big(1-\frac{1}{m}\Big)\,
\E\big\{1/\|\bZ\|_2\big\}\\
& = \sqrt{\frac{2}{m}} \frac{\Gamma((m+1)/2)}{\Gamma(m/2)}\, .
\end{align}
In the complex case 
\begin{align}
\lambda^{\RS}_{\rm c}(m)^{-1} &= \sqrt{m}\Big(1-\frac{1}{2m}\Big)\,
\E\big\{1/\|\bZ\|_2\big\}\\
& = \sqrt{\frac{1}{m}} \frac{\Gamma(m+(1/2))}{\Gamma(m)}\, .
\end{align}

Summarizing the (replica symmetric) critical point is 
\begin{align}
\lambda^{\RS}_{\rm c}(m) = \begin{cases}
\sqrt{m/2}\, \Gamma(m/2)\Gamma((m+1)/2)^{-1}\;\;\;\; &\mbox{for the real
  case,}\\
\sqrt{m}\, \Gamma(m)\Gamma(m+(1/2))^{-1}\;\;\; &\mbox{for the complex
  case.}\\
\end{cases}
\end{align}
In  particular, for $m=1$ we recover $\lambda^{\RS}_{\rm c}(1) =
\sqrt{\pi/2}$ for the real case, and $\lambda^{\RS}_{\rm c}(1) =
2/\sqrt{\pi}$ for the complex  case. These are the values derived in
Section \ref{eq:MLDense}. For large $m$, we get 
\begin{align}
\lambda^{\RS}_{\rm c}(m)=  1+\frac{1}{4\sg m} + O(m^{-2})\, ,
\end{align}
with $\sg=1$ (real case), or $\sg =2$ (complex case).

Let us emphasize once more: we do not expect the replica
  symmetric calculation above to be exact, but only an excellent approximation. In other words, \emph{for any
  bounded $m$, we expect $\lambda_{\rm c}(m)\approx\lambda^{\RS}_{\rm
    c}(m)$ but
 $\lambda_{\rm c}(m)\neq\lambda^{\RS}_{\rm c}(m)$.} However,
as $m\to\infty$ the problem becomes convex, and hence we expect
$\lim_{m\to \infty}|\lambda_{\rm c}(m)-\lambda^{\RS}_{\rm c}(m)| = 0$.
Hence
\begin{align}
\lsdp = \lim_{m\to\infty}\lambda_{\rm c}(m) = 1\, .
\end{align}

\subsubsection{SDP: $m\to\infty$ and $\beta \to\infty$}
\label{sec:SDP_Synchro}

In the limit $m\to\infty$, Eqs.~(\ref{muEq}) to (\ref{c1Eq}) simplify
somewhat.
We set $q_1=q$ and eliminate $q_0$ using Eq.~(\ref{eq:Normaliz}). Applying the
law of large numbers, the equation for $\rho$ reads
\begin{align}
1 &=  \frac{| \mu + \sqrt{q}\,Z_1|^2}{(\rho+r)^2} +
\frac{1-q}{\rho^2} \,. \label{eq:minfty_RhoEq}
\end{align}
As a consequence, $\rho$ becomes independent of $Z_2$. Hence, 
Eqs.~(\ref{muEq}) to (\ref{c1Eq}) reduce to
\begin{align}
\mu &= \lambda\, \E\left[ \frac{\mu + \sqrt{q}\,\Re(Z_1)}{\rho+r}\right]\,,\label{eq:minfty_1}\\
q &= \E\left[\frac{|\mu + \sqrt{q}\,Z_1 |^2}{(\rho+r)^2}\right]\,,\label{eq:minfty_2}\\
r &=
\E\left[ \frac{1}{\rho}-\frac{\mu\,\Re(Z_1)}{\sqrt{q}(\rho+r)}
- \frac{|Z_1|^2}{\rho+r}\right]\,, \label{eq:minfty_3}\\
1 &=\E\left\{\frac{1}{\rho^2}\right\}\, . \label{eq:RhoExp}
\end{align}
Denoting by $\mu_*(\lambda)$ and $q_*(\lambda)$ the solutions to the above equations, we have
\begin{align}
\MSE(\hbx^{(\infty,\infty)}) =
1-\frac{\mu_*(\lambda)^2}{\lambda^2 q_{1,*}(\lambda)}\, ,\label{eq:SDP-MSE}
\end{align}
Further, 
\begin{align}
\Ove(\hbx^{(\infty,\infty)}) 
= 1-2\Phi\Big(-\frac{\mu_*(\lambda)}{\sqrt{q_{1,*}(\lambda)}}\Big)\,.\label{eq:SDP-OL}
\end{align}

We solution of the above equations displays a phase transition at the
critical point $\lsdp =1$, which we next characterize.

\vspace{0.5cm}

\noindent{\bf Spin glass phase and critical point.}
The spin-glass phase corresponds to a symmetric solution $\mu=q = r =
0$. 

In order to investigate the critical behavior, we expand the
equations (\ref{eq:minfty_1}) to (\ref{eq:minfty_3}) for $\lambda = 1+\eps$,
$\eps\ll 1$. To leading order in $\eps$, we get the following solution
\begin{align}
\mu & = 2\, \eps^{3/2} + O(\eps^{2})\, ,\\
q & = \sg\, \eps^2 +O(\eps^{5/2})\, ,\\
r & = \eps +O(\eps^{3/2})\, ,\\
\rho & = 1+\eps^2\,(|Z_1|^2-1) + o(\eps^{3/2})\, .
\end{align}
Hence
\begin{align}
\MSE(\xsdp) = 1-\frac{\mu^2}{q} = 1-\frac{4}{\sg}\,\eps
  +O(\eps^{3/2})\, .
\end{align}

To check the above perturbative solution, note that   expanding the
denominator of Eq.~(\ref{eq:minfty_2}) and using $\rho = 1+O(\eps^2)$, we get
\begin{align}
q &= \E |\lambda \mu +\sqrt{q}Z_1|^2(1-2r+O(\eps^2))\, ,\\
\Leftrightarrow \;\;\;\; q &= \lambda^2\mu^2+q(1-2r) +O(\eps^4)\\
\Leftrightarrow \;\; \mu^2 &= 2rq +O(\eps^4)\, .
\end{align}
Multiplying Eq.~(\ref{eq:minfty_RhoEq}) by $\rho^2$ and expanding the
right-hand side, we get
\begin{align}
\rho^2 &= 1-q+q|Z_1|^2\left(\frac{\rho}{r+\rho}\right)^2
         +O(\mu\sqrt{q})\\
\Leftrightarrow \;\; \rho & = 1+\frac{q}{2}(|Z_1|^2-1) +
                            O(\eps^{5/2})\, .
\end{align}
Finally, expanding Eq.~(\ref{eq:minfty_3}), we get
\begin{align}
r
=1-\frac{\mu}{\sqrt{q}}\E\Big\{\frac{Z_1}{1+b+(q/2)(|Z_1|^2-1)}\Big\}-
\E\Big\{\frac{|Z_1|^2}{1+b+(q/2)(|Z_1|^2-1)}
\Big\}+O(\eps^{5/2})\, ,
\label{eq4b}
\end{align}
whence
\begin{align}
(\E|Z_1|^4-1)\, q= 2\,r^2+O(\eps^{5/2})\, .
\end{align}
Finally, expanding Eq.~(\ref{eq:minfty_1}) we get $r = \eps+O(\eps^{3/2})$.

\subsection{Free energy and energy}
\label{sec:FreeEnergy}

It is easier to derive the free energy using the replica method.
This also give an independent verification of the cavity calculations
in the previous section.

\subsubsection{Replica calculation}

In this section, apply the replica method to compute the free energy of model
(\ref{eq:GibbsModel}). Our aim is to compute asymptotics for the
partition function
\begin{align}
Z = \int\, \exp\Big\{2m\beta \sum_{i<j} \Re(Y_{ij}
  \<\bsigma_i,\bsigma_j\>)\Big\}\, \prod_{i=1}^np_0(\de \bsigma_i)\, ,
\end{align}
where  we recall that $p_0(\de\bsigma_i)$ is the uniform measure over $\bsigma_i\in
S^{m-1}$. The $k$-th moment is given by
\begin{align}
\E\{Z^k\} = \int \prod_{i<j}\E\exp\Big\{2\beta m\sum_{a=1}^k \Re(Y_{ij}
  \<\bsigma^a_i,\bsigma_j^a\>)\Big\} \, \bap_0(\de \bsigma)\, ,
\end{align}
where we introduced replicas $\bsigma_i^1,\dots ,\bsigma_i^k\in
S^{m-1}$, along with the notation $\bap_0(\de \bsigma) \equiv
\prod_{i=1}^n\prod_{a=1}^kp_0(\de \bsigma_i^a)$.
Taking the expectation over $Y_{ij} = (\lambda/n)+W_{ij}$, we get
\begin{align}
\E\{Z^k\} &= \int \prod_{i<j}\exp\Big\{\frac{2\beta m\lambda}{n} \sum_{a=1}^k\Re
  \<\bsigma^a_i,\bsigma_j^a\>+\frac{2\beta^2m^2}{\sg
  n}\Big|\sum_{a=1}^k \<\bsigma^a_i,\bsigma_j^a\>\Big|^2\Big\} \, 
\bap_0(\de \bsigma)\\
&\doteq \int \exp\Big\{\frac{\beta m\lambda}{n} \sum_{a=1}^k
\Big\|\sum_{i=1}^n\bsigma^a_i\Big\|_2^2+\frac{\beta^2m^2}{\sg
  n}\sum_{a,b=1}^k \Big\|\sum_{i=1}^n\bsigma^a_i(\bsigma_i^b)^*\Big\|_F^2\Big\} \, 
\bap_0(\de \bsigma)\, .
\end{align}
We next use the identity
\begin{align}
\exp\Big\{\frac{1}{2}\zeta^2 \|\bv\|^2\Big\} \doteq \int \exp\Big\{-\frac{\|\bw\|^2}{2\zeta^2} + \Re\<\bw,\bv\>\Big\} \de \bw\,, 
\end{align}
where the vectors $\bv$ and $\bw$ take their entries in $\field$.
We apply this identity and introduce Gaussian integrals over the variables $\bmu_a\in \field^m$,
$\bQ_{ab}\in\field^{m\times m}$ (with $\bQ_{ba}= \bQ_{ab}^*$)
\begin{align}
\E\{Z^k\} &\doteq \int \exp\Big\{
-\frac{\beta m n}{\lambda}\sum_{a=1}^k \|\bmu_a\|_2^2 +2\beta
            m\sum_{a=1}^k \sum_{i=1}^n\Re\<\bmu_a,\bsigma_i^a\>\\
&\phantom{AAAAA}-\frac{\beta^2m^2n}{\sg}
            \sum_{a,b=1}^k\Tr(\bQ_{ab}\bQ_{ab}^*)+\frac{2\beta^2
            m^2}{\sg}\sum_{a,b=1}^k\sum_{i=1}^n
\Re\<\bsigma^a_i,\bQ_{ab}\bsigma^b_i\>
\Big\}\;\bap_0(\de \bsigma)\, \de\bQ\de\bmu\, .\nonumber
\end{align}
The final formula for the free energy density is obtained by
integrating
with respect to $\bsigma$ (now the integrand is in product form) and
taking the saddle point in $\bQ$, $\bmu$, and is reported in the next
section, see Eq.~(\ref{eq:ReplicatedFreeEnergy}) below.

\subsubsection{Non-zero temperature ($\beta<\infty$)}

The final result of the calculations in the previous section is obtaining the moments
\begin{align}
\lim_{n\to\infty}\frac{1}{n}\log \E\{Z^k\} 
& = {\rm ext}_{\bQ,\bmu} S_k(\bQ,\bmu)\, ,\\
S_k(\bQ,\bmu)  &= -\frac{\beta m}{\lambda}\sum_{a=1}^k\|\bmu_a\|_2^2-
\frac{\beta^2 m^2}{\sg}\sum_{a,b=1}^k\Tr(\bQ_{ab}\bQ_{ab}^*) + \log
                 W_k(\bQ,\bmu)\, .\label{eq:ReplicatedFreeEnergy}
\end{align}

Here, for each $a, b\in\{1,2,\dots,k\}$, $\bmu_a\in \reals^m$,
$\bQ_{a,b}\in\field^{m\times m}$ with $\bQ_{b,a} =\bQ_{a,b}^*$. In
particular $\bQ_{a,a}$ are Hermitian (or symmetric) matrices. 
The notation ${\rm ext}_{\bQ,\bmu}$ indicates that we need to take a
stationary point over $\bQ_{ab}, \bmu_a$. As usual in the replica
method, this will be a local minimum over some of the parameters, and
local maximum over the others.
Finally, the one-site
replicated partition function is 
\begin{align}
W_k(\bQ,\bmu)\equiv \int \!\exp\Big\{2\beta\,
m\sum_{a=1}^k\Re \<\bmu_a,\bsigma_a\>
+\frac{2\beta^2 m^2}{\sg}\sum_{a,b=1}^k\Re\<\bsigma_a,\bQ_{a,b} \bsigma_b\>
\Big\}\; \prod_{a=1}^kp_0(\de\bsigma_a)\, ,
\end{align}
where we used the following identity in its derivation
\begin{align}
\Big(\int f(\bsigma) p_0(\de\bsigma)\Big)^n \equiv \int f(\bsigma_1) f(\bsigma_2) \dotsc f(\bsigma_n)\,
 p_0(\de\bsigma_1)  \dotsc p_0(\de\bsigma_n)\,.
\end{align}
Recall that $\sg=1$ for $\Group = \Z_2$ (real case), and $\sg=2$ for $\Group =
U(1)$  (complex case).

\vspace{0.5cm}

\noindent{\bf Replica-symmetric free energy.} The  replica-symmetric
(RS) ansatz is
\begin{align}
\bmu_a & = \, \bmu\, , \\
\bQ_{a,b} & = \begin{cases}
\bQ + \frac{1}{\beta m}\, \bC& \mbox{ if $a=b$,}\\
\bQ & \mbox{ if $a\neq b$.}
\end{cases}
\end{align}
It follows from the above that $\bC$, $\bQ\in\field^{m\times m}$
must be Hermitian (symmetric) matrices.
We next compute the RS free energy functional
\begin{align}
\phi(\bQ,\bC,\bmu) = \lim_{k\to 0} \frac{1}{k} S_k(\bQ,\bmu)\,.
\end{align}
By the RS ansatz we have
\begin{align}
\lim_{k\to 0} \frac{\beta m}{k\lambda} \sum_{a=1}^k \|\bmu_a\|_2^2 &= \frac{\beta m}{\lambda} \|\bmu\|_2^2\,\label{eq:term1}\\
\lim_{k\to 0} \frac{\beta^2 m^2}{k \sg}\sum_{a,b=1}^k\Tr(\bQ_{ab}\bQ_{ab}^*)  &=
\lim_{k\to 0} \frac{\beta^2 m^2}{\sg}\Big\{\Tr\Big(\Big(\bQ+\frac{1}{\beta m}\bC\Big)^2\Big)  + (k-1) \Tr(\bQ^2) \Big\}\nonumber\\
& = \frac{\beta^2 m^2}{\sg}\Big\{  \Tr\Big(\Big(\bQ+\frac{1}{\beta m}\bC\Big)^2\Big) -\Tr(\bQ^2) \Big\}\label{eq:term2}
\end{align}
For  computing the third term, we use the following identity. For a fixed arbitrary vector $\bv$,
\begin{align}
\E\Big\{ \exp(\Re\<\bv,\bxi\>)\Big\} = \exp\Big\{\Re\<\bv,\bmu\> + \frac{1}{2\sg} \Re\<\bv,\bQ\bv\> \Big\}\,,
\end{align}
where $\bxi\sim \normal(\bmu,\bQ)$ in the real case, and $\bxi\sim \cnormal(\bmu,\bQ)$ in the complex case. Further, the expectation $\E$ is 
with respect to $\bxi$.

Applying this identity, we write
\begin{align}
&W_k(\bQ,\bmu) \nonumber\\
&= \int \!\exp\Big\{2\beta
m\, \Re \<\bmu,\sum_{a=1}^k \bsigma_a\>
+\frac{2\beta m}{\sg} \sum_{a=1}^k \Re\<\bsigma_a, \bC\bsigma_a\>
+\frac{2\beta^2 m^2}{\sg}\Re\Big\<\sum_{a=1}^k \bsigma_a,\bQ \sum_{a=1}^k \bsigma_a\Big\>
\Big\}\; \prod_{a=1}^kp_0(\de\bsigma_a)\nonumber\\
&= \int \! \E \exp\Big\{\frac{2\beta m}{\sg} \sum_{a=1}^k \Re\<\bsigma_a, \bC\bsigma_a\>
+2\beta m \sum_{a=1}^k \Re\<\bsigma_a,\bxi\>\Big\}\prod_{a=1}^kp_0(\de\bsigma_a)\nonumber\\
&= \E \Big(\int\! \exp\Big\{\frac{2\beta m}{\sg} \Re\<\bsigma, \bC\bsigma\>
+2\beta m \Re\<\bsigma,\bxi\>\Big\}\, p_0(\de\bsigma)\Big)^k\,.
\end{align}
Hence,
\begin{align}
\lim_{k\to 0} \frac{1}{k} \log W_k(\bQ,\bmu) = 
\E\log \Big(\int \exp\Big\{2\beta m\Re\<\bxi,\bsigma\>
+\frac{2\beta m}{\sg}\<\bsigma,\bC\bsigma\>\Big\} p_0(\de \bsigma)\Big)\,.\label{eq:term3}
\end{align}
Combining Eqs.~\eqref{eq:term1},~\eqref{eq:term2} and~\eqref{eq:term3} we arrive at
\begin{align}
\phi(\bQ,\bC,\bmu)
&= -\frac{\beta m}{\lambda} \|\bmu\|_2^2
-\frac{1}{\sg}\beta^2 m^2 \Big\{\Tr\Big(\Big(\bQ+\frac{1}{\beta m}\bC\Big)^2\Big)-\Tr(\bQ^2)\Big\}\nonumber\\
&+\E \log \Big(\int \exp\Big\{2\beta m\Re\<\bxi,\bsigma\>
  +\frac{2\beta m}{\sg}\<\bsigma,\bC\bsigma\>\Big\}
p_0(\de \bsigma)\Big)\, ,\\
&\bxi\sim\normal(\bmu,\bQ)\, .\nonumber
\end{align}
In the complex case, the last line should be interpreted as $\bxi~\sim\cnormal(\bmu,\bQ)$.
Differentiating this expression against $\bmu,\bC,\bQ$ we recover
Eqs.~(\ref{eq:CavityGen_1}) to (\ref{eq:CavityGen_3}) as saddle point
conditions.

\subsubsection{Zero temperature ($\beta\to\infty$)}

As $\beta\to\infty$, the free energy behaves as
\begin{align}
\phi(\bQ,\bC,\bmu) = 2 \beta m\, u(\bQ,\bC,\bmu) + o(\beta)\, ,
\end{align}
where $u(\bQ,\bC,\bmu)$ is the replica-symmetric ground state energy
\begin{align}
u(\bQ,\bC,\bmu) & = -\frac{1}{2\lambda}\|\bmu\|_2^2 -\frac{1}{\sg}\Tr(\bC\bQ)+
\E \max_{\bsigma\in S^{m-1}}\big[\Re\<\bxi,\bsigma\>
                    +\frac{1}{\sg}\<\bsigma,\bC\bsigma\>\big]\, ,\\
&\bxi\sim\normal(\bmu,\bQ)\, .\nonumber
\end{align}
Let us stress that expectation is with respect to $\bxi$. Denote by $\bsm=\bsm(\bxi,\bC)$ the solution of the above maximization
problem. It is immediate to see that this is given by
\begin{align}
\bsm  &= \big(\rho\id-(2/\sg)\bC\big)^{-1} \bxi\, ,
\end{align}
where $\rho=\rho(\bxi,\bC)\in\reals$ is a Lagrange multiplier
determined by solving the equation
\begin{align}
\<\bxi,\big(\rho\id-(2/\sg)\bC\big)^{-2} \bxi\> = 1\, .
\end{align}

The equations for $\bmu$ and $\bQ$ are immediate by taking the
$\beta\to\infty$ limit on Eqs.~(\ref{eq:CavityGen_1}), (\ref{eq:CavityGen_2}).  In zero temperature,  
measure $\nu_{\bxi,\bC}$ concentrates around $\bsm(\bxi,\bC)$.
\begin{align}
\bmu & = \E\{\Re\bsm(\bxi,\bC)\}\,,\label{bmuEq}\\
\bQ & = \E\{\bsm(\bxi,\bC)\bsm(\bxi,\bC)^*\}\,. \label{bQEq}
\end{align}
Equivalently, we obtain the above equations by differentiating $u(\bQ,\bC,\bmu)$
with respect to $\bmu$ and $\bC$, as follows. We write $\bsm = \bsm(\bxi,\bC)$ to lighten the notation.  Since $\rho$ is a Lagrange multiplier, we have
$$\nabla_{\bsm} \Big( \Re\<\bxi,\bsm\> + (1/\sg) \<\bsm,\bC,\bsm\>\Big) = \rho \bsm\,.$$
Hence,
\begin{align}
\nabla_\bC u &= -\frac{1}{\sg}\bQ + \E\Big\{\frac{1}{\sg} \bsm \bsm^* + \rho \sum_{\ell=1}^m \bsigma_{{\rm M},\ell} \nabla_\bC \bsigma_{{\rm M},\ell}\Big\}\nonumber\\
&= -\frac{1}{\sg}\bQ + \E\Big\{\frac{1}{\sg} \bsm \bsm^*\Big\}\,,
\end{align}
where the second equation follows from the constraint $\|\bsm\|_2 = 1$.

We next substitute $\bxi = \bmu + \bQ^{1/2}\bZ$. By a similar calculation, we have
\begin{align}
\nabla_\bmu u &= -\frac{1}{\lambda}\bmu + \E\{\Re \bsm\}\,.\\
\nabla_{\bQ^{1/2}} u &= -[\bC\bQ^{1/2}]_s +\frac{\sg}{2}\, \E \big\{[\bZ\,\bsm(\bxi,\bC)^*]_s\big\}\,.\label{bCEq}
\end{align}
Recall that expectation $\E$ is with respect to $\bxi\sim \normal(\bmu,\bQ)$ or,
equivalently, $\bZ\sim\normal(0,\id_{m})$. Further $[\bA]_s$ denotes
the symmetric (Hermitian) part of the matrix $\bA$, i.e. $[\bA]_s =
(\bA+\bA^*)/2$. 

Using the ansatz (\ref{eq:AnsatzOn}), we recover
Eqs.~(\ref{muEq}) to (\ref{c1Eq}).  Specifically,
Eq~\eqref{muEq} follows readily from Eq.~\eqref{bmuEq}, restricting to the $(1,1)$ entry and plugging in for $\bsm$ from Eq.~\eqref{bsmEq}.
Also, Eqs.~\eqref{q1Eq} and~\eqref{q0Eq} follow from Eq.~\eqref{bQEq}, restricting to $(1,1)$ and $(2,2)$ entries, respectively. 
Derivation of Eq.~\eqref{c1Eq} requires more care. Note that since $\nu_{\xi,\bC}$, given by~\eqref{eq:Generalnu}, is a measure on the unit
sphere, the matrix $\bC$ is only defined up to a diagonal shift. Let $\sg \eta/2$ denote the slack shift parameter. The ansatz~\eqref{eq:AnsatzOn}
for $\bC$ then becomes 
\begin{align*}
\bC= \left(\begin{matrix}
-\sg (r+\eta)/2 & & & & &\\
&-\sg\eta/2& & & &\\
&&\cdot&&&\\
&&&\cdot&&\\
&&&&\cdot&\\
&&&&&-\sg\eta/2
\end{matrix}\right)\,.
\end{align*}
We set $\nabla_{\bQ^{1/2}} u =0$. Applying Eq.~\eqref{bCEq}, this results in the following two equations for $\bC_{11}$ and $\bC_{22}$:
\begin{align}
\frac{-\sg}{2}(r+\eta)\sqrt{q_1} &= \frac{\sg}{2} \E\Big\{\frac{\mu \Re(Z_1)}{\rho+r}+ \frac{\sqrt{q_1} |Z_1|^2}{\rho+r} \Big\}\,,\label{eq:eta1}\\
-\frac{\sg}{2}\eta \sqrt{q_0} &= \frac{\sg}{2} \E\Big\{ \frac{\sqrt{q_0}|Z_2|^2}{\rho}\Big\}\,.\label{eq:eta2} 
\end{align}
Solving for $\eta$ from Eq.~\eqref{eq:eta2} and substituting for that in Eq.~\eqref{eq:eta1}, we obtain Eq.~\eqref{c1Eq}.

\subsection{On the maximum likelihood phase transition}

In Section \ref{eq:MLDense} we computed the replica symmetric
approximation for the phase tranition of the maximum likelihood
estimator. We obtained
$\lambda_{c}^{\mbox{\tiny{ML,RS}}}(\integers_2)= \sqrt{\pi/2}$ (real case),
$\lambda_{c}^{\mbox{\tiny{ML,RS}}}(U(1))= 2/\sqrt{\pi}$ (complex
case). This is somewhat surprising because it suggests that the
maximum likelihood estimator has a worse threshold than the Bayes
optimal estimator. 

It turns out that this is an artifact of the replica symmetric approximation and instead
\begin{align}
\lambda_{c}^{\mbox{\tiny{ML}}}(\integers_2) =
  \lambda_{c}^{\mbox{\tiny{ML}}}(U(1)) =1\, .
\end{align}
We next outline the heuristic argument that support this claim. (For
the sake
of simplicity, we will consider the $\integers_2$ case.)

For a given noise realization $\bW$, the maximum likelihood estimator is
\begin{align}
\xml(\bY) =
  \arg\max_{\bx\{+1,-1\}^n}\left\{\frac{\lambda}{2n}\<\bx,\bxz\>^2+ \frac{1}{2}\<\bx,\bW\bx\>\right\}\,
  .
\end{align}
We then define the correlation
\begin{align}
M(\bW) \equiv \frac{1}{n}\big|\<\bxz,\xml(\bY)\>\big|\, ,
\end{align}
and recall that $\Ove_n(\xml) = \E\{M(\bW)\}$. We want to show that
$M(\bW)$ is (with high probability) bounded away from $0$ for $\lambda>1$.
Setting, without loss of generality, $\bxz = \bone$, we have
\begin{align}
M(\bW) &= \arg\max_{m\in [0,1]}\Big\{ F_{\bW,n}(m) +\frac{\lambda}{2}\, m^2\Big\}\,
  ,\\
F_{\bW,n}(m)&\equiv \max_{\bx\in\{+1,-1\}^n}\Big\{\frac{1}{2n}\<\bx,\bW\bx\>:\;
               \frac1n \sum_{i=1}^n x_i \in \big[m,m+1/n\big)\Big\}\, .
\end{align}
We expect $\lim_{n\to\infty}F_{\bW,n}(m) \equiv F(m)$ to exist and to
be non-random.
This implies that the asymptotic overlap is given by
\begin{align}
\Ove(\xml)&= \arg\max_{m\in [0,1]} \Big\{ F(m) +\frac{\lambda}{2}\, m^2\Big\}\,
\end{align}
By symmetry we have $F(m) = F(-m)$. Assuming $m\mapsto F(m)$ to be
differentiable, this implies  $F'(0) = 0$.
Hence $m=0$ is a local maximum for $\lambda < -F''(0)$ and a local
minimum for $\lambda> -F''(0)$. Since at $\lambda=0$ we obviously have
$\Ove(\xml)=0$, $F''(0)<0$.  Further, if $m=0$ is a local minimum, we
necessarily have $\Ove(\xml)>0$.  Hence $\lml \le -F''(0)$.

On the other hand, we know that we cannot estimate $\bxz$ with
non-vanishing overlap  for $\lambda<1$. This is a consequence --for instance--
of \cite[Theorem 4.3]{deshpande2015asymptotic} or can, in alternative,
be proved directly using the technique of \cite{montanari2014limitation}.
This implies that $\lml\ge 1$. Summarizing, we have
\begin{align}
1\le \lml\le -F''(0)\, .
\end{align}

We next claim that earlier work on the Sherrington-Kirkpatrick model
implies $F''(0) = -1$, thus yielding $\lml=1$.
Indeed, alternative expressions can be obtained by studying the modified
problem
\begin{align}
\hF_{\bW,n}(h)&\equiv
                \frac{1}{n}\max_{\bx\in\{+1,-1\}^n}\left\{\frac{1}{2}\<\bx,\bW\bx\>+h\<\bone,\bx\>\right\}\,
                ,
\end{align}
where $h$ is an added magnetic field.
Then, we have $\lim_{n\to\infty}\hF_{\bW,n}(h)=\hF(h)$, the Legendre
transform of $F$, and we get the alternative upper bound
\begin{align}
\lml &\le -F''(0) =\frac{1}{\hF''(0)}\, .
\end{align}
Note that $\hF(h)$ is the zero-temperature free energy density of the
Sherrington-Kirkpatrick model in a magnetic field $h$
\cite{SpinGlass}, whose $n\to\infty$ limit exists by
\cite{guerra2002thermodynamic}. 
Using well-known thermodynamic identities, we get 
\begin{align}
\lml \le -F''(0) = &=   \lim_{\beta\to\infty}\lim_{n\to\infty}\frac{1}{\chi(\beta,h=0+)} \\
& =
  \lim_{\beta\to\infty}\lim_{n\to\infty}\frac{1}{\beta\E_{\beta,h= 0+}\{1-Q\}}
  \, ,
\end{align}
where $\chi(\beta,h)$ is the magnetic susceptibility of the Sherrington-Kirkpatrick
model at inverse temperature $\beta$, and magnetic field $h$, and $Q$ is the  
random overlap.

To the best of our knowledge, the above connection between response to
a magneric field, and couplings with non-zero mean was first described
by G\'erard Toulouse\footnote{In \cite{toulouse1980mean}, this
  argument was put forward within the context of the so-called
  Parisi-Toulouse (PaT) scaling hypothesis. Let us emphasize that here we
  are not assuming PaT to hold (and indeed, it has been convincingly
  shown that PaT is not correct, albeit an excellent approximation,
  see e.g. \cite{crisanti2003parisi}).} in \cite{toulouse1980mean}. 

The relation $\beta\E_{\beta,h=0+}\{1-Q\}=1$ was derived  in
\cite{sommers1983properties} from Sompolisky's formulation of mean
field theory. This result is also confirmed by recent high-precision
numerical approximations of the  $\infty$-RSB solution of the
Sherrington-Kirkpatrick model \cite{oppermann2008universality}.
%
%*************
%
\section{Analysis of PCA estimator for synchronization problem}

Here, we study the PCA estimator for the synchronization problem. Recall the observation model
\begin{align}
\bY = \frac{\lambda}{n}\bx_0\bx_0^*+\bW\,,
\end{align}
with $\bx_0 \in \field^n$ and $\|\bx_0\|=n$, where $\|\cdot\|$ refers to the norm on $\field$.

Let $\bv_1(Y)$ denote the leading eigenvector of $\bY$. The PCA estimator $\xpca$ is defined as
\begin{align}
\xpca(\bY) \equiv \sqrt{n}\, \cpca(\lambda) \bv_1(\bY)\,,
\end{align}
with $\cpca$ a certain scaling factor discussed below. 

In order to characterize the error of $\xpca$, we use a simplified version of the main theorem in~\cite{capitaine2009largest}.
\begin{lemma}[\cite{capitaine2009largest}]
Let $\bY = \lambda \bv_0 \bv_0^* + \bW$ be a rank-one deformation of the Gaussian symmetric matrix $\bW$, with
$W_{ij}\sim \normal(0,1/n)$ independent for $i<j$, and $\|\bv_0\| =1$. Then, we have, almost surely
\begin{align}
\begin{split}
\lim_{n\to \infty} |\<\bv_1(\bY), \bv_0\>| = \begin{cases}
0 & \text{ if } \lambda\le 1\,,\\
\sqrt{1-\lambda^{-2}} & \text{ if }\lambda >1\,.
\end{cases}
\end{split}
\end{align}
Further, letting $\lambda_1(\bY)$ be the top eigenvalue of $\bY$, the following holds true almost surely
\begin{align}
\begin{split}
\lim_{n\to \infty} \lambda_1(\bY) = \begin{cases}
2 & \text{ if } \lambda\le 1\,,\\
\lambda+1/\lambda & \text{ if }\lambda >1\,.
\end{cases}
\end{split}
\end{align}
\end{lemma}
Applying this lemma, we compute $\MSE(\xpca;c)$ as follows
\begin{align}
\begin{split}
\MSE(\xpca;c) &= \lim_{n\to \infty} \Big\{1-\frac{2c}{n} \E\{\Re \<\bx_0,\xpca\>\} + \frac{c^2}{n}\|\xpca\|^2 \Big\}\\
&= 1- 2c\sqrt{1-\lambda^{-2}} + c^2\,,
\end{split}
\end{align}
which is optimized for $c =\cpca(\lambda) \equiv \sqrt{\max(1-\lambda^{-2},0)}$. Note that this choice can be written in terms of $\lambda_1(\bY)$ as well and so knowledge of $\lambda$ is not required. We then obtain
\begin{align}
\MSE(\xpca(\bY);\cpca(\lambda)) = \min(1,\lambda^{-2})\,.
\end{align}
%
%
%*************
%
\section{Analytical results for community detection}

In this Section we use the cavity method to analyze the semidefinite
programming approach to community detection. We refer, for instance, to \cite{MezardMontanari}
for general background on the cavity method for sparse graphs. Also,
see \cite{banavar1987graph,sherrington1987graph} for early statistical
mechanics work on the
related graph bisection problem.

Recall  (from the main text) that we are interested in the hidden partition model. Namely, 
consider a random graph $G_n = (V_n,E_n)$ over vertex set $V_n=[n]$, generated according to the following distribution.
We let $\bxz\in\{+1,-1\}^n$ be uniformly random:
this vector contains the vertex labels (equivalently, it encodes a
partition of the vertex set $V = V_+\cup V_-$, in the obvious way). Conditional on $\bxz$, edges
are independent with distribution
\begin{align}
\prob\big\{(i,j)\in E_n\big|\bxz\big\} = 
\begin{cases}
a/n & \mbox{ if $x_{0,i}x_{0,j} =+1$,}\\
b/n & \mbox{ if $x_{0,i}x_{0,j} =-1$.}
\end{cases}
\end{align}
As explained in the main text, we tackle this problem via the
semidefinite relaxation
\begin{align}
\begin{split}\label{eq:SDP_Graph}
\mbox{maximize} &\;\;\;\;\; \sum_{(i,j)\in E_n} X_{ij}\, ,\\
\mbox{subject to} &\;\;\;\;\; \bX\succeq 0\, ,\\
&\;\;\;\;\; \bX\bone=\bzero\, ,\;\;\;\;\; X_{ii} = 1\;\; \forall i\in [n]\, .
\end{split}
\end{align}
For our analysis, we use the non-convex formulation  
\begin{align}
\begin{split}
\mbox{maximize}  &\;\;\;\;\; \sum_{(i,j)\in E_n}\<\bsigma_i,\bsigma_j\>\, ,\\
\mbox{subject to} &\;\;\;\;\; \sum_{i=1}^n\bsigma_i  =\bzero\, ,\\
                    &\;\;\;\;\; \bsigma_i\in S^{m-1}\subseteq \reals^m\;\; \forall i\in [n]\, .
\end{split}
\end{align}
This is equivalent to the above SDP provided $m\ge n$. 
Note that, throughout this section, the spin variables $\bsigma_i$ are
\emph{real} vectors. 

We  introduce the following Boltzmann-Gibbs distribution 
\begin{align}
p_{\beta,m}(\de \bsigma)= \frac{1}{Z_{G}(\beta,m)}\exp\Big\{2m\beta
  \sum_{(i,j)\in E} 
  \<\bsigma_i,\bsigma_j\>\Big\}\, \bap_0(\de \bsigma)\, . \label{eq:GibbsSparse}
\end{align}
Here $\bap_0(\de\bsigma_i)$ is the uniform measure over 
$\bsigma = (\bsigma_1,\bsigma_2,\dots,\bsigma_n)$ with 
$\bsigma_i\in S^{m-1}$ and $\sum_{i=1}^n\bsigma_i = \bzero$.
In order to extract information about the SDP (\ref{eq:SDP_Graph}),
we take the limits $m\to\infty$, $\beta\to\infty$ \emph{after}
$n\to\infty$.

As $n\to\infty$, the graph $G_n$ converges locally to a rooted multi-type 
Galton-Watson tree with vertices of type $+$ (corresponding to
$x_{0,i}=+1$) or $-$ (corresponding to $x_{0,i}=-1$). Each vertex has
$\Poisson(a/2)$ offsprings of the same type, and $\Poisson(b/2)$
offsprings of the other type (see, e.g. \cite{dembo2010gibbs} for background on
local weak convergence in statistical mechanics).

We write the sum-product fixed point equations to compute the marginals at different nodes.
\begin{align}
\nu_{i\to j}(\de\bsigma_i) \normeq \prod_{\ell\in {\partial i}\setminus\{j\}}\int \, e^{ 2\beta
  m \<\bsigma_i,\bsigma_\ell\>} \, \nu_{\ell\to i}(\de\bsigma_\ell)\,,
\end{align}
where $\nu_{i\to j}$ are the messages associated to the directed edges of the graph.
The marginal $\nu_0(\de\bsigma_0)$, for an arbitrary node $0$, is given by
\begin{align}
\nu_0(\de\bsigma_0) \normeq \prod_{i\in\partial 0}\int \, e^{ 2\beta
  m \<\bsigma_0,\bsigma_i\>} \, \nu_{i\to 0}(\de\bsigma_i)\,.
\end{align}
We rewrite the above equations from another perspective. We designate node $0$ as the root of the tree and denote
its neighbors by $\{1,2,\dotsc,k\}$.
Let $T_i$ be the subtree rooted at node $i$ and induced by its descendants. We call $\nu_i(\de\bsigma_i)$
the marginal for $\bsigma_i$ w.r.t the graphical model in the subtree $T_i$. Replica symmetric cavity
equations relate the marginal $\nu_0(\de\bsigma_0)$ to the marginals at the descendant subtrees, i.e., $\nu_i(\de\bsigma_i)$.
Note that, in the above notation, $\nu_i(\de\bsigma_i) \equiv \nu_{i\to 0}(\de\bsigma_i)$ and therefore we obtain
\begin{align}
\nu_0(\de\bsigma_0) &\normeq \prod_{i=1}^k\hnu_i(\bsigma_0)\,. \label{eq:CavitySparse}\\
\hnu_i(\bsigma_0) &\normeq \int \, e^{2\beta
  m \<\bsigma_0,\bsigma_i\>} \, \nu_i(\de\bsigma_i)\, .
\end{align}
(The measures $\nu_i(\de\bsigma_i)$ are probability measures over
$S^{m-1}$ and the right-hand side should be interpreted as a density
with respect to the uniform measure on $S^{m-1}$.)

We will use the notation of Eq.~(\ref{eq:CavitySparse}) but both
interpretations are useful.

Notice that the (non-local) constraint $\sum_{i=1}^n\bsigma_i=\bzero$
does not enter these equations. However, it is enforced by selecting a
solution of the cavity equations such that $\E\{\int \bsigma_0\nu_0(\de\bsigma_0)\} = \bzero$, where $\E$ is
expectation with respect to the underlying graph which is a Galton-Watson tree with Poisson offspring distribution.

\begin{remark}\label{rem:vectorial}
We will carry out our calculations within a simple `vectorial' ansatz,
whereby $\nu_i(\de \bsigma_i)$ depends in a log-linear way on a
one-dimensional projection of $\bsigma_i$. While this ansatz is not
exact, it turns out to yield very accurate results. Also, it can be
systematically improved upon, a direction that we leave for future work.
\end{remark}

\subsection{Symmetric phase}
\label{sec:Symmetric}

For small $(a-b)$, we expect the solution to the cavity equation to be symmetric
(in distribution) under rotations in $\cO(m)$. By this we mean that, for
any rotation $R\in \cO(m)$, $\nu_i^R(\,\cdot\,)$ is distributed as $\nu_i(\,\cdot\,)$.
($\nu_i^R(\,\cdot\,)$ is defined as the measure induced by action
$\bsigma\mapsto R\bsigma$ on $S^{m-1}$, cf. Section \ref{sec:DenseGeneral}).

In the symmetric phase, assuming the `vectorial' ansatz, cf. Remark~\ref{rem:vectorial},  we look for an approximate solution of the form
\begin{align}
\nu_i(\de\bsigma_i) & 
\normeq
\exp\Big\{ 2\beta \sqrt{m\cond_i} \,
  \<\bz_i,\bsigma_i\> +O_m(1)\Big\}\, 
p_0(\de\bsigma_i) \label{eq:SymmetricSparseAnsatz}\\
&\normeq
\exp\Big\{ 2\beta \sqrt{m\cond_i} \,
  \<\bz_i,\bsigma_i\> +O_m(1)\Big\}\, 
\delta\Big(\|\bsigma_i\|_2^2-1\Big)\; \de\bsigma_i\, ,
\end{align}
where $\bz_i\sim\normal(0,\id_{m})$, and $O_m(1)$ represents a 
term of order one as $m\to\infty$.

Using the Fourier representation of the $\delta$
function (with associated parameter $\rho$), and performing the
Gaussian integral over $\bsigma$, we get 
\begin{align}
\hnu_i(\bsigma_0) &\normeq  \int 
\exp\Big\{2\beta
  m\<\bsigma_0,\bsigma\>+ \beta m\rho
+ 2\beta \sqrt{m\cond_i} \,
  \<\bz_i,\bsigma\>-\beta m\rho\|\bsigma\|_2^2+O_m(1)\Big\}\; 
\de\rho\, \de\bsigma\\
& \normeq  \int \rho^{-m/2}\, \exp\Big\{\frac{\beta }{\rho m}\Big\| \sqrt{m\cond_i}\bz_i
  +m\bsigma_0\Big\|^2_2 +\beta m\rho +O_m(1)\Big\}\, \de\rho\\
&\normeq \int \exp\Big\{\beta m \, S_i(\rho)+\, \frac{2\beta
  }{\rho}\sqrt{m\cond_i}\, \<\bz_i,\bsigma_0\>+\frac{\beta 
  \cond_i}{\rho}(\|\bz_i\|_2^2-m) +O_m(1)\Big\}\de\rho\, ,
\end{align}
where
\begin{align}
S_i(\rho) \equiv \rho +\frac{1}{\rho}\big( \cond_i
  +1\big)-\frac{1}{2\beta}\log\rho\, .
\end{align}
Here the indegral over  $\rho$ runs along the imaginary axis in the
complex plane, from $-i\infty$ and $+i\infty$.

Note, for $\bsigma_0$ uniformly random on the unit sphere, the term $(2\beta m/\rho)\<\bz_i,\bsigma_0\>$ is of order
$\sqrt{m}$,  i.e. of lower order with respect to the term including $S(\,\cdot\,
)$.
Also, the term $(\beta m \cond_i/\rho) (\|\bz_i\|_2^2-1)$ is of 
order $\sqrt{m}$ and does not depend on $\bsigma_0$. Hence, up to
an additional $O_m(1)$ term, we can reabsorb this in the normalization
constant.
We therefore get
\begin{align}
\hnu_i(\bsigma_0)& \normeq \int \exp\Big\{\beta m \, S_i(\rho)+\, \frac{2\beta
  }{\rho}\sqrt{m\cond_i}\, \<\bz_i,\bsigma_0\>+O_m(1)\Big\}\de\rho\,.
\end{align}
We next perform integration over $\rho$ by the saddle point
method. Since $\<\bz_i,\bsigma_0\> = O(m^{-1/2})$ and $S_i(\rho) = O(1)$, the saddle point is given by the stationary equation
$S_i'(\rho) = 0$. The saddle point $\rho_{i,*}$ lies on the real axis and is a minimum along the real axis but a maximum with respect to
the imaginary direction, i.e.,  
\begin{align}
\rho_{i,*}(\beta)= \arg\min_{\rho\in \reals_+} S_i(\rho)\, .
\end{align}
By Cauchy's theorem, we can deform the contour of integral to pass the saddle point  along the imaginary direction.
This in fact corresponds to the path that descents most steeply from the saddle point. 
The integral is dominated by $\rho = \rho_{i,*}+O(m^{-1/2})$ and hence,
\begin{align}
\hnu_i(\bsigma_0)&\normeq \exp\Big\{\frac{2\beta
  }{\rho_{i,*}}\sqrt{m\cond_i}\<\bz_i,\bsigma_0\>+O_m(1)\Big\}\, .\label{eq:HnuFormula}
\end{align}
While this expression for $\hnu_i(\bsigma_0)$ is accurate when
$\<\bz_i,\bsigma_0\>$ is small, it breaks down for large
$\<\bz_i,\bsigma_0\>$. In Section \ref{sec:Limitations} we will
discuss the regimes of validity of this approximation. Namely, we
expect it to be accurate for $d$ large and for $d$ close to one.

Substituting in Eq.~(\ref{eq:CavitySparse}), we get the recursion
\begin{align}
\cond_0 &= \sum_{i=1}^k \frac{\cond_i}{\rho^2_{i,*}}\, ,\\
\rho_{i,*}^2 &=\big( \cond_i
  +1\big)+\frac{1}{2\beta}\, \rho_{i,*} \, .
\end{align}
In particular, as $\beta\to\infty$, we get the simple equation
\begin{align}
\cond_0 &= \sum_{i=1}^k \frac{\cond_i}{1+ \cond_i}\, .
\end{align}

Note that $\cond_0,\cond_1,\dots,\cond_k$ are random variables, because of the
randomness in the underlying limiting tree, which is Galton-Watson
tree with Poisson offspring distribution. We get
\begin{align}
\cond\ed \sum_{i=1}^L\frac{\cond_i}{1+\cond_i}\, ,\label{eq:RDE_R}
\end{align}
where $L\sim \Poisson(d)$, $d=(a+b)/2$, and $\cond_1,\dots, \cond_L$ are
i.i.d. copies of $\cond$. We are interested in solutions supported on
$\reals_{\ge 0}$, i.e. such that $\prob(\cond\ge 0) = 1$.

This recursion is connected to potential theory on Galton-Watson
trees, see e.g.  \cite{lyons1997unsolved,lyons2013probability}.
The next proposition establishes a few basic facts about its
solutions. For the readers' convenience, we provide a proof of the
simplest statements, referring to \cite{lyons1997unsolved} for the
other facts. Given two random
variables $X,Y$, we write $X\preceq Y$ if $Y$ dominates stochastically
$X$,
i.e. if there exists a coupling of $X,Y$ such that $\prob(X\le Y)=1$.
\begin{proposition}[\cite{lyons1997unsolved}]
There exists random variables $\cond_{0}$, $\cond_*$ supported on
$\reals_{\ge 0}$, that solve Eq.~(\ref{eq:RDE_R}) and such that:
\begin{enumerate}
\item  $\cond_0=0$ almost surely.
\item $\cond_{*}\preceq L$, $L\sim\Poisson(d)$.
\item For $d\le 1$, $\cond_{*}=\cond_{0}=0$ identically. Hence the
  distributional equation (\ref{eq:RDE_R}) has a unique solution.
\item For $d>1$, $\cond_*>0$ with positive probability and further
equation (\ref{eq:RDE_R})  admits no other solution than $\cond_0$, $\cond_*$.
\end{enumerate}
\end{proposition}
\begin{proof}
Let $\Psp$ denote the space of probability measures over $[0,\infty]$,
and $\T_{d}:\Psp\to \Psp$ the map defined by the 
right-hand
side of Eq.~(\ref{eq:RDE_R}). Namely $\T_d(\mu)$ is the probability
distribution of the right-hand side of Eq.~(\ref{eq:RDE_R}) when
$\cond_i\sim_{i.i.d.} \mu$. Notice that this is well defined on the
extended real line because the summands are non-negative.
It is immediate to see that this map is monotone, i.e.
\begin{align}
\mu_1\preceq \mu_2\; \Rightarrow\; \T_d(\mu_1)\preceq \T_d(\mu_2)\,. \label{eq:monotone}
\end{align}
We define $\cond_{0}=0$ identically and $\cond_*\sim \mu^+$,
where $\mu^+ = \lim_{\ell\to \infty}\T^{\ell}_d(\delta_{+\infty})$
where ``$\lim$'' refers to the weak limit, which exists by monotonicity.
Points 1, 2 follows from an immediate monotonicity argument (see
e.g. \cite{aldous2005survey}).

For point 3, note that by Jensen inequality
\begin{align}
\E \cond\le \frac{d\E \cond}{1+\E \cond}\, ,
\end{align}
which implies $\E \cond = 0$, and hence $\cond=0$ identically.

Point 4 is just Theorem 4.1 in \cite{lyons1997unsolved}.
\end{proof}
The next proposition establishes an appealing interpretation of the
random variable $\cond_*$. 
Again, this  puts together results of \cite{lyons1990random} and
\cite{lyons1997unsolved}. We give here a proof of this connection for
the readers' convenience.
We refer to \cite{lyons2013probability} for
further background on discrete potential theory (electrical networks) and trees.
\begin{proposition}[\cite{lyons1990random,lyons1997unsolved}]
Let $\cT$ be a rooted Galton-Watson tree with offspring
distribution $\Poisson(d)$, and consider the associated (infinite)
electric network, whereby each edge of $\cT$ is replaced by a
conductor with unit resistance.
Let $\cC(0\leftrightarrow \ell)$ be the conductance between the root
and vertices at level $\ell$ (when nodes at level $\ell$ are connected
together). Let $\cC(0\leftrightarrow \infty)\equiv
\lim_{\ell\to\infty}\cC(0\leftrightarrow \ell)$. Then:
\begin{align}
\cond_* \ed \cC(0\leftrightarrow \infty)\, .
\end{align}
In particular, if $d>1$, then $\cond_*>0$ with positive probability.
\end{proposition}
\begin{proof}
Let $0$ denote the root of $\cT$,  and $\{1,2,\dots,k\}$ be its
children. Let the  conductance between the root and level $\ell$ on $\cT$ be
$\gamma_0(\ell)$,. Also, let $\cT_1$, \dots $\cT_k$ denote the subtrees rooted
at $1,\dots,k$. Let  the conductance between the root and level $\ell$
on $\cT_i$ (i.e. between vertex $i$ and vertices at distance $\ell$
from $i$ on $\cT_i$) be denoted by  $\gamma_i(\ell)$.
Let $\hgamma_i(\ell)$, $i\in \{1,\dots,k\}$ be the conductance of the
tree obtained from $\cT$ by removing edges $\{(0,j): j\in [k]\setminus
\{k\}\}$. Equivalently $\hgamma_i(\ell)$ is the conductance of the tree obtained
from $\cT_i$ by connecting the root of $\cT_i$ (vertex $i$) to $0$ and
moving the root to $0$. Since the resistance of a series is the sum of
resistances of each component, we have
\begin{align}
\frac{1}{\hgamma_i(\ell+1)} = 1+\frac{1}{\gamma_i(\ell)}\, .
\end{align}
This is of course equivalent to $\hgamma_i(\ell+1) = \gamma_i(\ell)/(1+\gamma_i(\ell))$. 

Now since the conductance of several resistances in parallel is equal
to the sum of the conductances of the components,  we get 
\begin{align}
\gamma_{0}(\ell+1) = \sum_{i=1}^k\hgamma_i(\ell+1) = \sum_{i=1}^k 
\frac{\gamma_i(\ell)}{(1+\gamma_i(\ell))}\, ,
\end{align}
with boundary condition $\gamma_0(0) =\infty$.
Notice that this coincides with the recursion for $\cond^{\ell}$,
included the boundary condition, thus proving our claim.

It follows from Theorem \cite[Theorem 4.3]{lyons1990random} and \cite[Proposition 6.4]{lyons1990random}
that  $\cond_*>0$ with positive probability whenever $d>1$.
\end{proof}
We will hereafter consider the case $d>1$ and focus on the
maximal solution $\cond_*$.

We first study the behavior of $c_*$ for large $d$ limit.
When $d\to\infty$, by the  law of large numbers the
right-hand of Eq.~(\ref{eq:RDE_R})  concentrates around a
deterministic value, and hence so does $\cond$.
In order to characterize this value for large $d$, we write $\cond = \ro + \dc$
where $\E(\dc) = 0$ and $\ro = \E(\cond)$ is deterministic. Further, we
expect $\dc = \Theta(\sqrt{d})$ and $\ro=\Theta(d)$. Expanding Eq.~(\ref{eq:RDE_R}), we get
\begin{align}
\ro + \dc &\ed \sum_{i=1}^L\Big\{\frac{\ro}{1+\ro} +\frac{\dc_i}{(1+\ro)^2}
-\frac{\dc_i^2}{(1+\ro)^3}+\frac{\dc_i^3}{(1+\ro)^4} +O(d^{-3})\Big\}\\
 &= \frac{L\, \ro}{1+\ro} + \frac{d\,(\E(\dc^2))^{1/2}}{(1+\ro)^2} \, Z_1
  -\frac{d\,\E(\dc^2)}{(1+\ro)^3} + O(d^{-3/2})\, ,\label{eq:dum1}
\end{align}
where  $Z_1\sim\normal(0,1)$ is independent of $L\sim\Poisson(d)$. Note that we used central limit theorem and law of large numbers
in obtaining~\eqref{eq:dum1}.

Taking expectation we get
\begin{align}
\ro = \frac{d\,\ro}{1+\ro}-\frac{d\,\E(\dc^2)}{(1+\ro)^3} + O(d^{-3/2})\, ,
\end{align}
whence
\begin{align}
\ro = d-1-\frac{\E(\dc^2)}{d^2} + O(d^{-3/2})\, .\label{eq:RoExpansion}
\end{align}

On the other hand, taking the variance, we obtain
\begin{align}
\E(\dc^2) = d \Big(\frac{\ro}{1+\ro}\Big)^2 + \frac{d^2\,\E(\dc^2)}{(1+\ro)^4} +O(d^{-3})\,.
\end{align}
Therefore,
\begin{align}\label{eq:dc2}
\E(\dc^2) = d-2 +O(d^{-1})\,.
\end{align}
Using equation~\eqref{eq:dc2} in equation~\eqref{eq:RoExpansion}, we get
\begin{align}
\ro = d-1-\frac{1}{d} + O(d^{-3/2})\, .
\end{align}
Summarizing, we found that
\begin{align}
\cond_* = d-1-\frac{1}{d} +\sqrt{d-2}\; Z  +O(d^{-3/2})\, ,
\end{align}
for $Z\sim\normal(0,1)$.

\subsection{Linear stability of the symmetric phase and critical
  point}
\label{sec:LinearStability}

We next study the stability of the symmetric solution
(\ref{eq:SymmetricSparseAnsatz}).  We break the $\cO(m)$ 
symmetry by letting 
\begin{align}
\nu_i(\de\bsigma_i) \normeq
\exp\Big\{2\beta \sqrt{m\cond_i} \,
  \<\bz_i,\btau_i\>+2\beta m   h_i \, s_i + O_m(1)\Big\}\, 
\delta\Big(\|\btau_i\|_2^2+s_i^2-1\Big)\; \de s_i\de^{m-1}\btau_i\, ,\label{eq:InstabilityAnsatz}
\end{align}
where $\bsigma_i = (s_i,\btau_i)$ and $\bz_i
\sim\normal(0,\id_{m-1})$ is $(m-1)$-dimensional. 
Note that each coordinate  of $\bz_i$ is of order $1$, and is
multiplied  by a factor $\sqrt{m}$ in the above expression. We will
consider $h_i\lll 1$, and expand all expressions to linear order in
$h_i$. 

Proceeding as in the symmetric phase, we get 
\begin{eqnarray}
\begin{split}
\hnu_i(\bsigma_0) &\normeq -j \int 
\exp\Big\{2\beta   m s_0s+ \beta m\rho-\beta
                    m\rho s^2+2\beta m h_i s\\
                  &\phantom{AAAA}+2\beta   m\<\btau_0,\btau\>+ 2\beta \sqrt{m\cond_i}\<\bz_i,\btau\>-\beta m\rho\|\btau\|_2^2+O_m(1)\Big\}\; 
\de\rho\, \de^{m-1}\btau\, \de s\\
& \normeq -j \int \rho^{-m/2}\, \exp\Big\{\frac{\beta }{\rho m}\Big\| \sqrt{m\cond_i}\bz_i
  +m\btau_0\Big\|^2_2
%\nonumber \\
%&\phantom{AAAAAAA}
+\frac{\beta m}{\rho}\Big(h_i
  +s_0\Big)^2+ \beta m\rho+O_m(1)\Big\}\,\, \de\rho\\
&\normeq -j \int \exp\Big\{\beta m \, S_i(\rho)+\, \frac{2\beta
  }{\rho}\sqrt{m\cond_i}\,
  \<\bz_i,\btau_0\>+\frac{2\beta m}{\rho} \, h_i s_0+
\frac{\beta
   \cond_i}{\rho}(\|\bz_i\|_2^2-m)+\frac{\beta m}{\rho} h_i^2+O_m(1)\Big\}\,\, \de\rho\, ,
\end{split}
\end{eqnarray}
where $S_i(\rho)$ is defined as in the symmetric phase, namely
\begin{align}
S_i(\rho) = \rho +\frac{1}{\rho}\big( \cond_i
  +1\big)-\frac{1}{2\beta}\log\rho\, .
\end{align}
We perform the integral by saddle point, around
$\rho_{i,*}(\beta)= \arg\min_{\rho\in \reals_+} S_i(\rho)$. Proceeding
as in the symmetric case, and neglecting terms quadratic in $h_i$, we have
\begin{align}
\hnu_i(\bsigma_0) &\normeq \exp\Big\{\frac{2\beta m}{\rho_{i,*}} \, h_i s_0+\frac{2\beta
  }{\rho_{i,*}}\sqrt{m\cond_i}\,  \<\bz_i,\btau_0\>+O_m(1) +O(h^2)\Big\}\, .
\end{align}

Substituting in Eq.~(\ref{eq:CavitySparse}), we obtain the equations
\begin{align}
\cond_0 = \sum_{i=1}^k \frac{\cond_i}{\rho^2_{i,*}}\, ,\;\;\;\;\;\;\;
h_0 = \sum_{i=1}^k \frac{h_i}{\rho_{i,*}}\, ,
\end{align}
which simplify at zero temperature to
\begin{align}
\cond_0 = \sum_{i=1}^k \frac{\cond_i}{1+ \cond_i}\, ,\;\;\;\;\;\;\;\;
h_0 = \sum_{i=1}^k \frac{h_i}{\sqrt{1+ \cond_i}}\, .
\end{align}

Recall that the graph $G_n$ converges locally to a two-types
Galton-Watson tree, whereby each vertex has $\Poisson(a/2)$ vertices
of the same type, and $\Poisson(b/2)$ vertices of the opposite type.
We look for solutions that break the symmetry $+1\leftrightarrow -1$.
If $(\cond_i,h_i)$ is the pair of random variables introduced above, for
vertex $i$, we therefore assume $(\cond_{i(+)},h_{i(+)}) \ed
(\cond_{i(-)},-h_{i(-)})$ for $i(+) \in V_+$, $i(-) \in V_-$. 
This leads to the following distributional recursion for the sequence
of random vectors $\{(\cond^t,h^t)\}_{t\ge 0}$:
\begin{align}
\big(\cond^{t+1};h^{t+1}\big)\ed
  \Big(\sum_{i=1}^{L_++L_-}\frac{\cond^t_i}{1+\cond^t_i}
  ;\sum_{i=1}^{L_++L_-}\frac{s_ih^t_i}{\sqrt{1+\cond^t_i}}
\Big)\, ,\label{eq:Stability}
\end{align}
where  $L_+\sim \Poisson(a/2)$, $L_-\sim\Poisson(b/2)$, $s_1,\dots, s_{L_+} = +1$,
$s_{L_++1},\dots, s_{L_++L_-} = -1$, and $\{(\cond_i^t,h^t_i)\}$ are
  i.i.d. copies of $(\cond^t,h^t)$

Let us pause for two important remarks:
\begin{enumerate}
\item The recursion (\ref{eq:Stability}) is invariant under the
  rescaling $h^t\to a\, h^t$ for $a\in\reals$. Hence, only properties
  that are invariant under thus rescaling are meaningful.
\item It admits the fixed point  $(\cond^t,h^t) \ed (\cond_*,\sqrt{\cond_*}\, Z)$
  where $\cond_*$ is the distributional fixed point of the symmetric
  phase, constructed in the previous section, and $Z\sim\normal(0,1)$
  is independent of $\cond_*$. This is a fixed point\footnote{Indeed, by
    the scaling invariance,  $(\cond^t,h^t) \ed (\cond_*,a\sqrt{\cond_*}\, Z)$ is
  a fixed point for any fixed scale factor  $a\in\reals$.}
that corresponds to
  the symmetric phase, and does not  break the $+1\leftrightarrow -1$ 
symmetry.
\end{enumerate}
Therefore, in order to investigate stability, we initialize the above
recursion in a way that breaks the symmetry, $(\cond^0,h^0) = (\infty,1)$.
Note that by monotonicity property~\eqref{eq:monotone}, starting with $\cond^0 = \infty$,
we have $\cond^t\stackrel{{\rm d}}{\Rightarrow}\cond_*$. 
We ask whether this perturbation grows, by computing the exponential
growth rate
\begin{align}
G_{\alpha}(d,\lambda) \equiv \lim\inf_{t\to\infty} \frac{1}{t\alpha}
  \log \E(|h^t|^{\alpha})\,, \label{eq:Gdef}
\end{align}
where $d=(a+b)/2$, and $\lambda = (a-b)/\sqrt{2(a+b)}$ parametrize the
model.
We define the critical point as the smallest $\lambda$ such that the
growth rate is strictly positive:
\begin{align}
\tlsdp(d) \equiv \inf\big\{ \lambda\in [0,\sqrt{d}] :\,\,
G_2(d,\lambda)>0\big\}\, .\label{eq:LSDP}
\end{align}
Notice that in the definition we used the second moment, i.e. set
$\alpha=2$. However, the result appear to be insensitive to the choice
of $\alpha$.
In the next section we will discuss the numerical solution of the
above distributional equations and our analytical
prediction for $\lsdp(d)$. 

In the rest of this section
we analyze the behavior of $\tlsdp(d)$ for large $d$. Along the way,
we analyze the behavior of perturbation $h^t$, which in turn clarifies why the critical point $\tlsdp$ is defined based
on the \emph{exponential growth rate} of the perturbation.

For the sake of simplicity we assume the initialization
$(\cond^0,h^0)\ed(\cond_*,1)$. Since with initialization
$\cond^0=\infty$ we have $\cond^t\stackrel{{\rm
    d}}{\Rightarrow}\cond_*$, this should be  equivalent to our
initialization $(\cond^0,h^0) = (\infty,1)$. We let $\cond^t = \ro +\dc^t$, $\ro = \E(\cond_*) = \E(\cond^t)$ as in the
previous section. Note that although $\cond^t$ is distributed as per $\cond_*$, joint distribution of $(\cond^t,h^t)$ varies
over time and so we make the iteration number explicit in $\dc^t$.

We start by taking expectation of Eq.~(\ref{eq:Stability}).
\begin{align}
\E(h^{t+1}) &= \frac{a-b}{2}\,  \E\Big\{\frac{h^t}{\sqrt{1+\cond^t}}\Big\}\\
& =  \frac{a-b}{2}\,  \E\Big\{h^t\Big[\frac{1}{(1+\ro)^{1/2}}
-\frac{\dc^t}{2(1+\ro)^{3/2}} +\frac{3(\dc^t)^2}{8(1+\ro)^{5/2}} +O(d^{-2})\Big]\Big\}\, . \label{eq:ExpansionMeanH}
\end{align}
By taking the covariance of $h^t$ and $\cond^t$, we obtain 
\begin{align}
\E(\dc^th^t) = \cov(\cond^t;h^t) &= \frac{a-b}{2}\, \E\Big\{\frac{h^{t-1}\cond^{t-1}}{(1+\cond^{t-1})^{3/2}}\Big\}\\
& = \frac{a-b}{2} \frac{\ro}{(1+\ro)^{3/2}} \E(h^{t-1})
\big(1+O(d^{-1})\big)\, .
\end{align}
Further we have
\begin{align}
\E\{h^t(\dc^t)^2\} &= \E\{h^t\}\E\{(\dc^t)^2\}
                     \big(1+O(d^{-1/2})\big)\\
& = d\,\E\{h^t\}
                     \big(1+O(d^{-1/2})\big)\, .
\end{align}
Substituting in Eq.~(\ref{eq:ExpansionMeanH}), using  $a-b =
2\lambda\sqrt{d}$, and setting $\oh^t = \E\{h^t\}$, we get
\begin{align}
\oh^{t+1} &= \lambda\sqrt{d}
\Big\{\frac{1}{(1+\ro)^{1/2}}\,\oh^t
  -\frac{\ro\lambda\sqrt{d}}{2(1+\ro)^3}\, \oh^{t-1}
+\frac{3d}{8(1+\ro)^{5/2}}\,\oh^t +O(d^{-2}\oh^t)\Big\}\\
&= \lambda\Big\{\Big(1+\frac{3}{8d}\Big)\oh^t -\frac{1}{2d}\,\oh^{t-1}
            +O(d^{-3/2}\oh^t)\Big\}\,, \label{eq:ht}
\end{align}
Hence, to this order
\begin{align}
\tlsdp(d) &= \rho_{{\rm sp}}(M_d)^{-1} + O(d^{-3/2})\, ,\\
M_d & \equiv \begin{bmatrix}
1+3/(8d) & -1/2d\\
1 & 0
\end{bmatrix}\, .
\end{align}
where $\rho_{{\rm sp}}(\,\cdot\,)$ denotes the spectral radius of a
matrix.
A simple calculation yields
\begin{align}
\tlsdp(d) = 1+\frac{1}{8d} + O(d^{-3/2})\, .
\end{align}

\subsection{Numerical solution of the distributional recursions}
\label{sec:numLinear}

We solved numerically the distributional recursions (\ref{eq:RDE_R}),
(\ref{eq:Stability}) through a sampling algorithm that is known as
`population dynamics' within spin glass theory \cite{mezard2001bethe}.
The algorithm updates a sample that, at iteration $t$, is meant to
be an approximately iid samples with the same law as the one defined by
the distributional equation, at iteration $t$.
For concreteness, we define the algorithm here in the case of the
iteration corresponding to Eq.~(\ref{eq:RDE_R}):
\begin{align}
\cond^{t+1} \ed\sum_{i=1}^L\frac{\cond^t_i}{1+\cond^t_i}\, .
\end{align}
The distribution of $\cond^t$ will be approximated by a sample 
$\bcond^{t} \equiv (\cond_1^{t},\cond_2^{t},\dots,\cond_N^t)$ (we
represent this by a vector but ordering is irrelevant).

\begin{algorithm}[h]
\caption{Population dynamics algorithm}
\begin{algorithmic}[1]
\REQUIRE Sample size $N$; Number of iterations $T$ 
\ENSURE Samples $\bcond^1, \bcond^2,\dotsc, \bcond^T$
\STATE $\bcond^1\leftarrow (L_1, L_2,\dots,L_N)$, with $L_i\sim_{i.i.d.} \Poisson(d)$
\FOR { $t=1,2, \dots, T=1$}
\STATE $\bcond^{t+1} \leftarrow (\;)$
\FOR{$i=1,2,\dots, N$}
\STATE Generate $L\sim\Poisson(d)$
\STATE Generate $i(1),\dots, i(L)\sim_{i.i.d.}\Unif(\{1,2,\dots,N\})$
\STATE Compute 
\begin{align}
\cond^{{\rm new}} =
  \sum_{a=1}^{L}\frac{\cond^t_{i(a)}}{1+\cond^t_{i(a)}}\, .
\end{align}
\STATE Set $\bcond^{t+1} \leftarrow [\bcond^{t+1} |\cond^{{\rm new}}]$
(append entry $\cond^{{\rm new}}$ to vector $\bcond^{t+1}$)
\ENDFOR
 \ENDFOR
\RETURN $\bcond^1, \dotsc, \bcond^T$
\end{algorithmic}
\end{algorithm}

The notation $[a|b]$ in the step 8 of the algorithm denotes appending element $b$ to vector $a$.
Note that with initial point $c_0 = \infty$, we have $c_1 = L\sim \Poisson(d)$. In the population dynamic algorithm
we start from $c^1$.

As an illustration, Figure \ref{fig:RDE_COND} presents the results of
some small-scale calculations using this algorithm.

\begin{figure}[t!]
\includegraphics[width=.5\textwidth]{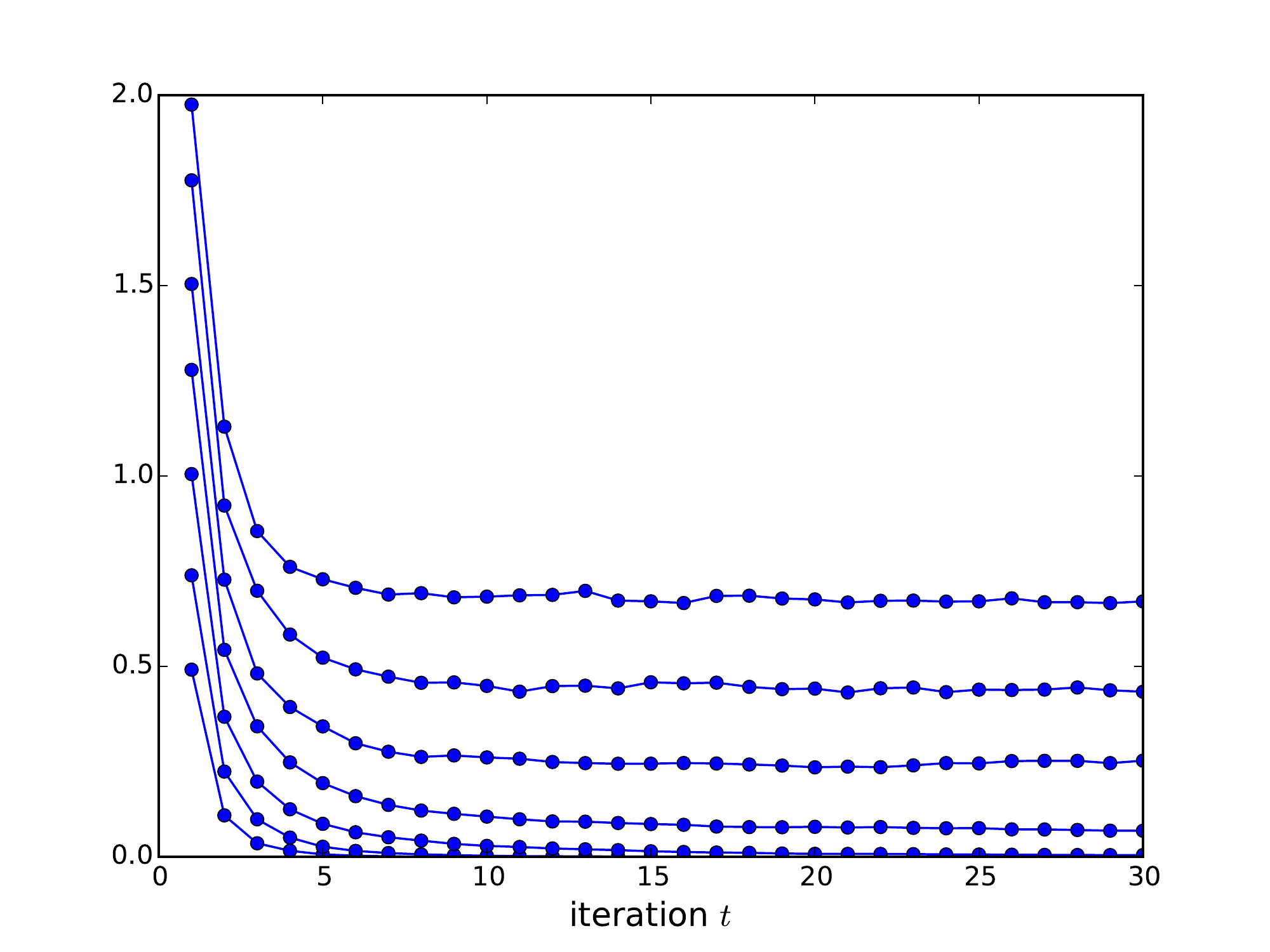}
\includegraphics[width=.5\textwidth]{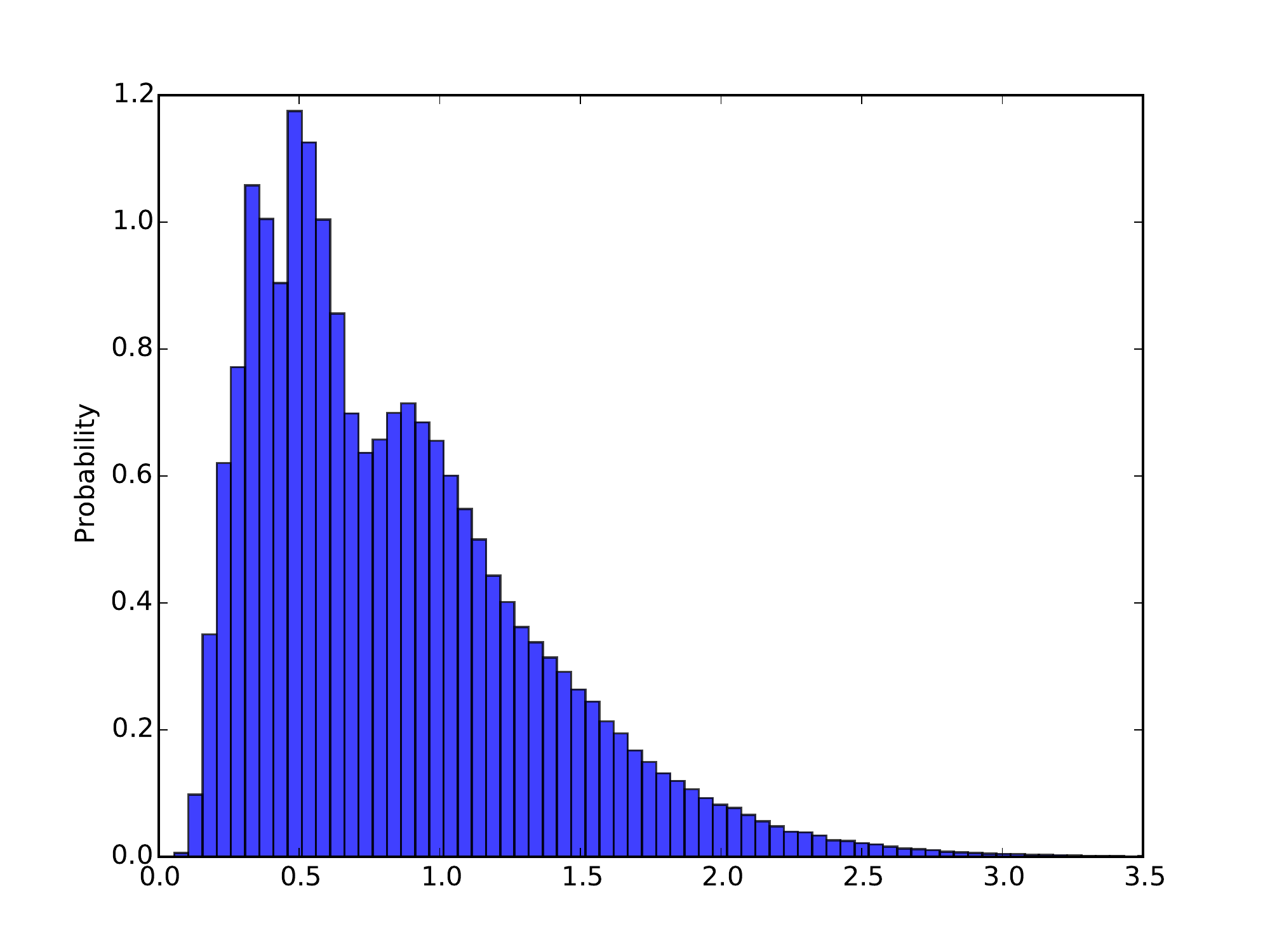}
\put(-475,87){{\scriptsize $\E\{\cond^t\}$}}
\put(-115,5){{\scriptsize $\cond$}}
\caption{Solution of the recursive distributional equation
  (\ref{eq:RDE_R}) using the population dynamics algorithm.
Left frame: evolution of the mean $\E\{\cond^t\}$ versus the number of
iterations $t$, as estimated by the algorithm. 
Various curves refer (from bottom to top) to $d = 0.5$, $0.75$, $1$,
$1.25$, $1.5$, $1.75$, $2$.
Here sample size is $N=5\cdot 10^3$.
 Right frame: histogram of the samples at convergence, for $d=2$. 
Here $N=5\cdot 10^3$.
\label{fig:RDE_COND}}
\end{figure}

\begin{figure}[t!]
\begin{center}
\includegraphics[width=.6\textwidth]{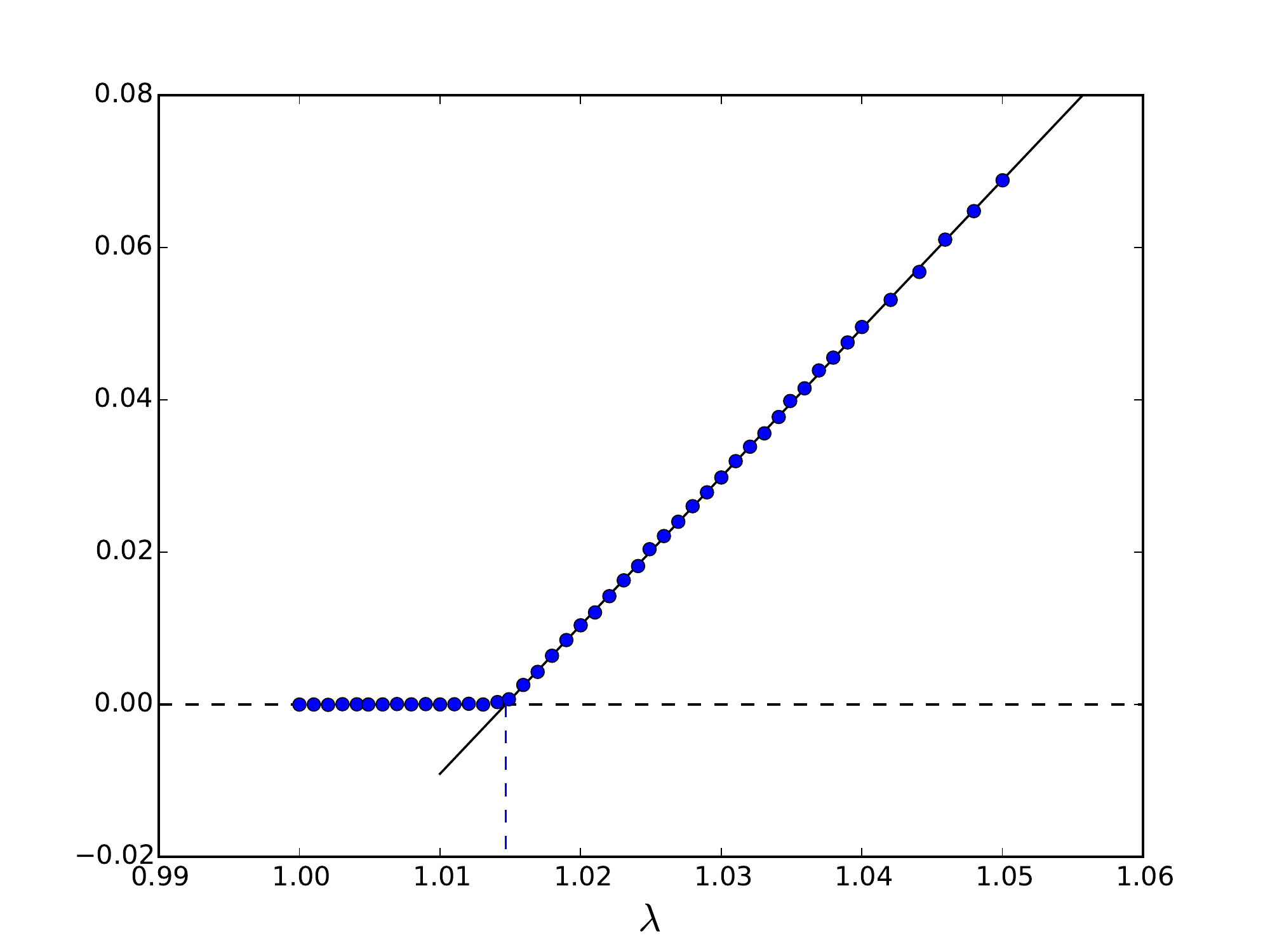}
\put(-153,145){$\hG_2(d=6,\lambda)$}
\put(-160,27){{\small $\tlsdp(d=6)$}}
\end{center}
\caption{Local stability of the $\cO(m)$ symetric phase. We evaluate
  $G_2(d,\lambda)$ defined by Eq.~(\ref{eq:Gdef}) using the
  population dynamics algorithm. Data here reder to average degree
  $d=6$.  The phase transition point
  $\tlsdp(d)$ (within the vectorial
  ansatz) is determined by a local
  linear fit to the estimated $G_2(d,\lambda)$, when it is
  significantly larger than $0$. \label{fig:RDE_LSDP}}
\end{figure}

\begin{figure}[t!]
\begin{center}
\includegraphics[width=.7\textwidth]{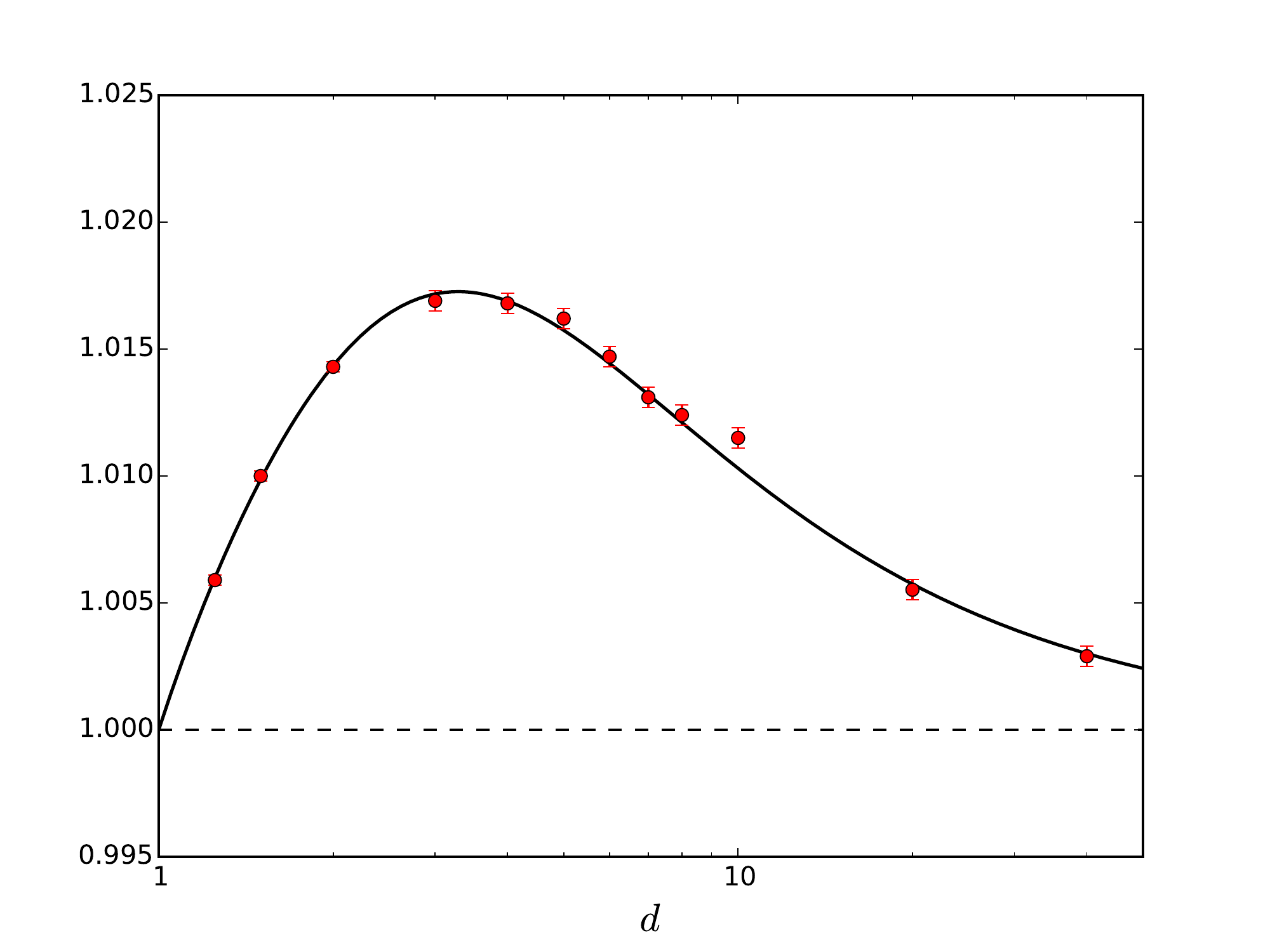}
\put(-350,130){$\tlsdp(d)$}
\end{center}
\caption{Theoretical prediction of the phase transition location for
  detecting hidden partition in sparse graphs using semidefinite
  programming. here the phase transition location is reported in the
rescaled variable $\lambda = (a-b)/\sqrt{2(a+b)}$ as a function of the
average degree $d=(a+b)/2$. Points with error bars correspond to
the definition (\ref{eq:LSDP}), evaluated numerically by
approximating the recursion (\ref{eq:Stability}) with the population
dynamics algorithm. The continuous curve is a rational fit to the
data, constrained at $d=1$, and for large $d$.
\label{fig:LSDP_vs_d}}
\end{figure}

We used the obvious modification of this algorithm to implement the
recursion (\ref{eq:Stability}), whereby a population is now formed of
pairs $(\cond^t_1,h^t_1)$, \dots $(\cond^t_N,h^t_N)$. An important
difference is that the overall scaling of the $h^t_i$ is
immaterial. We hence normalize them at each iteration as follows
\begin{align}
(h^t_i)^{{\rm new}} = \frac{h^t_i}{\sqrt{M_t}}\, ,\;\;\;\;\;
M_t \equiv \frac{1}{N}\sum_{i=1}^N(h^t_i)^2\, .
\end{align}
The normalization constant $M_t$ also allow us to estimate
$G_2(d,\lambda)$, namely 
\begin{align}
\hG_2(d,\lambda)\equiv \frac{1}{2(t_{\max}-t_{\min}+1)}
  \sum_{t=t_{\min}}^{t_{\max}} \log{M_t}\, .
\end{align}
Figure \ref{fig:RDE_LSDP} presents the typical results of this
calculation, using $N=10^7$, $t_{\min} = 100$, $t_{\max} = 400$, at average degree $d=6$.

The behavior in Figure \ref{fig:RDE_LSDP} is generic. For small
$\lambda$, the estimate $\hG_2(d,\lambda)$ is statistically
indistinguishable from $0$. Above a critical point, that we identify
with $\tlsdp(d)$,  $\hG_2(d,\lambda)$ is strictly positive, and
essentially linear in $\lambda$, close to $\tlsdp(d)$.
In order to estimate $\tlsdp(d)$ we use a local linear fit, with parameters
$g_0$, $g_1$
\begin{align}
\hG_2(d,\lambda) = g_0 + g_1\lambda\, ,
\end{align}
and set $\tlsdp(d) = -\hat{g}_0/\hat{g}_1$. In the fit we include only the values
of $\lambda$ such that $\hG_2(d,\lambda)$ is significantly different
from $0$. Note that the linear fit here is in a local neighborhood of $\lambda =1$.
By equation~\eqref{eq:ht}, it is easy to see that for large fixed $d$, $G_2(d,\lambda)$ is logarithmic
in $\lambda$.

The resulting values of $\tlsdp(d)$ are plotted in Figure
\ref{fig:LSDP_vs_d}, together with statistical errors. We report the
same values in Table \ref{table:LSDP} (note that the value for $d=10$ appears
to be an outlier).

\begin{table}
\begin{center}
\begin{tabular}{|c|c|}
\hline
$d$ & $\tlsdp(d)$\\
\hline
$1.25$ & $1.0059\pm 0.0002$\\
$1.5$   & $1.0100\pm 0.0002$\\
$2$  & $1.0143\pm 0.0002$\\
$3$  & $1.0169\pm 0.0004$\\
$4$  & $1.0168\pm 0.0004$\\
$5$  & $1.0162\pm 0.0004$\\
$6$  & $1.0147\pm 0.0004$\\
$7$  & $1.0131\pm 0.0004$\\
$8$  & $1.0124\pm 0.0004$\\
$10$  & $1.0115\pm 0.0004$\\
$20$  & $1.0055\pm 0.0004$\\
$40$  & $1.0029\pm 0.0004$\\
\hline
\end{tabular}
\end{center}
\caption{Numerical determination of
  the critical point $\tlsdp(d)$
  (within the vectorial ansatz)  for
  a few values of $d$.\label{table:LSDP}}
\end{table}

In order to interpolate these results, we fitted a rational function,
with the correct asymptotic behavior at $d=1$ (namely, $\tlsdp(d=1) =
1$) and at $d\to\infty$  (namely, $\tlsdp(d) =1+1/(8d)+o(d^{-1})$).
This fit is reported as a continuous line in Figure
\ref{fig:RDE_LSDP},
and is given by
\begin{align}
\lambda^{\mbox{\tiny FIT}}_c(d)=
1+ \frac{a_1(d-1)+a_2(d-1)^2}{1+b_1(d-1)+(8 a_1+19.5 a_2) (d-1)^2+8 a_2
  (d-1)^3}\, ,
\end{align}
with parameters
\begin{align}
a_1 &= 0.0307569\, ,\\
a_2  & = 0.030035\, ,\\
b_1  &= 2.16454\, .
\end{align}
This is also the curve reported in the main text.

\subsection{The recovery phase (broken $\cO(m)$ symmetry)}
\label{sec:RecoveryPhase}

In this section we describe an approximate solution of the cavity
equations within the  region $\lambda>\lsdp(d)$. In this regime the SDP
estimator has positive correlation with the ground truth (in the
$n\to\infty$ limit). We will work within a generalization of the vectorial ansatz
introduced in Section \ref{sec:LinearStability}. 
Namely, we look for a
solution which breaks the $\cO(m)$
symmetry to $\cO(m-1)$ along the first direction, as follows.
Letting $\bsigma_i =
(s_i,\btau_i)$, $s_i\in\reals$, $\btau_i\in\reals^{m-1}$, we adopt the ansatz
\begin{align}
\nu_i(\de\bsigma_i) \normeq
\exp\Big\{2\beta \sqrt{m\cond_i} \,
  \<\bz_i,\btau_i\>+2\beta m h_i\, s_i-\beta m r_i s_i^2+O_m(1)\Big\}\, 
\delta\Big(s_i^2+\|\btau_i\|_2^2-1\Big)\; \de^{m}\bsigma_i\, ,\label{eq:GeneralVectorAnsatz}
\end{align}
where we will assume $\bz_i\sim\normal(0,\id_{m-1})$.
In the following, we let
$\P_1=\bfe_1\bfe_1^{\sT}$ be the projector along the first direction
and $\Pp_1=\id-\P_1$ denote the orthogonal projector. 

Recalling the cavity equations (\ref{eq:CavitySparse}), and using the
Fourier representation of the delta function, we get 
\begin{align}
\hnu_i(\bsigma_0) &\normeq\int \exp\Big\{\beta m\rho-\beta m\big(
(\rho+r_i)s_i^2+\rho\|\btau_i\|_2^2\big) \\
&\phantom{AAAAAAAAAAA}+2\beta
                  m\big(s_0s_i+\<\btau_0,\btau_i\>\big)+2\beta m h_i s_i+2\beta
\sqrt{m\cond_i}\<\bz_i,\btau_i\>\Big\}\, \de\rho\,\de^{m-1}\btau_i \de s_i\nonumber\\
& \normeq \int \rho^{-m/2}\, \exp\Big\{\beta m\rho
+\frac{(2\beta m  s_0+2\beta m h_i)^2}{4\beta  m(\rho+r_i)}\
+\frac{\|2\beta m \btau_0+2\beta\sqrt{m\cond_i}\bz_i\|_2^2}{4\beta  m\rho}
\Big\}\,\de\rho\\
& \normeq \int \exp\Big\{\beta m S_i(\rho;s_0)
+\frac{2\beta\sqrt{m\cond_i}}{\rho}\<\bz_i,\btau_0\>\Big\}\, \de\rho\, ,
\end{align}
where 
\begin{align}
S_i(\rho;s_0) & = S_{i,0}(\rho)+S_{i,1}(\rho)\, s_0 +
            \frac{1}{2}S_{i,2}(\rho)\, s_0^2\, ,\label{eq:ExpansionSi}\\
S_{i,0}(\rho) & \equiv \rho-\frac{1}{2\beta}\log\rho +
                \frac{h_i^2}{\rho+r_i}
                +\frac{1+\cond_i}{\rho}\, ,\\
S_{i,1}(\rho) & = \frac{2h_i}{\rho+r_i}\, ,\\
S_{i,2}(\rho) & = \frac{2}{\rho+r_i}-\frac{2}{\rho}\, .
\end{align}
where we approximated $\|\bz_i\|_2^2/m \approx 1$, and used the
identity $\|\btau_0\|_2^2=1-s_0^2$.
In the $m\to\infty$ limit, we approximate the integral over $\rho$ by its
saddle point. In order to obtain a set of equations for the parameters
of the ansatz (\ref{eq:GeneralVectorAnsatz}), we will expand the
exponent to second order in $s_0$. The saddle point location is given
by
\begin{align}
\rho_{i,*}(s_0) & \equiv \arg\min_{\rho\in\reals_{>0}}S_i(\rho;s_0)= \rho_{i} + \rho_{i,1} s_0 +O(s_0^2)\, .
\end{align}
Here $\rho_{i}$ solves the equation $S'_{i,0}(\rho_i) =0$. Henceforth
we shall focus on the $\beta\to\infty$ limit, in which the equation $S'_{i,0}(\rho_i) =0$
reduces to
\begin{align}
1 = \frac{h_i^2}{(\rho_{i}+r_i)^2}+\frac{1+\cond_i}{\rho_i^2}\, .\label{eq:RhoI}
\end{align}
The first order correction is given by $\rho_{i,1} =-S'_{i,1}(\rho_i)/S_{i,0}''(\rho_i)$. Substituting
in Eq.~(\ref{eq:ExpansionSi}), we get the saddle point value of $S_i(\rho_i;s_0)$:
\begin{align}
\min_{\rho\in\reals_{>0}}S_i(\rho;s_0) &= S_{i,0}(\rho_i)+S_{i,1}(\rho_i)\, s_0+\Delta_2  s_0^2+O(s_0^3)\, ,\\
\Delta_2 &=\frac{1}{2}\left[S_{i,2}(\rho_i)-\frac{S'_{i,1}(\rho_i)^2}{S''_{i,0}(\rho_i)}\right]\\
&= -\frac{1}{\rho_i}+\frac{1}{\rho_i+r_i}-
\left(\frac{h_i^2}{(\rho_i+r_i)^3}+\frac{1+\cond_i
  }{\rho_i^3}\right)^{-1}\frac{h_i^2}{(\rho_i+r_i)^4}\\
&= -\frac{1}{\rho_i}+\frac{1}{\rho_i+r_i}-\left(1+\frac{(1+\cond_i) r_i}{\rho_i^3}\right)^{-1}
\frac{h_i^2}{(\rho_i+r_i)^3}
\,,
\end{align}
where in the last expression we used Eq.~(\ref{eq:RhoI}).
Substituting in Eq.~(\ref{eq:CavitySparse}), we get
a recursion for  the triple $h_i$, $\cond_i$, $r_i$:
\begin{align}
\cond_0 & = \sum_{i=1}^k\frac{\cond_i}{\rho_{i}^2}\,,\\
h_0 & = \sum_{i=1}^k\frac{h_i}{\rho_{i}+r_i}\, ,\\
r_0 &=\sum_{i=1}^k\left\{\frac{1}{\rho_i}-\frac{1}{\rho_i+r_i}+
\left(1+\frac{(1+\cond_i) r_i}{\rho_i^3}\right)^{-1}\frac{h_i^2}{(\rho_i+r_i)^3}
\right\}
\end{align}
If $G_n$ is distributed according to the two-groups stochastic block
model, the distributions of this triples on different type vertices
are related by symmetry $(\cond_{i(+)},h_{i(+)},r_{i(+)}) \ed
(\cond_{i(-)},-h_{i(-)},r_{i(-)})$ for $i(+) \in V_+$, $i(-) \in V_-$. 
This leads to the following distributional recursion for the sequence
of random vectors $\{(\cond^t,h^t,r^t)\}_{t\ge 0}$:
\begin{align}
\cond^{t+1} & = \sum_{i=1}^{L_++L_-}\frac{\cond^t_i}{\rho_{i}^2}\, ,\label{eq:FullVectorial_1}\\
h^{t+1}& = \sum_{i=1}^{L_++L_-}\frac{s_ih^t_i}{\rho_{i}+r^t_i}\, ,\label{eq:FullVectorial_2}\\
r^{t+1} &=\sum_{i=1}^{L_++L_-}\left\{\frac{1}{\rho_i}-\frac{1}{\rho_i+r^t_i}+
\left(1+\frac{(1+\cond^t_i) r^t_i}{\rho_i^3}\right)^{-1}\frac{(h^t_i)^2}{(\rho_i+r^t_i)^3}
\right\}\, ,\label{eq:FullVectorial_3}
\end{align}
where  $L_+\sim \Poisson(a/2)$, $L_-\sim\Poisson(b/2)$, $s_1,\dots, s_{L_+} = +1$,
$s_{L_++1},\dots, s_{L_++L_-} = -1$, and $\{(\cond_i^t,h^t_i,r^t_i)\}$ are
  i.i.d. copies of $(\cond^t,h^t,r^t)$. Finally, $\rho_i$ is a
  function of $(\cond_i^t,h^t_i,r^t_i)$
  implicitly defined as the solution of
\begin{align}
1 = \frac{(h^t_i)^2}{(\rho_{i}+r^t_i)^2}+\frac{1+\cond^t_i}{\rho_i^2}\, .\label{eq:RhoI_T}
\end{align}

As a check of the above derivation, we next verify that we recover the
correct $d\to\infty$ limit, captured by the $\integers_2$
synchronization model, cf. Section \ref{sec:SDP_Synchro}.
We will focus on the fixed point of recursions
(\ref{eq:FullVectorial_1}) to (\ref{eq:FullVectorial_3}), and hence
omit iteration index $t$.
For large degree, we expect central limit theorem to imply the
following asymptotic behaviors:
\begin{align}
\cond_i &= d\, \mu_{\cond}+ O(\sqrt{d})\, ,\\
h_i & \ed  \sqrt{d}\big(\mu_h+\sigma_h\, Z\big) +o(\sqrt{d})\, ,\\
r_i & = \sqrt{d} \, \mu_r  +O(1)\, , 
\end{align}
where $\mu_{\cond}$, $\mu_h$, $\sigma_h$, $\mu_r$ are deterministic parameters to be
determined. Equation (\ref{eq:RhoI_T}) thus  implies $\rho_i
=\sqrt{d}\rho$, with $\rho$ solution of 
\begin{align}
1 =
  \frac{(\mu_h+\sigma_hZ)^2}{(\rho+\mu_r)^2}+\frac{\mu_{\cond}}{\rho^2}\, .\label{eq:RhoNew}
\end{align}
Equations (\ref{eq:FullVectorial_1}) to  (\ref{eq:FullVectorial_3})
then yield the following. From Eq.~ (\ref{eq:FullVectorial_1}) we get
\begin{align}
\mu_{\cond} &= \mu_{\cond}\E\left\{\frac{1}{\rho^2}\right\}\, ,\label{eq:Muc}
\end{align}
i.e., assuming $\mu_{\cond}\neq 0$, $\E\{\rho^{-2}\}=1$.
From Eq.~ (\ref{eq:FullVectorial_2}) we get
\begin{align}
\mu_h & = \lambda\,\E\left\{\left(\frac{\mu_h+\sigma_h
        Z}{\rho+\mu_r}\right)\right\}\, ,\label{eq:mu_h}\\
\sigma_h^2 & = \E\left\{\left(\frac{\mu_h+\sigma_h
        Z}{\rho+\mu_r}\right)^2\right\}\, .\label{eq:sigma_h}
\end{align}
Finally, from Eq.~(\ref{eq:FullVectorial_3}) we get
\begin{align}
\mu_r & =
        \E\left\{\frac{1}{\rho}-\frac{1}{\rho+\mu_r}+\left(1+\frac{\mu_{\cond}\mu_r}{\rho^3}\right)^{-1}
\frac{(\mu_h+\sigma_h Z)^2}{(\rho+\mu_r)^3}
\right\}\, .\label{eq:mu_r}
\end{align}
We claim that these equations are equivalent to the ones for the
$\integers_2$ synchronization problem, derived using the replica
method, namely Eqs.~(\ref{eq:minfty_1}) to (\ref{eq:RhoExp}). Indeed
Eq.~(\ref{eq:RhoExp}) coincides with Eq.~(\ref{eq:Muc}), for
$\mu_{\cond}\neq 0$. Setting $\mu_h=\mu $, $\sigma_h^2=q$, and
$\mu_r = r$, we obtain that Eqs.~(\ref{eq:minfty_1}),
(\ref{eq:minfty_2}) coincide with Eqs.~(\ref{eq:mu_h}),
(\ref{eq:sigma_h}). Further, taking expectation of Eq.~(\ref{eq:RhoNew}),
and using $\E\{\rho^{-2}\} =1$, we get $\mu_\cond = 1-q$. Hence, we
obtain that Eq.~(\ref{eq:RhoNew}) coincides with
Eq.~(\ref{eq:minfty_RhoEq}).

Substituting the values of various parameters in Eq.~(\ref{eq:mu_r}),
we obtain
\begin{align}
r & =
        \E\left\{\frac{1}{\rho}-\frac{1}{\rho+r}+\left(1+\frac{r(1-q)}{\rho^3}\right)^{-1}
\frac{(\mu+\sqrt{q} Z)^2}{(\rho+r)^3}
\right\}\, . \label{eq:mu_r_new}
\end{align}
We claim that this is equivalent to Eq.~(\ref{eq:minfty_3}). To see
this, notice that differentiating Eq.~(\ref{eq:RhoNew}) with respect
to $Z$ we get
\begin{align}
\frac{\partial \rho}{\partial Z} &=
  \left[\frac{(\mu+\sqrt{q}Z)^2}{(\rho+r)^3}+\frac{(1-q)}{\rho^3}\right]^{-1}
\cdot \frac{(\mu+\sqrt{q}Z)\sqrt{q}}{(\rho+r)^2}\\
& = \left(1+\frac{r(1-q)}{\rho^3}\right)^{-1} \frac{(\mu+\sqrt{q}Z)\sqrt{q}}{\rho+r}\, ,
\end{align}
where the second equality follows again from
Eq.~(\ref{eq:RhoNew}). Using this identity, we can rewrite
Eq.~(\ref{eq:mu_r_new}) as
\begin{align}
r & =
        \E\left\{\frac{1}{\rho}-\frac{1}{\rho+r}+\frac{1}{\sqrt{q}}
\frac{(\mu+\sqrt{q} Z)}{(\rho+r)^2}\, \frac{\partial \rho}{\partial Z}  
\right\}\, , 
\end{align}
or 
\begin{align}
r & =
        \E\left\{\frac{1}{\rho}-\frac{1}{\sqrt{q}}
\frac{\partial \phantom{Z}}{\partial Z}  \frac{\mu+\sqrt{q} Z}{\rho+r}
\right\}\, , 
\end{align}
Using Gaussian integration by parts in the second term we finally
obtain Eq.~(\ref{eq:minfty_3}).
This concludes our verification for the case of $d\to\infty$.

\subsection{Limitations of the
  vectorial ansatz}
\label{sec:Limitations}

Our analytical estimate of the SDP phase transition location, $\tlsdp(d)$ was carried out within
the vectorial ansatz in Eqs.~(\ref{eq:SymmetricSparseAnsatz}), (\ref{eq:InstabilityAnsatz}). Let us
stress once again that this ansatz is only approximate. 

The origin of this approximation can be gleaned
from the calculation in Section \ref{sec:Symmetric}. As we have seen Eq.~(\ref{eq:HnuFormula}) is only accurate
when $\<\bsigma_0,\bz_i\>$ is small. However, according to the same ansatz, $\bsigma_0$ will be aligned
to $\bz_0$, which --in turn-- can be aligned with $\bz_i$.

We expect this approximation to be accurate in the following regimes:
\begin{itemize}
\item For large average degree $d$. Indeed, in this case, $\bz_0$ is weakly correlated with $\bz_i$.
\item For $d$ close to $1$. In this case $\cond_0$ is small and hence, under $\nu_0$, $\bsigma_0$ is
approximately uniformly distributed, and hence has a small scalar product $\<\bz_i,\bsigma_0\>$.
\end{itemize}
Let us also notice that the vectorial ansatz can be systematically improved upon by 
considering quadratic terms tepending in two-dimensional projections, and so on. We 
leave this direction for future work.
%
%*************
%
\section{Numerical experiments for community detection}
\label{sec:commDetec}

In this section we provide details about our numerical simulations with 
the SDP estimator for the community detection problem. For the
reader's convenience we begin by recalling some definitions. 

We denote by $G_n = (V_n,E_n)$ the random graph over vertex set
$V_n=[n]$, generated according the hidden partition model,
and by  $\bxz\in\{+1,-1\}^n$ the vertex labels.
Conditional on $\bxz$,  edges are independent with distribution
\begin{align}
\prob\big\{(i,j)\in E_n\big|\bxz\big\} = 
\begin{cases}
a/n & \mbox{ if $x_{0,i}x_{0,j} =+1$,}\\
b/n & \mbox{ if $x_{0,i}x_{0,j} =-1$.}
\end{cases}\label{eq:SBMDefinition}
\end{align}
We  denote by $d = (a+b)/2$ the average degree, and by
$\lambda=(a-b)/\sqrt{2(a+b)}$ the `signal strength.'

Throughout this section, $\di$ indicates the set of neighbors of
vertex $i$, i.e. $\di \equiv \{j\in [n] : (i,j)\in E\}$.

We next recall the SDP relaxation for estimating community memberships:
\begin{align}
\begin{split}
\mbox{maximize} & \;\;\;\sum_{(i,j)\in E}X_{ij}\, ,\label{eq:SDP}\\
\mbox{subject to} & \;\;\; \bX\succeq 0\, ,\\
& \;\;\; \bX\bone = \bzero\, ,\;\;\;\; X_{ii}=1 \;\; \forall i\in [n]\, .
\end{split}
\end{align}
Denote by $\bX_\sopt= \bX_\sopt(G)$ an optimizer of the above problem. The estimated
membership vector is then obtained by `rounding' the principal
eigenvector of $\bX_\sopt$
as follows. Letting $\bv_1 = \bv_1(\bX_\sopt)$ be the principle eigenvector 
of $\bX_\sopt$, the SDP estimate is given by 
\begin{align}
\xsdp(G) = \sign\big(\bv_1(\bX_\sopt(G))\big)\, .
\end{align}

We measure the performance of such an estimator via the overlap:
\begin{align}
\Ove_n(\xsdp)= \frac{1}{n}\E\big\{\big|\<\xsdp(G),\bxz\>\big|\big\}\, ,\label{eq:myoverlap}
\end{align}
where $\bxz \in \{-1,+1\}^n$ encodes the ground truth memberships with $x_{0,i} = +1$ if $i\in V_+$
and $x_{0,i} = -1$ if $i\in V_-$. Note that $\Ove_n(\xsdp) \in [0,1]$
and a random guessing estimator yields overlap of order
 $O(1/\sqrt{n})$. 

The majority of our calculations were run on a cluster with $160$ cores
(Intel Xeon), taking roughly a month (hence total CPU time was roughly
10 years).

%======
\subsection{Optimization algorithms}

We write SDP optimization~\eqref{eq:SDP} in terms of the vector spin model. Let $\bsigma_i\in \reals^m$
be the vector spin assigned to node $i$, for $i\in [n]$, and define $\bsigma\equiv (\bsigma_1,\bsigma_2,\dotsc,\bsigma_n)$. 
We then rewrite the SDP~\eqref{eq:SDP} as 
\begin{align}
\begin{split}\label{eq:NonConvex_NEW}
\underset{\usigma}{\mbox{maximize}} & \;\;\;F(\bsigma) \equiv \sum_{(i,j)\in E}\<\bsigma_i,\bsigma_j\>\, ,\\
\mbox{subject to} & \;\;\;\bsigma \in \cM(n,m)\, ,
\end{split}
\end{align}
where the manifold $\cM(n,m)$ is defined as below:
\begin{align}
\cM(n,m) = \Big\{\bsigma = (\bsigma_1,\dots,\bsigma_n)\in (\reals^m)^n:\; \|\bsigma_i\|_2
= 1\, ,\;\; \sum_{i=1}^n \bsigma_i = 0\Big\}\, .
\end{align}
We will omit the dimensions when they are clear from the context. 

As discussed in the main text, the two optimization problems~\eqref{eq:SDP} and~\eqref{eq:NonConvex_NEW} 
have a value that differ by a relative error of $O(1/m)$, uniformly in
the size $n$. In particular, the asymptotic value of the SDP is the
same, if we let  $m\to \infty$ \emph{after} $n\to \infty$.

In fact the following empirical findings (further discussed below) point at a much stronger
connection:
\begin{enumerate}
\item  With high probability, the optimizer appears to be essentially independent of $m$
  already for moderate values of $m$ (in practice, already for $m =
  40\sim 100$, when $n \lesssim 10^4$).
\item Again, for moderate values of $m$, optimization methods do not
  appear to be stuck in local minima. Roughly speaking, while the
  problem is non-convex from a worst case perspective, typical
  instances are nearly convex.
\end{enumerate}
Motivated by these findings, we solve optimization problem~\eqref{eq:NonConvex_NEW} 
in lieu of SDP problem~\eqref{eq:SDP}, using the two algorithms described
below: $(i)$  Projected gradient ascent; $(ii)$ Block-coordinate ascent.

The rank-constrained formulation also allows to accelerate
the rounding step to compute $\xsdp(G)$, which can be obtained in time
$O(nm^2+m^3)$, instead of the naive $O(n^3)$.
Namely, given an optimizer $\bsigma^\sopt$, we compute the $m \times
m$ empirical covariance matrix 
\begin{equation}
\bhSigma\equiv\frac{1}{n} \sum_{i=1}^n \bsigma_i^\sopt
(\bsigma_i^\sopt)^\sT\,.
\end{equation}
Denoting by $\bphi$ the principal eigenvector of $\bhSigma$, we obtain the
estimator $\xsdp(G)\in\{1,-1\}^n$ via
\begin{align}\label{eq:fromMto2}
\hat{x}^{\mbox{\tiny{SDP}}}_i = \sign\big(\<\bphi,\bsigma^\sopt_i\>\big)\, .
\end{align}

This approach allows us to carry out high-precision simulations for
large instances, namely up to $n=64,000$. By comparison, standard SDP
solvers are based on interior-point methods  and cannot scale beyond
$n$ of the order of a few hundreds.

\subsubsection{Projected gradient ascent}\label{sec:ProjM}

The tangent space at $\bsigma \in \cM$ is given by
\begin{align}\label{eq:tangent}
\sT_{\bsigma}\cM \simeq \Big\{\bz = (\bz_1,\dots,\bz_n):\,
\<\bsigma_i,\bz_i\> = 0\, ,\; \sum_{i=1}^n\bz_i = 0\Big\}\, .
\end{align}
Define the orthogonal projectors in $\reals^m$:
\begin{align}
\P_i \equiv  \bsigma_i\bsigma_i^{\sT}\, ,\;\;\;
\Pp_i \equiv  \id-\bsigma_i\bsigma_i^{\sT}\, .
\end{align}
By identification~\eqref{eq:tangent}, the manifold gradient of $F$ reads
\begin{align}
\begin{split}
\nabla F(\bsigma) & = \big(\nabla F(\bsigma)_1,\dots,\nabla
F(\bsigma)_n\big)\, ,\\
\nabla F(\bsigma)_i & = \Pp_i(\bv_i-\obv_i)\, ,
\end{split}
\end{align}
where
\begin{align}
\begin{split}
\bv_i(\bsigma) & \equiv \sum_{j\in \di} \bsigma_j\, ,\label{VelDef}\\
\obv & \equiv
\Big(\sum_{i=1}^n\Pp_i\Big)^{-1}\Big(\sum_{i=1}^n\Pp_i\bv_i\Big)\, .
\end{split}
\end{align}

We next define the convex envelope of $\cM$:
\begin{align}
\conv(\cM) = \Big\{\bsigma = (\bsigma_1,\dots,\bsigma_n)\in (\reals^m)^n:\; \|\bsigma_i\|
\le 1\, ,\;\; \sum_{i=1}^n \bsigma_i = 0\Big\}\, ,
\end{align}
and the corresponding orthogonal projector
\begin{align}\label{eq:PcM}
\P_{\cM}(\by) \equiv\arg\min_{\bz\in\conv(\cM)}\|\by-\bz\|_2\, .
\end{align}
The projected gradient method alternates between a step in the
direction of $\nabla F(\bsigma)$ and a projection onto $\conv(\cM)$.
Pseudocode is given as Algorithm \ref{alg:NonConvex}.

%**********************************************
\begin{algorithm}[]
\caption{Projected gradient ascent for relaxed min bisection}\label{alg:NonConvex}
\begin{algorithmic}[1]
\REQUIRE $n$,$m$, edge set $E \subseteq [n]\times[n]$, $\tol_2$
\ENSURE $\bsigma^\sopt \equiv (\bsigma_1^\sopt ,\dotsc, \bsigma_n^\sopt)$
\STATE Initialize $\bsigma^0\in\cM$ at random\\
 (e.g. by letting $\bsigma_i\in \{\pm\be_1,\pm\be_2,\dotsc,
 \pm\be_m\}$ uniformly random conditional on
  $\sum_{i=1}^n\bsigma_i^0=0$).

\FOR { $t=0,1,2, \dots$}
\STATE Gradient step
\begin{align}
\boldsymbol{\tilde{\bsigma}}^{t+1} = \bsigma^t+ \eps_t\nabla F(\bsigma^t)\, ,
\end{align}
with $\eps_t = 1/\sqrt{t}$ the step size.
\STATE Projection step:
\begin{align}
\bsigma^{t+1} = \P_{\cM}(\boldsymbol{\tilde{\bsigma}}^{t+1})\,.
\end{align}
\IF {$\|\nabla F(\bsigma^{t+1})\|_2/\sqrt{n} \le \tol_2$}
\STATE set $T=t$
\STATE{\rm break}
\ENDIF
\ENDFOR
\RETURN 
$\bsigma^\sopt = \bsigma^{t}$
\end{algorithmic}
\end{algorithm}

The projected gradient method requires a subroutine for computing the
projection $\P_{\cM}$ onto $\conv(\cM)$. 
In order to compute this projection, we write the Lagrangian
corresponding to problem \eqref{eq:PcM}:
\begin{align}
\mathcal{L} \equiv \frac{1}{2} \sum_{i=1}^n \|\by_i-\bz_i\|_2^2 + \bw^\sT \sum_{i=1}^n \bz_i + \sum_{i=1}^n \frac{\mu_i}{2} (\|\bz_i\|_2^2-1)\,,
\end{align}
with $\mu \ge 0$ for $1\le i\le n$.
Setting $\nabla_{\bz_i} \mathcal{L} = 0$, we obtain
\begin{align}
\bz_i = \frac{\by_i - \bw}{\mu_i+1}\,.
\end{align}
Further, the constraint $\sum_{i=1}^n \bz_i = 0$ implies
\begin{align}
\bw = \Big(\sum_{i=1}^n\frac{1}{\mu_i+1}\Big)^{-1}\Big(\sum_{i=1}^n\frac{\by_i}{\mu_i+1}\Big)\,.\label{eq:PcM2}
\end{align}
Due to constraint $\|\bz_i\|\le 1$, we have $\mu_i \ge \|\by_i-\bw\|-1$. Also, by the KKT conditions, if the inequality is strict we have $\mu_i = 0$. Therefore,
\begin{align}
\mu_i = \max(\|\by_i-\bw\|_2-1,0)\,.\label{eq:PcM3}
\end{align}
Substituting for $\mu_i$ from Eq.~\eqref{eq:PcM3} into~\eqref{eq:PcM2} we arrive at
\begin{align}
 \bw =\Big(\sum_{i=1}^n\frac{1}{\max(\|\by_i-\bw\|_2,1)}\Big)^{-1}\Big(\sum_{i=1}^n\frac{\by_i}{\max(|\by_i-\bw\|,1)}\Big)\,.
\end{align}
We compute the  Lagrange multiplier $\bw$ in an iterative manner as described in Algorithm~\ref{alg:PcM}.
\vspace{0.2cm}

%**********************************************
\begin{algorithm}[]
\caption{Algorithm for computing the projection $\P_{\cM}(\by)$ }\label{alg:PcM}
\begin{algorithmic}[1]
\REQUIRE $\by$, $\tol_1$
\ENSURE $\P_{\cM}(\by)$
\STATE $\bw^0 \leftarrow \mathbf{0}$
\FOR { $t=0,1,2, \dots$}
\STATE update 
 \begin{align}
 \bw^{t+1} =\Big(\sum_{i=1}^n\frac{1}{\max(\|\by_i-\bw^t\|_2, 1)}\Big)^{-1}\Big(\sum_{i=1}^n\frac{\by_i}{\max(\|\by_i-\bw^t\|_2,1)}\Big)\,.
 \end{align}
 \IF {$\|\bw^{t+1} -\bw^t\|_2 \le \tol_1$}
 \STATE set $T=t$
 \STATE{\rm break}
 \ENDIF
 \ENDFOR
\RETURN 
\begin{align}
\P_{\cM}(\by)_i = \frac{\by_i-\bw^T}{\max(\|\by_i-\bw^T\|,1)}\,,\quad
\quad  {\rm for } \quad i=\{1,2,\dotsc,n\}\,.
\end{align}
\end{algorithmic}
\end{algorithm}
%******************************************

%******************************************

\subsubsection{Block coordinate ascent}\label{sec:greedyDyn}

We present here a second algorithm to solve
problem~\eqref{eq:NonConvex_NEW}, that uses block-coordinate descent.
This provides independent check of numerical results. Further, this
second method appears to be faster than projected gradient ascent.
%(particularly because it does not require to optimize the stepsize).

We start by considering an unconstrained version of the optimization
problem (with $\bA_G$ the adjacency matrix of the graph $G$):
\begin{align}\label{eq:maxWithConstr}
\mbox{maximize}\;\;&  \<\bsigma,\big(\bA_G-\eta\bone\bone^{\sT}\big)\bsigma\>\,.
\end{align}
Equivalently, this objective function can be written as $F(\usigma) -
\eta\|\bM(\bsigma)\|^2_2/2$, where $\bM(\bsigma) \equiv \sum_{i=1}^n\bsigma_i$.
As $\eta\to\infty$, this is of course equivalent to problem
\eqref{eq:NonConvex_NEW}. 

For $G$ distributed according to  the hidden partition model
(\ref{eq:SBMDefinition}), with random vertex labels $\bxz$, we have
$\E\{(\bA_G)_{ij}\} = d/n$. This suggests that $\eta\ge d/n$ should be
sufficient to obtain a balanced partition. In practice we  will take
$\eta \approx 1$ and check that this yields well balanced partitions
(i.e. $\bM \approx 0$), cf. Section \ref{sec:BlockParam} below.

We maximize the objective by iteratively maximizing over each of the
vectors $\bsigma_i$. The latter optimization has a close form
expression. More precisely, at each step of this dynamics, we sort
the variables in a random order and we update them sequentially, by
maximizing the objective function. This is easily done by aligning
$\bsigma_i$ along the `local field'
\begin{align}
\bh_i \equiv -\eta\, \bM(\usigma) + \sum_{j:(i,j)\in E} \bsigma_j \,.
\end{align}
We check the convergence  by measuring the largest variation in a spin
variable during the last iteration. Namely we define
\begin{align}
\Delta_\text{max}(t) =& \max_{i\in [n]} \|\bsigma_i^{t+1} - \bsigma_i^t\|_2\, ,
\end{align}
and we use as convergence criterion $\Delta_\text{max} < \tol_3$.
The corresponding pseudocode is presented as
Algorithm~\ref{alg:greedy}.

The resulting algorithm is very simple and depends on two parameters ($\eta$
and $\tol_3$) that will be discussed in the
Section~\ref{sec:numGreedy}, 
together with dependence on the number $m$ of spin components.

%**********************************************
\begin{algorithm}[]
\caption{Block coordinate  ascent for relaxed min bisection}\label{alg:greedy}
\begin{algorithmic}[1]
\REQUIRE $n$, $m$, edge set $E \subseteq [n]\times[n]$, $\eta$, $\tol_3$
\ENSURE $\usigma^\sopt \equiv (\bsigma_1^\sopt ,\dotsc, \bsigma_n^\sopt)$
\STATE Initialize $\usigma^0\in\cM$ at random\\
 (e.g. by letting $\bsigma_i\in \{\pm\be_1,\pm\be_2,\dotsc, \pm\be_m\}$ at random so that $\sum_i\bsigma_i^0=0$).
\FOR { $t=1,2, \dots$}

\STATE Choose a random permutation $\pi:[n]\to[n]$

\FOR { $i=1,2,\dots,n$}

\STATE Update spins, aligning to the local field

\begin{align}
\bsigma^{t+1}_{\pi(i)} = \frac{-\eta\, \bM(\bsigma^t) + \sum_{j:(\pi(i),j)\in E} \bsigma^t_j}
{\|-\eta\, \bM(\bsigma^t) + \sum_{j:(\pi(i),j)\in E} \bsigma^t_j\|_2}\,,
\end{align}

\ENDFOR

\IF {$\Delta_{\text{max}}(t) \le \tol_3$}
\STATE{set $T=t$}
\STATE{\rm break}
\ENDIF
\ENDFOR
\RETURN 
$\bsigma^\sopt = \bsigma^T$
\end{algorithmic}
\end{algorithm}
%******************************************

\subsection{Numerical experiments with the projected gradient ascent}
\label{sec:numAdel}

In this section we report our results with the projected gradient
algorithm. As mentioned above, we found that the block coordinate
ascent method was somewhat faster, and therefore we used the latter
for large-statistics simulations, and high-precision determinations of
the critical point $\lsdp(d)$. We defer to the next section for
further discussion.

We use the projected gradient ascent discussed in Section~\ref{sec:ProjM}, with $\tol_1 = \tol_2 = 10^{-6}$. 
For each value of $d \in \{5, 10, 15 , 20 , 25, 30\}$ and
$n\in\{2000$, $4000$, $8000$, $16000\}$, we generate $500$
realizations of graph $G$ from the stochastic block model defined in
Eq.~(\ref{eq:SBMDefinition}). 
In these experiments, we observed that the estimated membership vector
does not change for $m\ge 40$, cf. Section~\ref{sec:m}. 
The results reported here correspond to $m = 40$.

Figure~\ref{fig:SDP_OL} reports the estimated overlap $\Ove_n(\xsdp)$
(across realizations) achieved by the SDP estimator, for different
values of $d$ and $n$. 
The solid curve corresponds
to the cavity prediction, cf. equation~\eqref{eq:SDP-OL}, for large $d$. As we see the empirical results are
in good agreement with the analytical curve even for small average
degrees $d = 5, 10$. 

In particular, \emph{the phase transition location seems indistinguishable,
on this scale, from $(a-b)/\sqrt{2(a+b)}=1$.}

\begin{figure}[]
\centering
\captionsetup[subfigure]{aboveskip=5pt}
\begin{subfigure}[t]{0.47\textwidth}
\includegraphics*[width =2.8in]{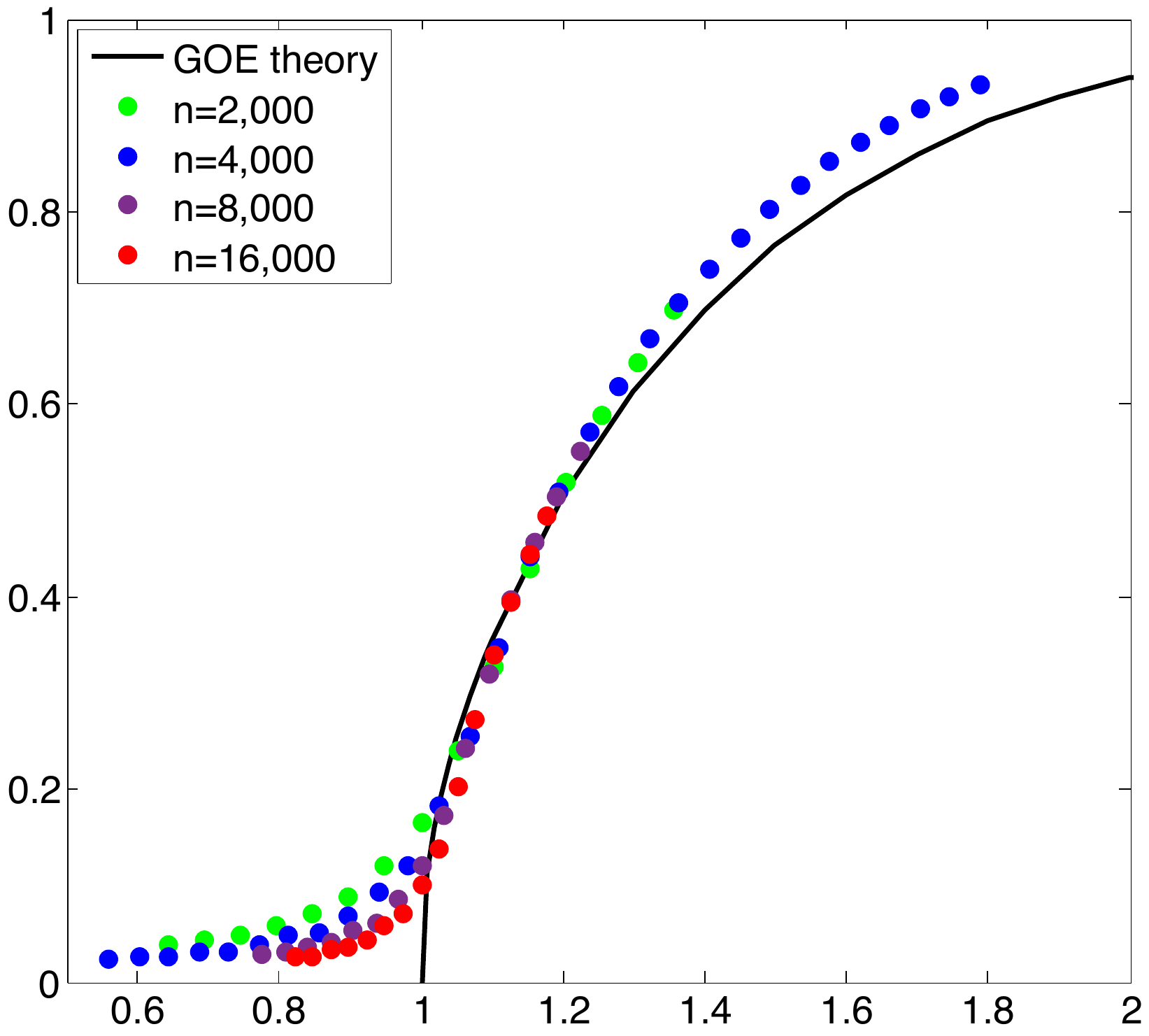}
\put(-215,77){\rotatebox{90}{\small {\sf Overlap}}}
\put(-120,-10){{\scriptsize $a-b/\sqrt{2(a+b)}$}}
\caption{$d = 5$}
\end{subfigure}
\hspace{0.7cm}
\begin{subfigure}[t]{0.47\textwidth}
\includegraphics[width =2.8in]{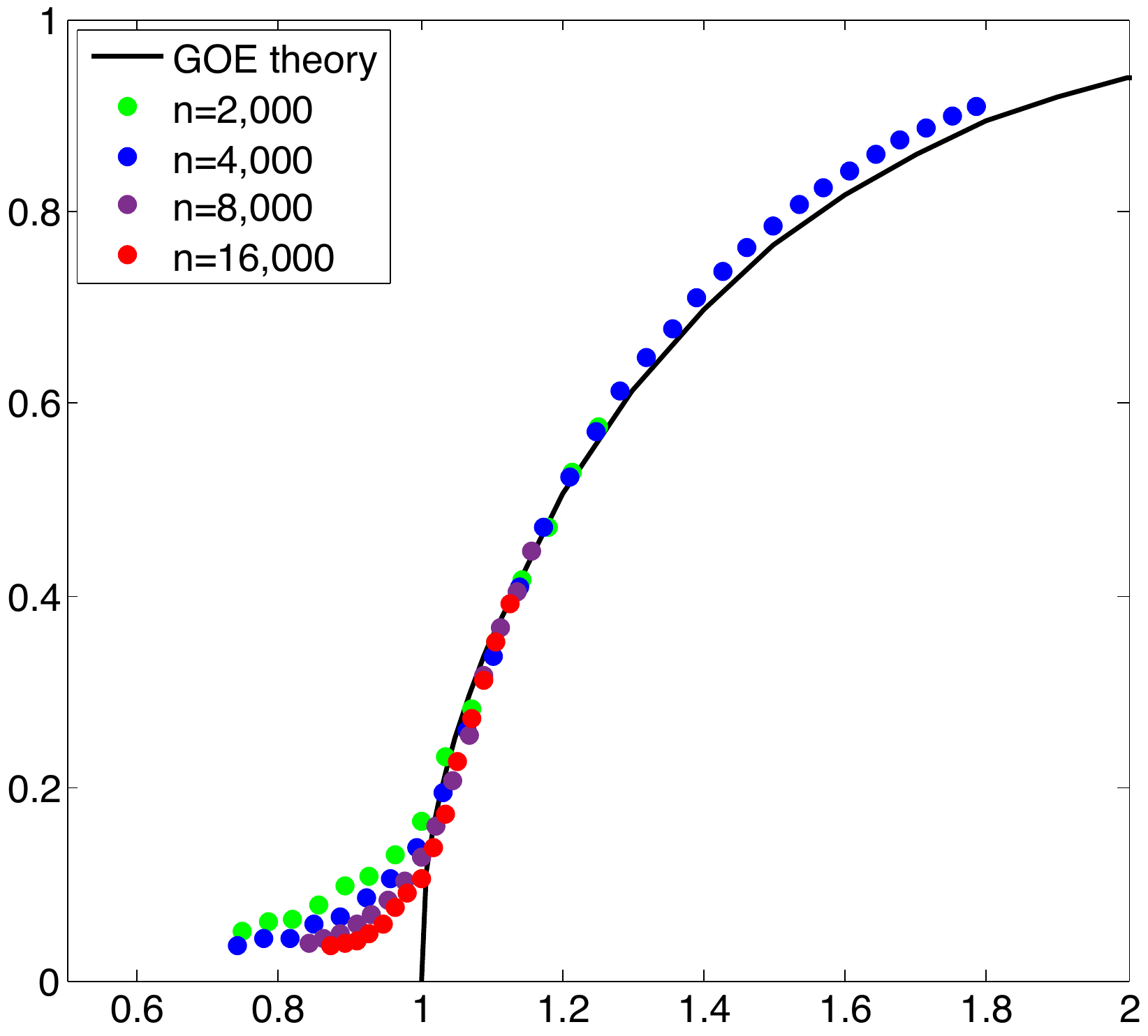}
\put(-215,77){\rotatebox{90}{\small {\sf Overlap}}}
\put(-120,-10){{\scriptsize $a-b/\sqrt{2(a+b)}$}}t
\caption{$d = 10$}
\end{subfigure}
\begin{subfigure}[t]{0.47\textwidth}
\includegraphics[width =2.8in]{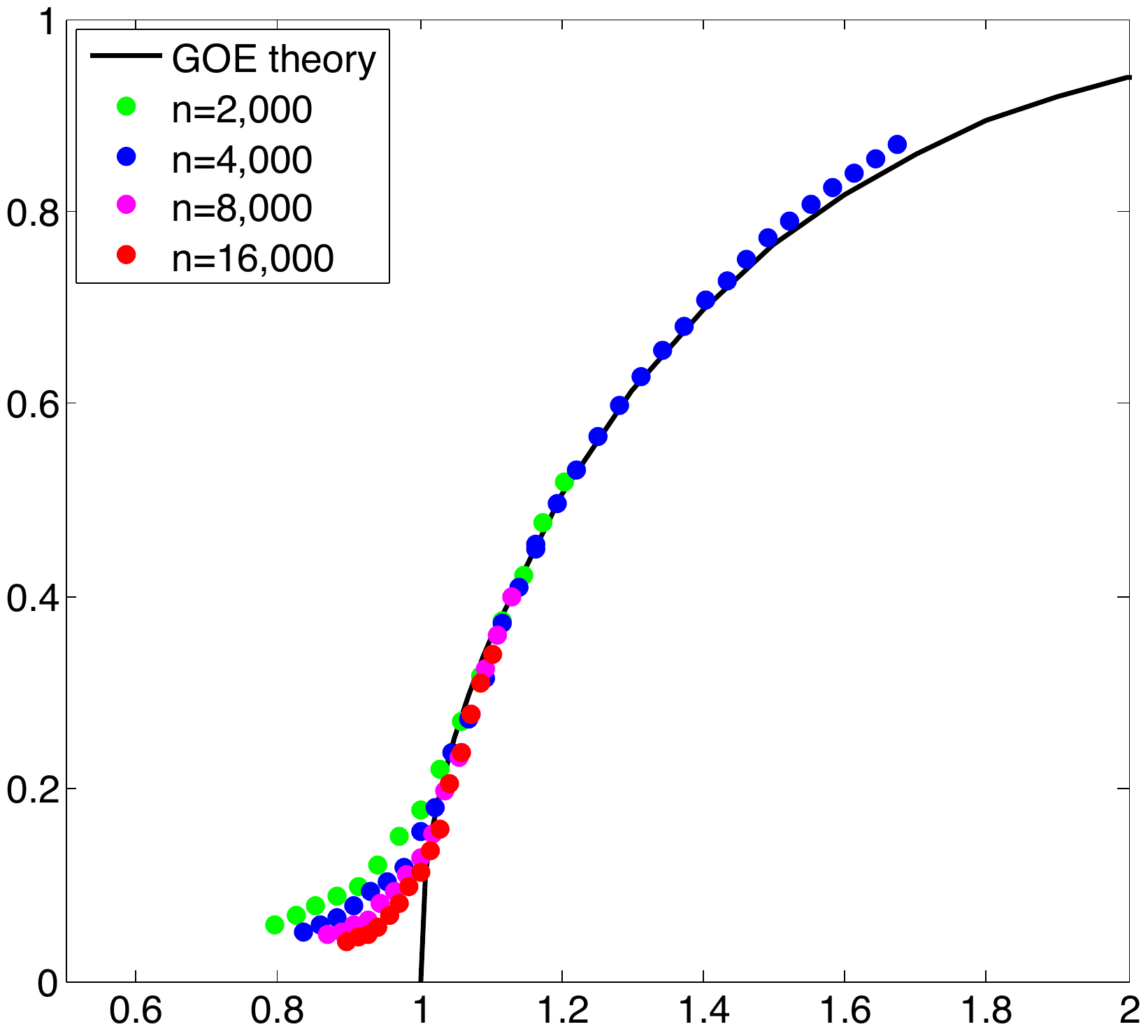}
\put(-215,77){\rotatebox{90}{\small {\sf Overlap}}}
\put(-120,-10){{\scriptsize $a-b/\sqrt{2(a+b)}$}}
\caption{$d = 15$}
\end{subfigure}
\hspace{0.7cm}
\begin{subfigure}[t]{0.47\textwidth}
\includegraphics[width =2.8in]{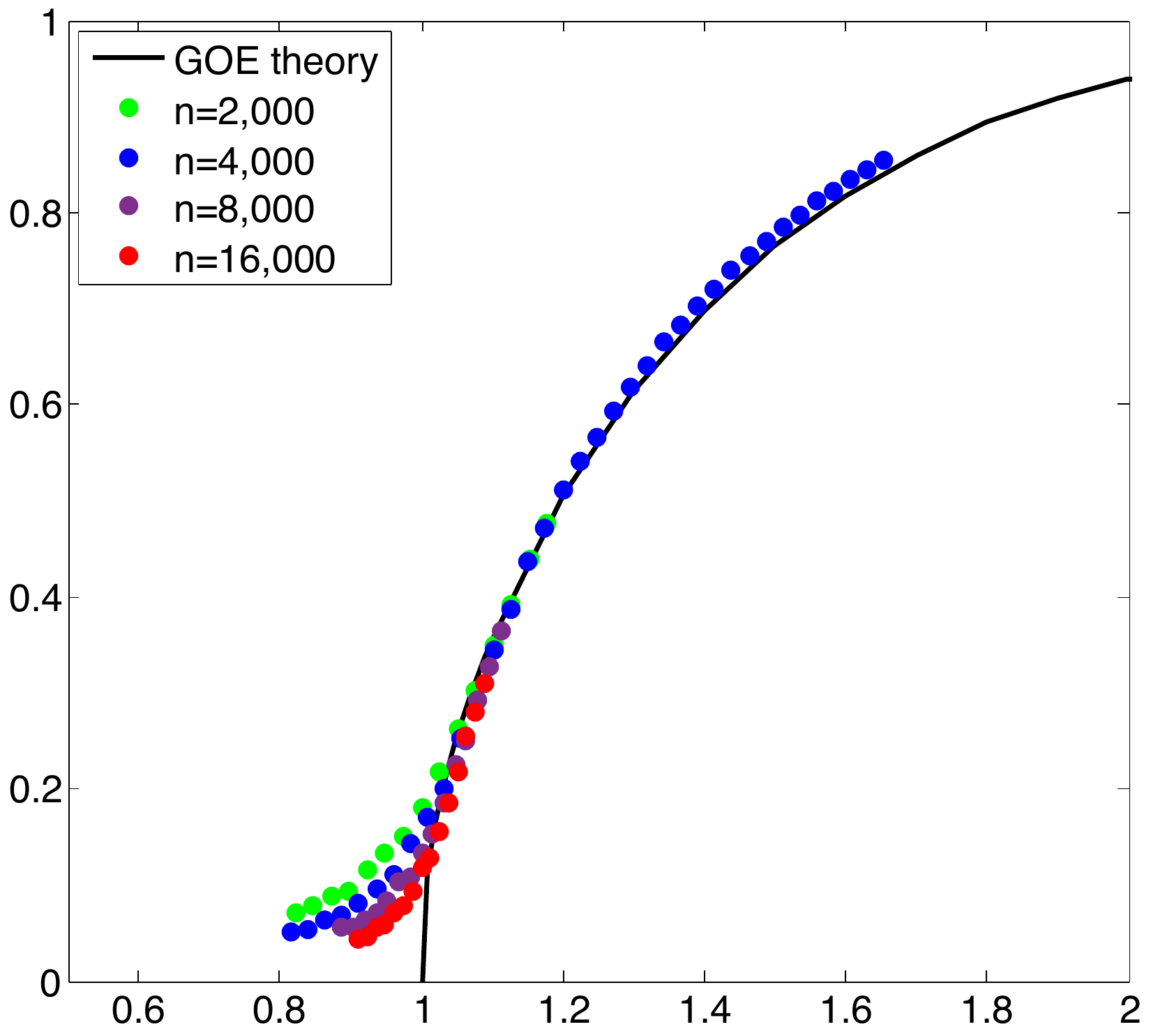}
\put(-215,77){\rotatebox{90}{\small {\sf Overlap}}}
\put(-120,-10){{\scriptsize $a-b/\sqrt{2(a+b)}$}}
\caption{$d = 20$}
\end{subfigure}

\begin{subfigure}[t]{0.47\textwidth}
\includegraphics[width =2.8in]{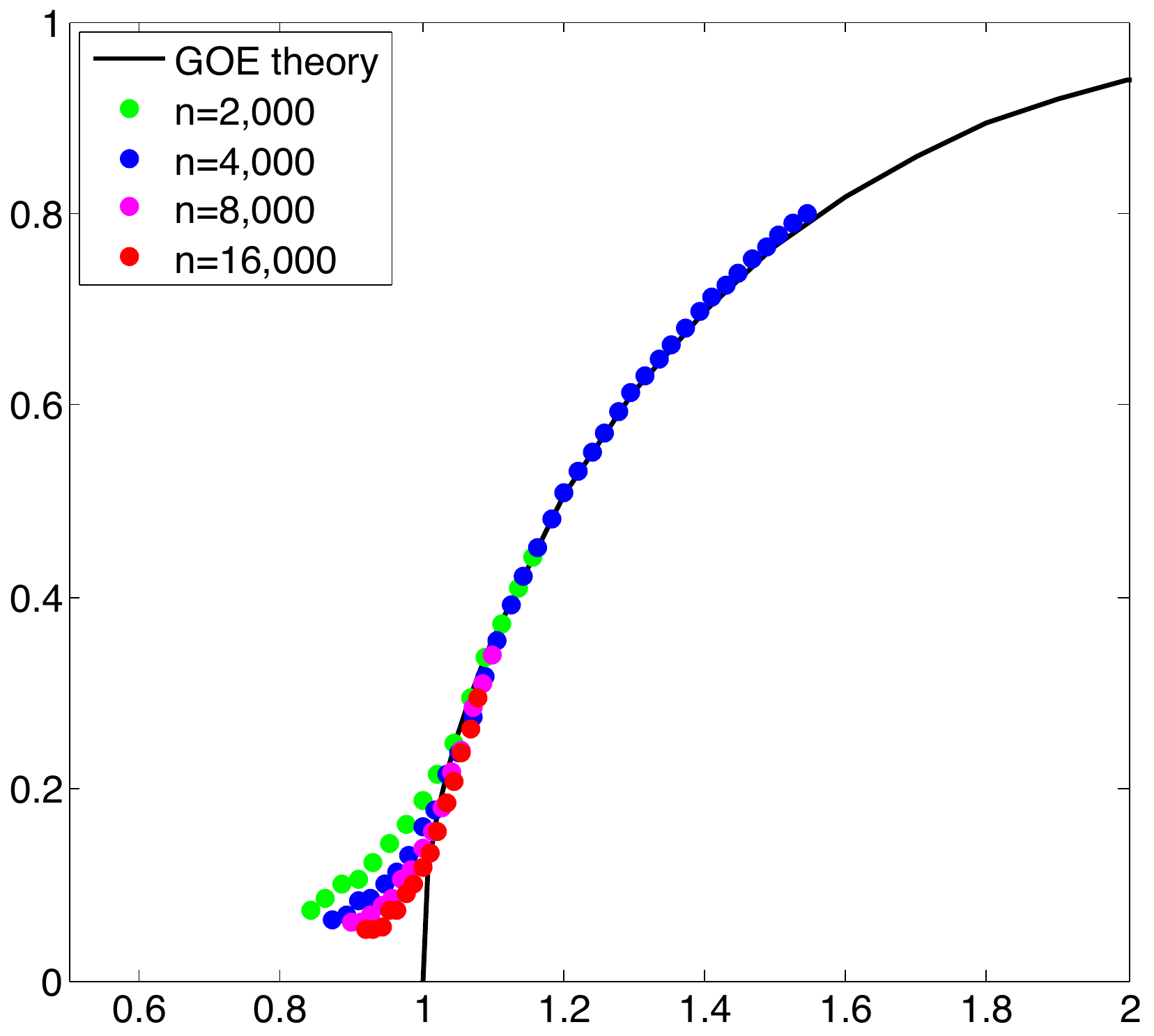}
\put(-215,77){\rotatebox{90}{\small {\sf Overlap}}}
\put(-120,-10){{\scriptsize $a-b/\sqrt{2(a+b)}$}}
\caption{$d = 25$}
\end{subfigure}
\hspace{0.7cm}
\begin{subfigure}[t]{0.47\textwidth}
\includegraphics[width =2.8in]{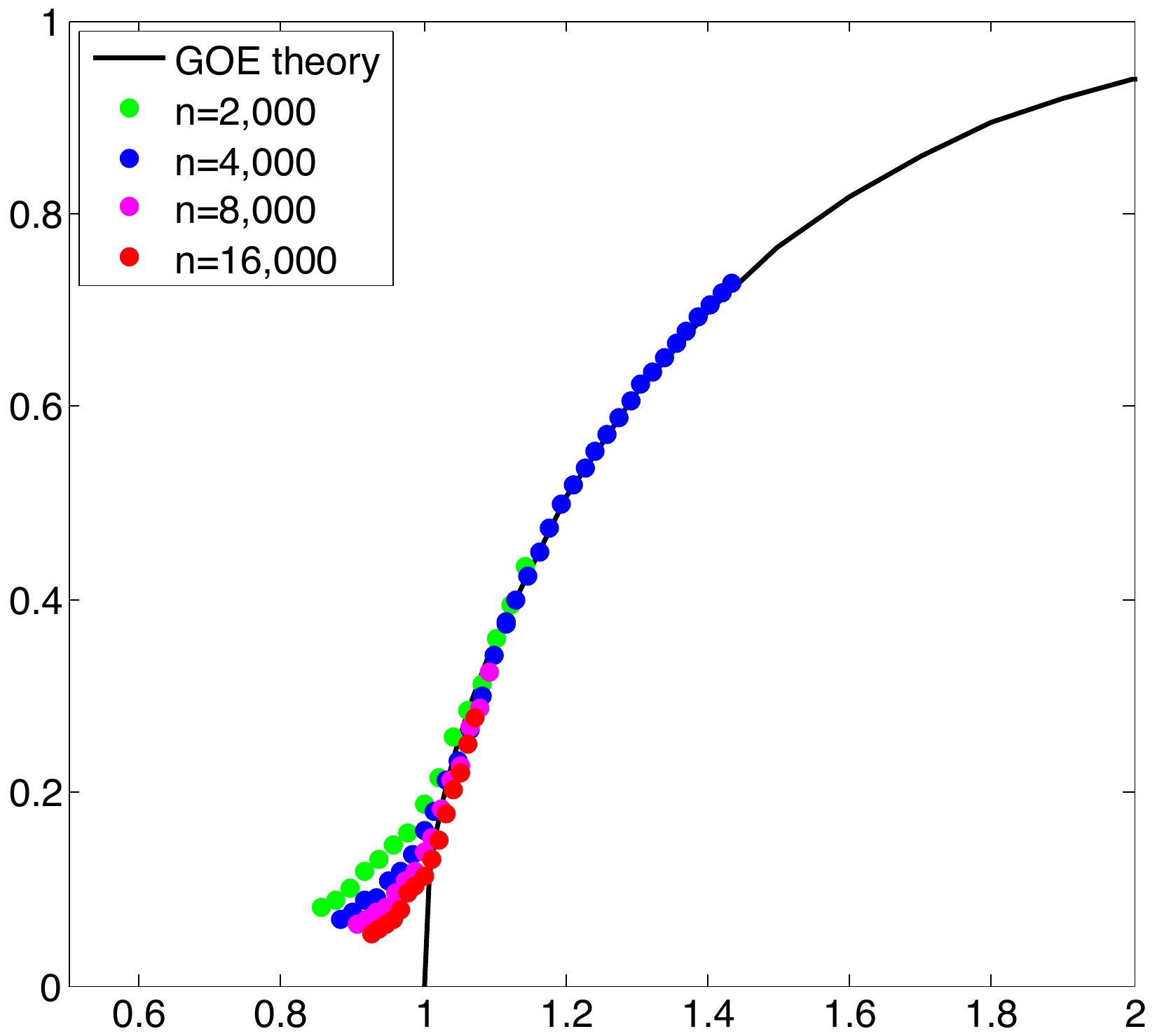}
\put(-215,77){\rotatebox{90}{\small {\sf Overlap}}}
\put(-120,-10){{\scriptsize $a-b/\sqrt{2(a+b)}$}}
\caption{$d = 30$}
\end{subfigure}
\vspace{3mm}
\caption{Community detection under the hidden partition model of
  Eq. \eqref{eq:SBMDefinition},
for several average degree $d =(a+b)/2$. Dots corresponds to the
performance of the SDP reconstruction method (averaged over $500$
realizations). The continuous curve is the asymptotic analytical
prediction for the Gaussian model (which captures the large-degree behavior).\label{fig:SDP_OL}}
\end{figure}
%
%***********************************
\subsubsection{Dependence on $m$}\label{sec:m}

As we explained before, optimization problem~\eqref{eq:NonConvex_NEW} and SDP~\eqref{eq:SDP} are equivalent provided $m\ge n$.
In principle, one can solve~\eqref{eq:NonConvex_NEW} applying Algorithm~\ref{alg:NonConvex} with $m = n$. 
However, this choice leads to a computationally 
expensive procedure. On the other hand, we expect the solution to be
essentially independent of  $m$ already for moderate values of $m$.
Several theoretical arguments point to this (in particular, the
Grothendieck-type inequality of \cite{montanari2015semidefinite}).
We provide numerical evidence in this section (supporting in
particular the choice $m=40$).

In the first experiment, we set the average degree $d = (a+b)/2 = 5$, $n = 4000$ 
and vary $\lambda = (a-b)/\sqrt{2(a+b)} \in \{0.9,1,1.1,1.2\}$.
For each $\lambda$, we solve for $a$ and $b$ and generate $100$ realizations of the graph as per model~\eqref{eq:SBMDefinition} with parameters $a, b$.
For several values of $m$, we solve optimization~\eqref{eq:NonConvex_NEW} and report the average overlap and its standard deviation. The results are summarized in Table~\ref{tbl:diffd_5}. 
As we see, for $m\ge 40$ the changes in average overlaps are comparable with the corresponding standard deviations. 
An interesting observation is that error bars for smaller $m$ are larger, indicating more variations of overlaps across different realizations.

Table~\ref{tbl:diffd_10} demonstrates the results for an analogous experiment with $d = 10$. 

We observe a similar trend for other several values of $a, b, n$. Based on these observations, we use $m = 40$ in our numerical experiments throughout this section.
   
\vspace{0.5cm}

\begin{table}[h]
\begin{center}
\begin{tabular}{|c|c|c|c|c|}
\hline
$$ & $\lambda=0.9$ & $\lambda=1$ & $\lambda=1.1$ & $\lambda=1.2$\\
\hline
$m = 5$    &  $0.1231\pm 0.0037$ & $0.0898\pm 0.0059$ & $0.2637\pm 0.0074$ & $0.4627\pm 0.0017$ \\
$m = 10$  &  $0.1544\pm 0.0008$ & $0.0659\pm 0.0045$ & $0.3303\pm 0.0018$ & $0.4612\pm 0.0006$ \\
$m = 20$  &  $0.1663\pm 0.0002$ & $0.1363\pm 0.0008$ & $0.3358\pm 0.0005$ & $0.4598\pm 0.0002$ \\
$m = 40$  &  $0.1668\pm 0.0001$ & $0.1398\pm 0.0002$ & $0.3358\pm 0.0002$ & $0.4599\pm 0.0001$ \\
$m = 80$  &  $0.1669\pm 0.0001$ & $0.1397\pm 0.0002$ & $0.3358\pm 0.0002$ & $0.4596\pm 0.0001$ \\
$m =160$ &  $0.1669\pm 0.0001$ & $0.1397\pm 0.0002$ & $0.3358\pm 0.0002$ & $0.4596\pm 0.0001$ \\
\hline
\end{tabular}
\end{center}
\caption{We fix $d = (a+b)/2 = 5$ and vary $\lambda = (a-b)/\sqrt{2(a+b)} \in \{0.9,1,1.1.,1.2\}$. The reported values are the average overlaps (over 100 realizations)
with one standard deviations, for several values of $m$. As we see for $m\ge 40$, the changes in the average overlaps are comparable to the corresponding standard deviations.
\label{tbl:diffd_5}}
\end{table}

\begin{table}[h]
\begin{center}
\begin{tabular}{|c|c|c|c|c|}
\hline
$$ & $\lambda=0.9$ & $\lambda=1$ & $\lambda=1.1$ & $\lambda=1.2$\\
\hline
$m = 5  $   &  $0.0875\pm0.0035$ & $0.1385\pm 0.0072$ & $0.3760\pm 0.0021$ & $0.5045\pm 0.0008$\\
$m = 10$   &  $0.1322\pm0.0007$ & $0.1943\pm 0.0014$ & $0.3873\pm 0.0006$ & $0.5050\pm 0.0003$\\
$m = 20$   &  $0.1346\pm0.0002$ & $0.2100\pm 0.0003$ & $0.3890\pm 0.0001$ & $0.5070\pm 0.0001$\\
$m = 40$   &  $0.1353\pm0.0001$ & $0.2089\pm 0.0001$ & $0.3885\pm 0.0001$ & $0.5088\pm 0.0001$\\
$m = 80$   &  $0.1354\pm0.0000$ & $0.2090\pm 0.0000$ & $0.3885\pm 0.0000$ & $0.5089\pm 0.0000$\\
$m = 160$ &  $0.1354\pm0.0000$ & $0.2090\pm 0.0000$ & $0.3885\pm 0.0000$ & $0.5089\pm 0.0000$\\
\hline
\end{tabular}
\end{center}
\caption{We fix $d = (a+b)/2 = 10$ and vary $\lambda = (a-b)/\sqrt{2(a+b)} \in \{0.9,1,1.1.,1.2\}$. The reported values are the average overlaps (over 100 realizations)
with one standard deviations, for several values of $m$. As we see for $m\ge 40$, the changes in the average overlaps are comparable to the corresponding standard deviations.
\label{tbl:diffd_10}}
\end{table}

%
%
%***********************************
\subsubsection{Robustness and comparison with spectral methods}

Spectral methods are among the most popular nonparametric approaches
to clustering. These 
methods classify nodes according to a the eigenvectors of a matrix associated with the graph, for instance its  adjacency matrix or
Laplacian.  While standard spectral clustering works well when the
graph is sufficiently dense or is regular, it is  significantly
suboptimal for sparse graphs. The reason is that the leading
eigenvector of the adjacency matrix is localized around the high
degree nodes.\footnote{Note that for sparse stochastic block models as
  in~\eqref{eq:SBMDefinition}, node degrees do not concentrate and we
  observe highly heterogeneous degrees.} 

Recently,~\cite{krzakala2013spectral} proposed a class of very interesting spectral
methods based on the non-backtracking walk on the directed edges of
the graph $G$. The spectrum of non-backtracking matrix is more robust to 
high-degree nodes because a walk starting at a node cannot return to
it immediately.  Later,~\cite{saade2014spectral} proposed 
another spectral method, based on the Bethe Hessian operator, that is
computationally more efficient than the non-backtracking
operator. Further, the (determinant of the) Bethe Hessian is closely
related to the spectrum of the non-backtracking operator and exhibits
the same convenient properties for the aim of clustering.
Rigorous analysis of spectral methods under the model
\eqref{eq:SBMDefinition} was carried out in \cite{massoulie2014community,mossel2013proof,bordenave2015non}.
The main result of these papers is that spectral methods allow to
estimate the hidden partition significantly better than random
guessing immediately above the ideal threshold
$\lambda=(a-b)/\sqrt{2(a+b)}=1$. 

As we saw in the previous section, the threshold of the SDP-based
method is extremely close to the ideal one.
Here, we compare the Bethe Hessian algorithm with SDP approach in terms of robustness to model miss-specification.    
We perturb the hidden partition model~\eqref{eq:SBMDefinition} as follows. For a perturbation level $\alpha \in [0,1]$,
we draw $n\alpha$ vertices $i_1,i_2, \dotsc, i_{n\alpha}$ uniformly at
random and for each vertex $i_\ell$, we  add to graph $G$ connecting
all of the neighbors of $i_{\ell}$. 
This results in adding $O(nd^2 \alpha)$ edges to the underlying
model~\eqref{eq:SBMDefinition}. This perturbation mimics an important
feature of real networks that is absent from the stochastic block
model \eqref{eq:SBMDefinition}, the so-called triangle closure
property \cite{easley2010networks} (two friends of a person are often friends).

 For perturbation levels $\alpha\in \{0, 0.025, 0.05\}$, we compare
 the performance of SDP and Bethe Hessian algorithms in terms of
 Overlap, defined by~\eqref{eq:myoverlap}. 
Figure~\ref{fig:SDP-BH} summarizes the results for $n=16,000$ and
average degree $d = (a+b)/2 = 10$. 
The reported overlaps are averaged over $100$ realizations of
the model. 

In absence of any perturbation (curves $\alpha=0$), the two algorithms
have nearly equivalent performances. However, already for
$\alpha=0.025$, SDP is substantially superior.  While  SDP appears to
be rather insensitive to the perturbation,  the performance of the Bethe Hessian
algorithm is severely degraded by it. This is because the added triangles perturb the spectrum of
the non-backtracking operator (and similarly of the Bethe Hessian operator) significantly, resulting in poor classification of the nodes.  

\begin{figure}[]
\centering
\includegraphics[width=4in]{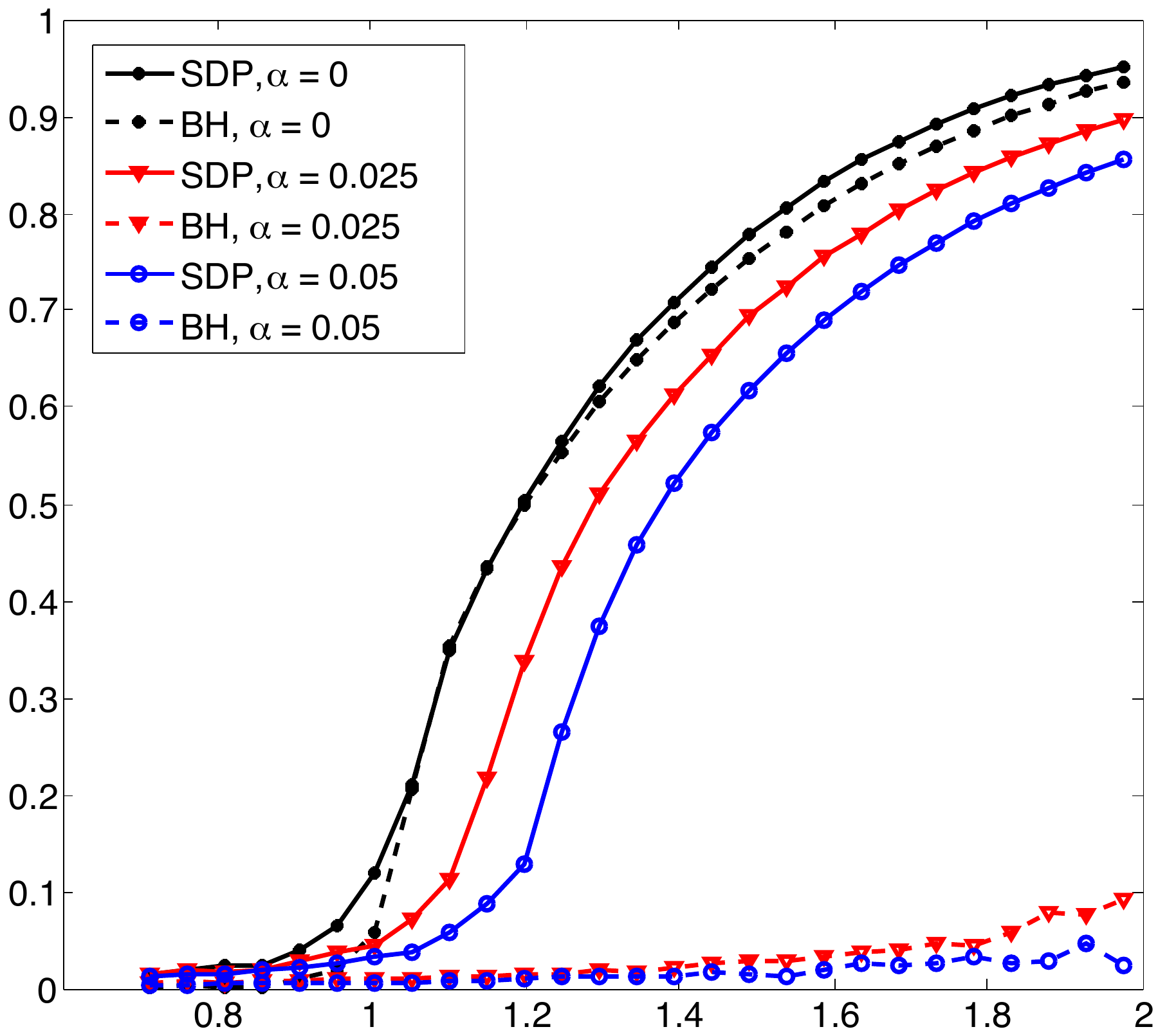}
\put(-300,120){\rotatebox{90}{\small {\sf Overlap}}}
\put(-173,-10){{\small $a-b/\sqrt{2(a+b)}$}}
\caption{Comparison of SDP and Bethe Hessian algorithm on perturbed hidden partition model.
Here, $n = 16,000$, $d = (a+b)/2 = 10$ and we report average overlaps over $100$ realizations. 
Different colors (markers) represent different perturbation levels $\alpha$. Solid curves are for SDP
algorithm and dashed curves are for Bethe-Hessian algorithm. As we see SDP algorithm is more robust than Bethe-Hessian
algorithm to the perturbation level $\alpha$.
\label{fig:SDP-BH}}
\end{figure}
%

%
%***********************************

\subsection{Numerical experiments with block coordinate ascent}
\label{sec:numGreedy}

In this section we present our simulations  with the block
coordinate ascent algorithm, cf. Algorithm \ref{alg:greedy}.
We first discuss the choice of the algorithm parameters $\eta$ and
$\tol_3$. cf. Section \ref{sec:M_Times}. 
In Section \ref{sec:M_Times} we  analyze the dependence on the number of dimensions
$m$ and the behavior of the convergence time. We conclude by
determining the phase transition point in Section \ref{sec:LambdaC},
and comparing this location with our analytical predictions.

\subsubsection{Selection of the algorithm parameters}
\label{sec:BlockParam}

 Algorithm~\ref{alg:greedy} requires specifying the parameters $\eta$
 (that penalizes $\bM(\bsigma)\neq 0$)
 and $\tol_3$ (for the convergence criterion). In order to investigate
 the dependance on these parameters, we set $m=100$ which, as we will
 see, is large enough to approximate the behavior at $m=n$.

\begin{figure}[]
\centering
\includegraphics[width=0.7\textwidth]{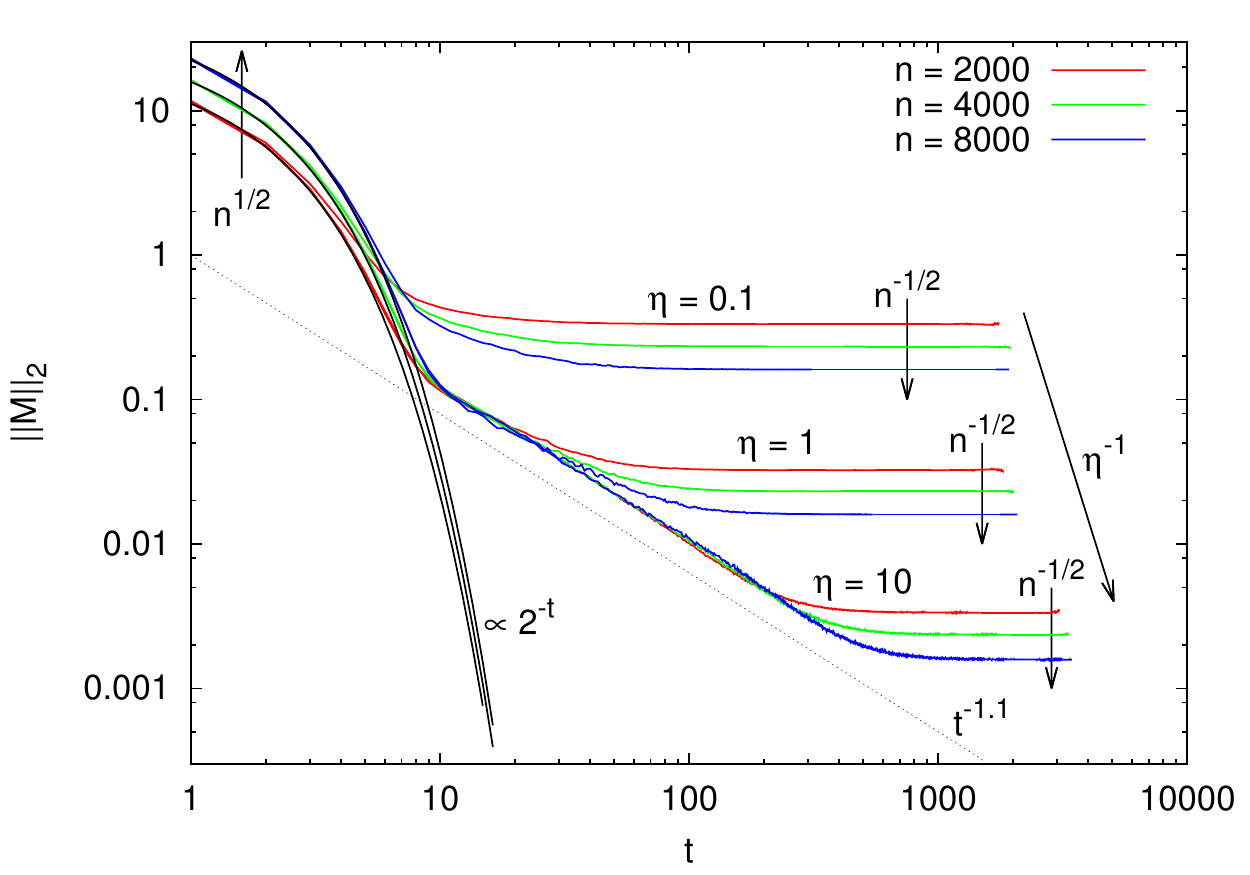}
\caption{Decrease of the norm of the total magnetization
  $\|\bM(\bsigma^t)\|_2$ during block-coordinate ascent. The starting value $\|\bM(\bsigma^0)\|_2$ is $O(\sqrt{n})$ since the initial configuration is randomly chosen. The asymptotic value is always very small for the sizes studied and decreases as $1/\eta$ and $1/\sqrt{n}$.}
\label{fig:globalMag}
\end{figure}

In Figure~\ref{fig:globalMag} we plot the evolution of the norm of the
`global magnetization,' $\|\bM(\bsigma^t)\|_2$, as a function of the
number of iterations $t$. Notice that each iteration corresponds to
$n$ updates, one update of each vector $\bsigma_i$, $i\in [n]$.
 We used $d=5$ and $\lambda=1.1$, and  we averaged over a number of samples
 ranging from $100$ (for $n=8000$) to $400$ (for $n=2000$).

We observe three regimes: 
\begin{itemize}
\item[$(i)$] Initially the magnetization decays exponentially,
  $\|\bM(\bsigma^t)\|_2 \approx  \|\bM(\bsigma^0)\|_2\: 2^{-t}$.  Further, it
  increases  slowly with $n$. Indeed from central limit theorem, we
  have $\|\bM(\bsigma^0)\|_2= \Theta(\sqrt{n})$. The same behavior
  $\|\bM(\bsigma^t)\|_2 =\Theta(\sqrt{n})$ is found empirically at
  small $t$.
\item[$(ii)$]  In an intermediate interval of times, we have a power
  law decay $\|\bM(\bsigma^t)\|_2 \propto t^{-a}$, with exponent
  $a\approx 1.1$. This intermediate regime is present only for $\eta$
  large enough.
\item[$(iii)$] For large $t$, $\|\bM(\bsigma^{t})\|_2$ reaches a
  plateau whose value scales like 
$\|\bM(\bsigma^{t})\|_2 =\Theta(1/\sqrt{n})$ with the system size and is proportional to $1/\eta$.
\end{itemize}

Already for $\eta=1$, the value of the plateau is very small, namely
\begin{align}
\left\|\sum_{i=1}^n\bsigma_i^t\right\|_2\lesssim 0.1\, .
\end{align}
Further, this value is decreasing with $n$. Given that
$\|\bsigma_i\|_2=1$, we interpret the above as evidence that the
constraint $\bX\bone=0$ is satisfied with good approximation. We will
therefore use $\eta=1$ in our simulations.

As an additional remark, notice that there is no special reason to enforce the
constraint $\bX\bone=0$ strictly. Indeed the SDP (\ref{eq:SDP_Graph})
can be replaced by
\begin{align}
\mbox{maximize} &\;\;\;\;\; \<\big(\bA_G-\eta\bone\bone^{\sT}\big),\bX\>\, ,\\
\mbox{subject to} &\;\;\;\;\; \bX\succeq 0\, ,\\
&\;\;\;\;\; X_{ii} = 1\;\; \forall i\in [n]\, ,
\end{align}
with an arbitrary value of $\eta$. Of course this is useful provided
$\eta$ is large enough to rule out the solution
$\bX=\bone\bone^{\sT}$. As mentioned above, $\eta\ge d/n$ should be already
large enough \cite{montanari2015semidefinite}.

\begin{figure}[]
\centering
\includegraphics[width=0.9\textwidth]{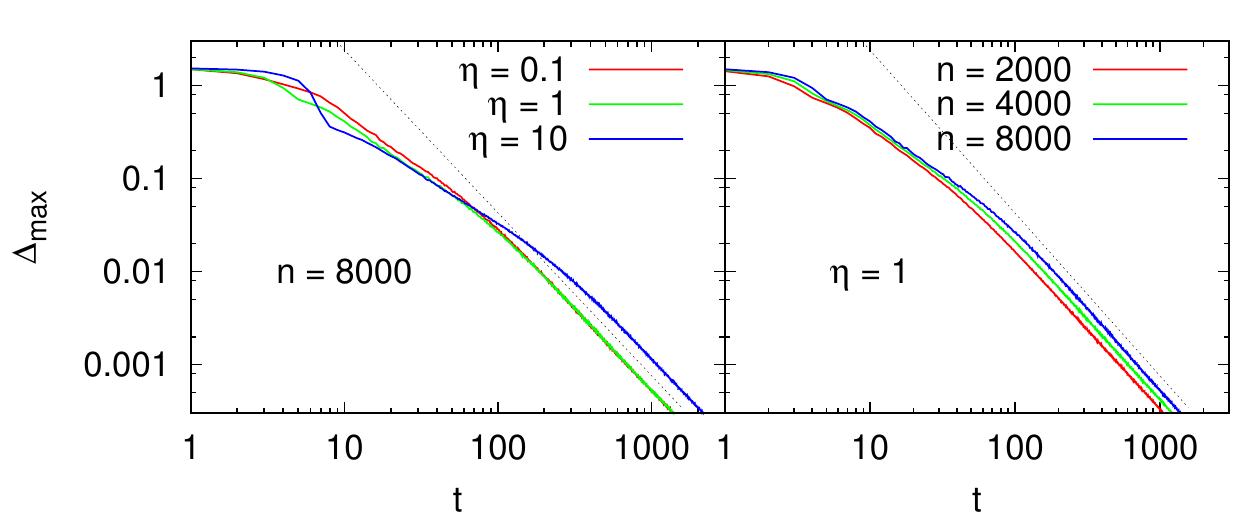}
\caption{Evolution of $\E\Delta_\text{max}(t)  = \E\max_{i\in [n]}
  \|\bsigma_i^{t+1} - \bsigma_i^t\|_2$ during block-coordinate
  ascent.  Left panel: dependence on the choice of the Lagrange
  parameter $\eta$. Right panel: dependence on the number of vertices.}
\label{fig:error}
\end{figure}

\begin{figure}[]
\centering
\includegraphics[width=0.7\textwidth]{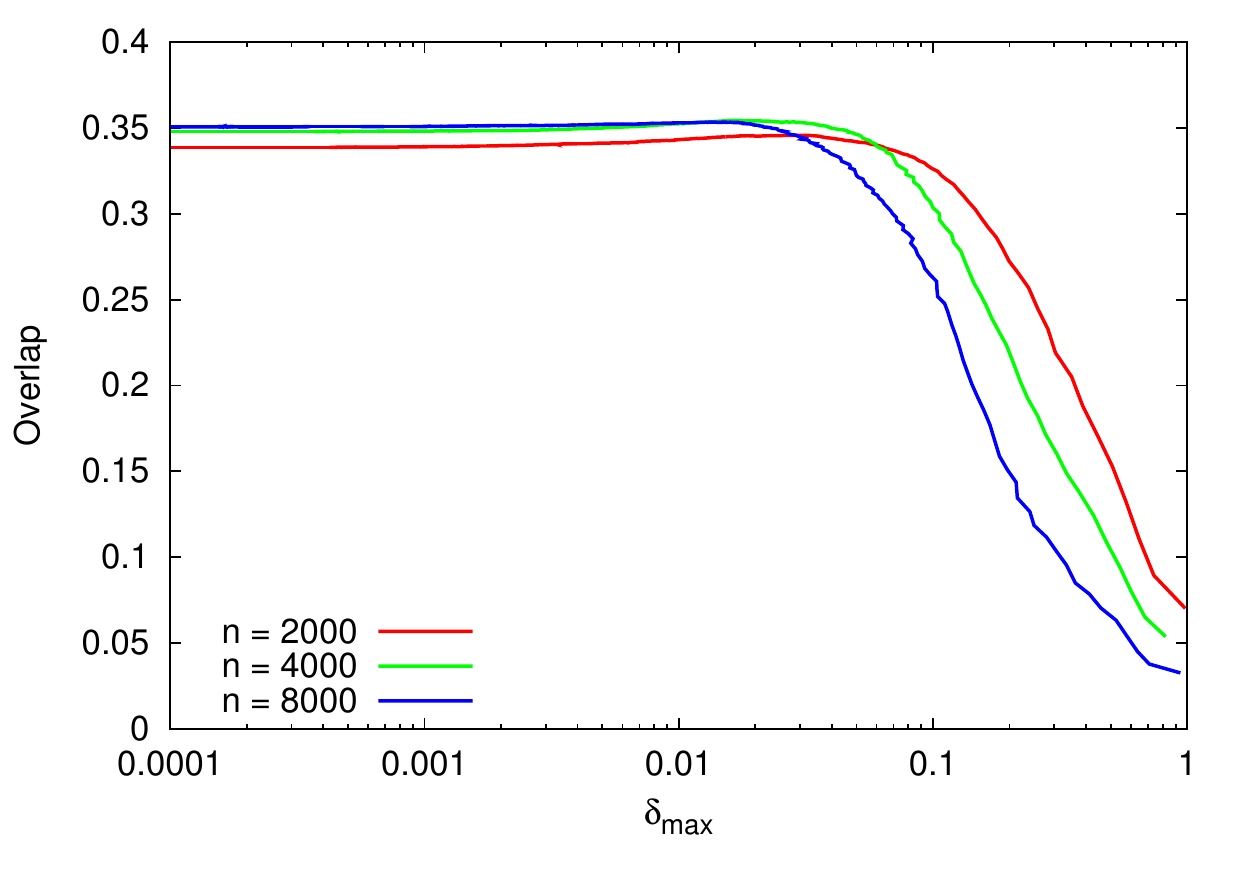}
\caption{The mean overlap does not depend on the value of $\E\Delta_\text{max}$ at which the greedy algorithm is stopped, as soon as $\delta_\text{max} \lesssim 10^{-3}$.}
\label{fig:overlapVSdelta}
\end{figure}

In Figure~\ref{fig:error} we show how $\DeltaMax(t)$ decreases with
time in Algorithm~\ref{alg:greedy}, again with $d=5$ and
$\lambda=1.1$. 

In the left panel we fix $n=8000$ and study  the
dependence on the Lagrange parameter $\eta$. 
We observe two regimes. While for $\eta\lesssim 1$, the convergence rate is
roughly independent of $\eta$, for $\eta\gtrsim 1$, it becomes
somewhat slower with $\eta$. This supports the choice $\eta=1$. 

In the right we fix $\eta=1$ and study the dependence of the
convergence time on the graph size $n$. The number of iterations
appears to increase slowly with $n$  (see also Figure
\ref{fig:convTimesAll}).   Both datasets are consistent with a power
law convergence
\begin{align}
\Delta_\text{max}(t) \approx C(n)\, t^{-b}\, ,
\end{align}
with $b\approx 1.75$ (dotted line), and $C(n)$ polynomially increasing
with $n$
(see below for a discussion of the overall scaling of computational
complexity with $n$).

In order to select the tolerance parameter for convergence,  $\tol_3$,
we study the evolution of estimation error. Define the overlap
achieved after $t$ iteration as follows. First estimate the vertex labels
by computing the top-left singular vector of $\bsigma^t$, namely
\begin{align}
\hbx^t(G) = \sign\big(\bv_1(\bsigma^t(\bsigma^t)^{\sT})\big)\, .
\end{align}
Then define, as before
\begin{align}
\Ove_n(\hbx^t)=
  \frac{1}{n}\E\big\{\big|\<\hbx^t(G),\bxz\>\big|\big\}\, .
\end{align}
Of course, the accuracy of the SDP estimator is given by
\begin{align}
\Ove_n(\xsdp) = \lim_{t\to\infty}\Ove_n(\hbx^t)\, .
\end{align}
In Figure~\ref{fig:overlapVSdelta}, we plot $\Ove_n(\hbx^t)$ as a
function of $\E\Delta_{\max}(t)$  for several values of $n$, $d=5$ and $\lambda=1.1$.
These data suggest that $\tol_3=10^{-3}$ is small enough to
approximate the $t\to\infty$ behavior. We will fix such a value hereafter.

%
%************
%
\subsubsection{Selection of $m$ and scaling of convergence times}
\label{sec:M_Times}

The last important choice is the value of the dimension (rank)
parameter  $m$. We know from \cite{montanari2015semidefinite} that the
optimal value of the rank constrained problem (\ref{eq:NonConvex_NEW}) is
within a relative error of order $O(1/m)$ of the value of the SDP
(\ref{eq:SDP}). Also, a result by Burer and Monteiro \cite{burer2003nonlinear} implies
that, for $m\ge 2\sqrt{n}$, the objective function
(\ref{eq:NonConvex_NEW}) has no local maxima that are not also global
maxima (barring accidental degeneracies).

We empirically found that $m$ of the order of $10$ or larger is
sufficient to obtain accurate results. Through most of our simulations,
we fixed however $m=100$, and we want to provide evidence that this  is a safe choice

\begin{figure}[]
\centering
\includegraphics[width=0.9\textwidth]{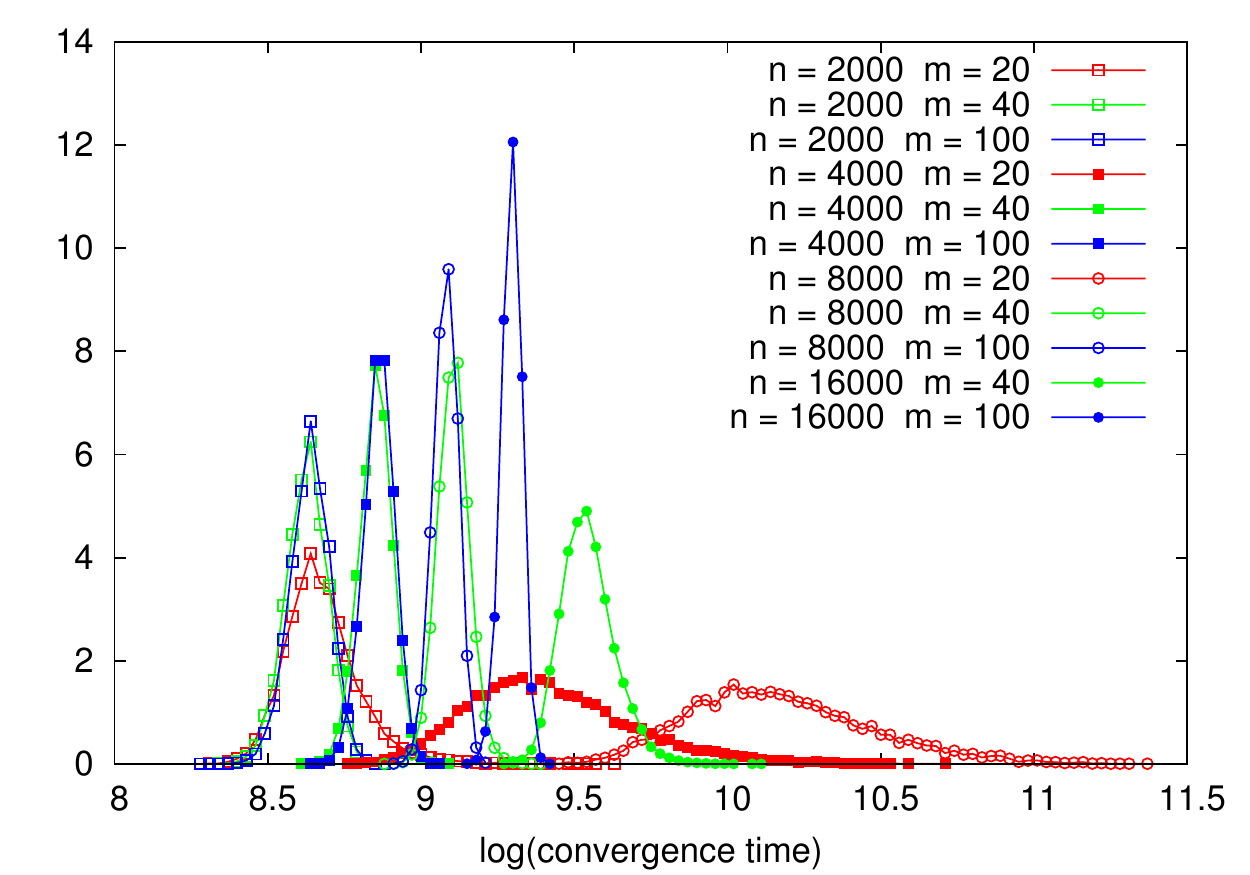}
\caption{Histograms of the logarithm of the convergence time for the greedy algorithm run with $d=10$ and $\lambda=1$. Red, green and blue histograms are for $m=20,40$ and 100 respectively. System size increases from left to right histograms.}
\label{fig:convTimesAll}
\end{figure}

\begin{figure}[]
\centering
\includegraphics[width=0.7\textwidth]{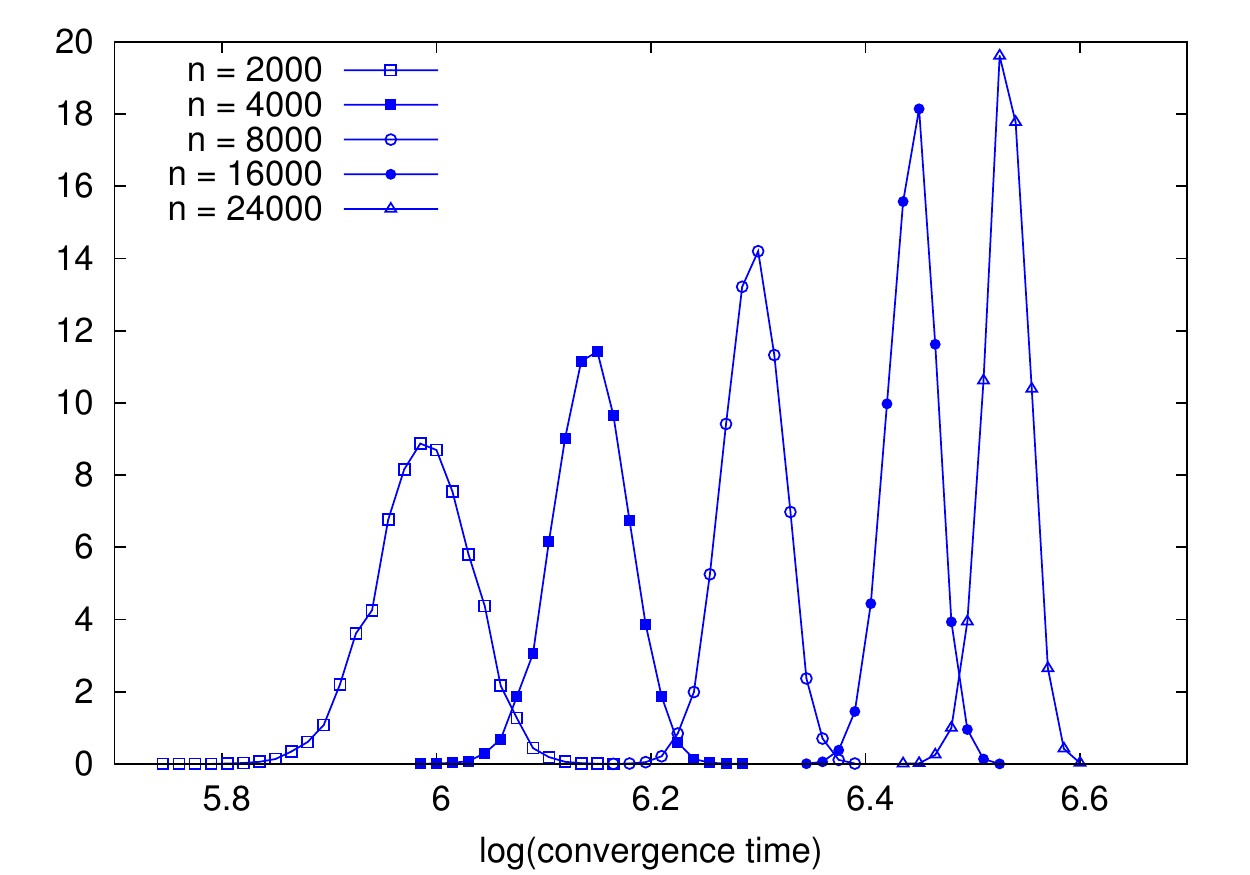}
\caption{Histograms of the logarithm of the convergence times for $m=100$, $d=10$ and $\lambda=1$.}
\label{fig:convTimesM100}
\end{figure}

For each realization of the problem we compute the convergence time
$t_\text{conv}$ as the first time such that the condition 
$\Delta_{\max}(t)\le \tol_3=10^{-3}$ is met. In Figure
\ref{fig:convTimesAll} we plot histograms of $\log(t_\text{conv})$ for $n\in\{2000,4000,8000,16000,24000\}$ and $m\in\{20,40,100\}$.
Here $d=10$ and $\lambda=1$, but $t_\text{conv}$ does not seems to
depend  strongly on $\lambda$, $d$ in the range we are interested in.

We observe that, for $m$ large enough (in particular $m=100$, see also
data in Figure~\ref{fig:convTimesM100}), the
histogram of  $\log(t_\text{conv})$ concentrates around its median.
We interpret this as evidence of convergence towards a well defined
global minimum, whose properties concentrate for $n$ large.
On the other hand,   for $m$ small, e.g. $m=20$, the  histogram broadens
as $n$ increases. This is a typical signature of convergence towards
local minima, whose properties fluctuate from one graph realization to
the other.  

Intermediate values of $m$ display a mixed behavior, with the
histogram of convergence times concentrating for small $n$ and 
broadening for larger $n$. This crossover behavior is consistent with
the analytical results of \cite{braun2006m}. Extrapolating this
crossover suggests that $m=100$ is sufficient for obtaining very
accurate results for the range $n\lesssim 10^5$ of interest to us (and
most likely, well above).

Focusing on $m=100$ (data in Figure~\ref{fig:convTimesM100}), we
computed the mean and variance of
$\log(t_{\text{conv}})$, for each value of $n$. These appear to be well fitted by the
following expressions
\begin{align}
\E \log(t_\text{conv}) &\approx 4.3 + 0.22\, \log n\,,\\
\Var\Big(\log(t_\text{conv}) \Big) &\approx 0.063\cdot n^{-1/2}\,.
\end{align}
In other words the typical time complexity of our block coordinate ascent algorithm is
--empirically-- $O(m\, n^{1.22})$ (recall that each iteration comprises
$n$ updates).
%
%*******************************
%
\subsubsection{Determination of the phase transition location}
\label{sec:LambdaC}

As already shown in Section~\ref{sec:numAdel}, the overlap
$\Ove_n(\xsdp)$ undergoes a phase transition at a critical point
$\lsdp(d)$ close to $1$. Namely
$\lim_{n\to\infty}\Ove_n(\xsdp)=0$ for $\lambda\le \lsdp(d)$,
while $\lim_{n\to\infty}\Ove_n(\xsdp)>0$ strictly for $\lambda>\lsdp(d)$.
In order to determine more precisely the phase transition location, we
use the Binder's cumulant method, which is standard in statistical
physics \cite{binder1981finite,landau2014guide}. We summarize the main ideas of this method
for the readers that might not be familiar with this type of analysis.

For a given graph realization $G$, we define $Q(G)$ to be the overlap
achieved by the SDP estimator on that realization, i.e.
\begin{align}
Q(G) \equiv \frac{1}{n}\<\xsdp(G),\bxz\>\, ,
\end{align}
Notice that $Q(G)$ is a random variable taking values in
$[-1,1]$. Also, by the symmetry of the model, its distribution is
symmetric around $0$. Further $\Ove_n(\xsdp) = \E\{|Q(G)|\}$.

\begin{figure}[]
\centering
\includegraphics[width=0.6\textwidth]{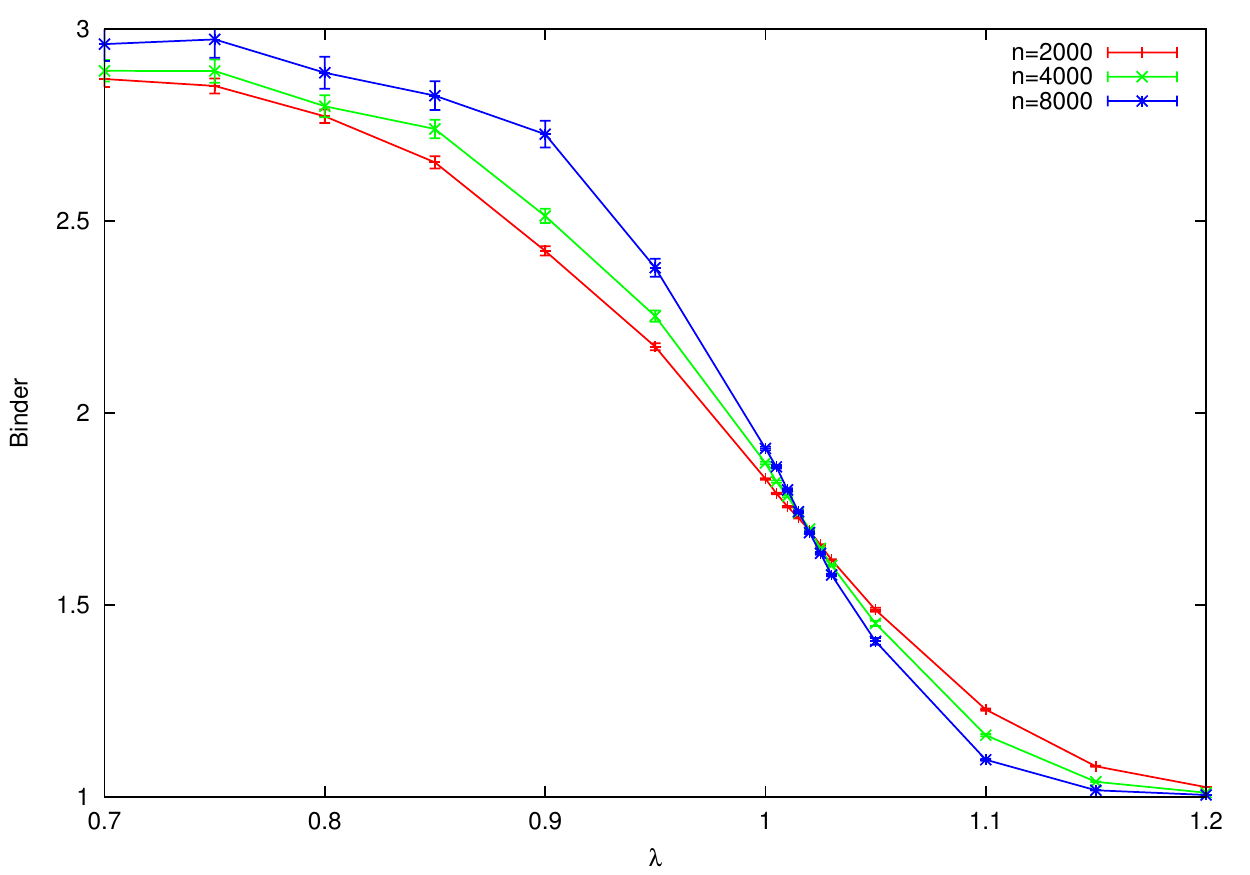}
\caption{Empirical estimates of the Binder cumulant for $d=5$.}
\label{fig:binderLong}
\end{figure}
We define the Binder cumulant by
\begin{align}
\Bind(n,\lambda, d) \equiv
  \frac{\E\big\{Q(G)^4\big\}}{\E\big\{Q(G)^2\big\}^2}\, .\label{eq:BinderDef}
\end{align}
For $\lambda>\lsdp(d)$, we expect $|Q(G)|$ to concentrate around
its expectation $\Ove_n(\xsdp)$, which converges to a non-zero
limit. Hence $\lim_{n\to\infty}\Bind(n,\lambda,d) = 1$.
On the other hand, for $\lambda>\lsdp(d)$, $Q(G)$ concentrates around
$0$, and we expect it to obey a central limit theorem asymptotics,
namely $Q(G)\approx \normal(0,\sigma^2_Q(n))$, with $\sigma^2_Q(n)
\approx \sigma_{Q,*}^2/n$. This implies
$\lim_{n\to\infty}\Bind(n,\lambda,d) = 3$.
Summarizing
\begin{align}
\lim_{n\to\infty}\Bind(n,\lambda,d) = \begin{cases}
3 & \mbox{ if $\lambda<\lsdp(d)$,}\\
1 & \mbox{ if $\lambda>\lsdp(d)$.}
\end{cases} \label{eq:BinderBehavior}
\end{align}
We carried out extensive simulations with the block coordinate ascent,
in order to evaluate the Binder cumulant, and will present our data in
the next plots.
In order to approximate the expectation over the random graph $G$, we
computed empirical  averages over $N_{\text{sample}}$ random graph
samples,  with $N_{\text{sample}}$ chosen so that $N_{\text{sample}} \times n = 6.4\; 10^8$.
(The rationale for using less samples for larger graph sizes is that we
expect statistical uncertainties to decrease with $n$.)

Figure \ref{fig:binderLong} reports a first evaluation of
$\Bind(n,\lambda,d)$ for $d=5$ and a grid of values of
$\lambda$. The results are consistent with the prediction of
Eq.~(\ref{eq:BinderBehavior}).
The approach to the $n\to\infty$ limit is expected to be described by
a finite-size scaling ansatz \cite{cardy2012finite,landau2014guide}
\begin{align}
\Bind(n,\lambda, d) \approx \cF\big(n^{1/\nu}(\lambda-\lsdp(d))\big)\, ,
\end{align}
for a certain scaling function $\cF$, and exponent $\nu$.
Formally, the above approximation is meant to be asymptotically exact
 in the sense that, for any $z$ fixed, letting $\lambda(z,n) =
 \lsdp(p)+n^{-1/\nu}z$, we have $\lim_{n\to\infty}\Bind(n,\lambda(z,n),d) = \cF(z)$. 
We refer to \cite{bollobas2001scaling,dembo2008finite} for recent examples of rigorous
finite-size scaling results in random graph problems.

\begin{figure}[]
\centering
\includegraphics[width=0.6\textwidth]{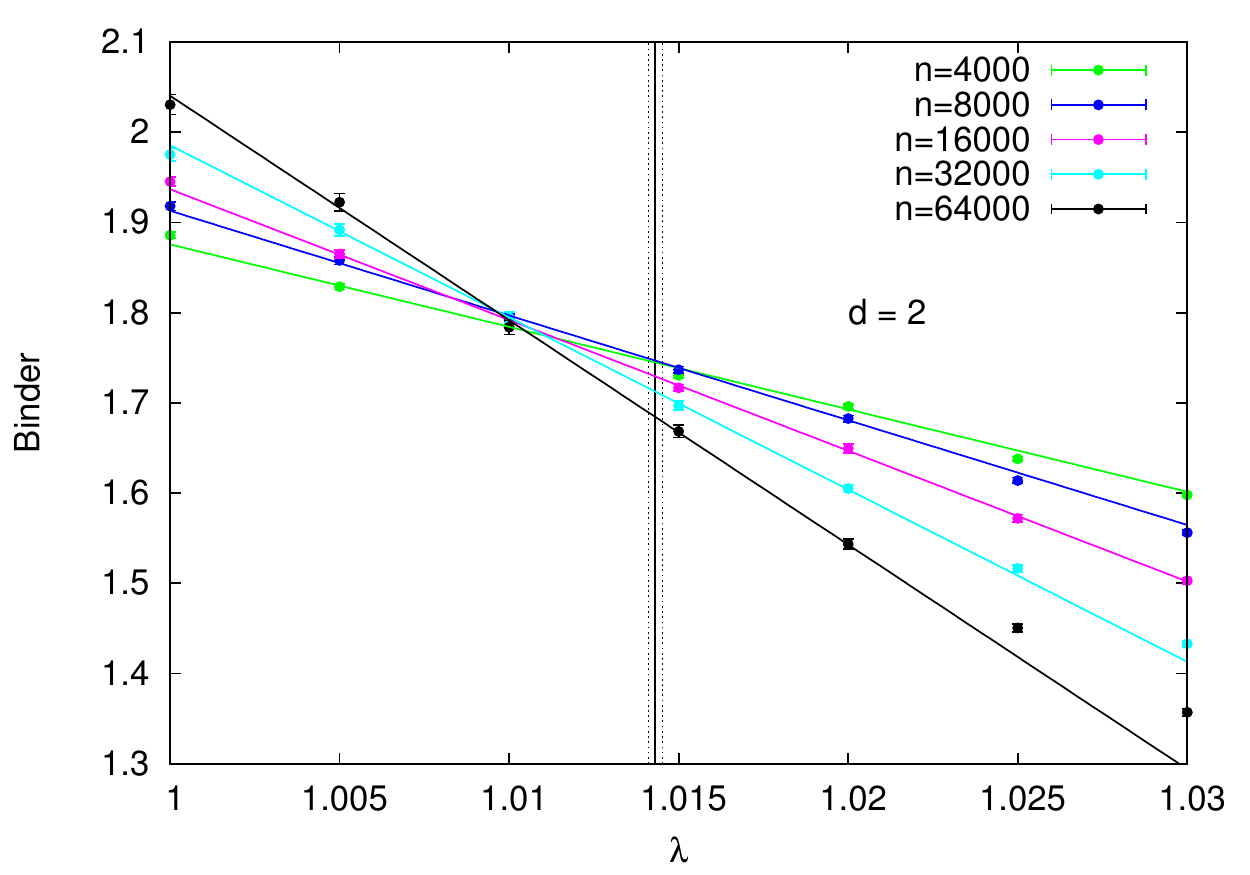}
\includegraphics[width=0.6\textwidth]{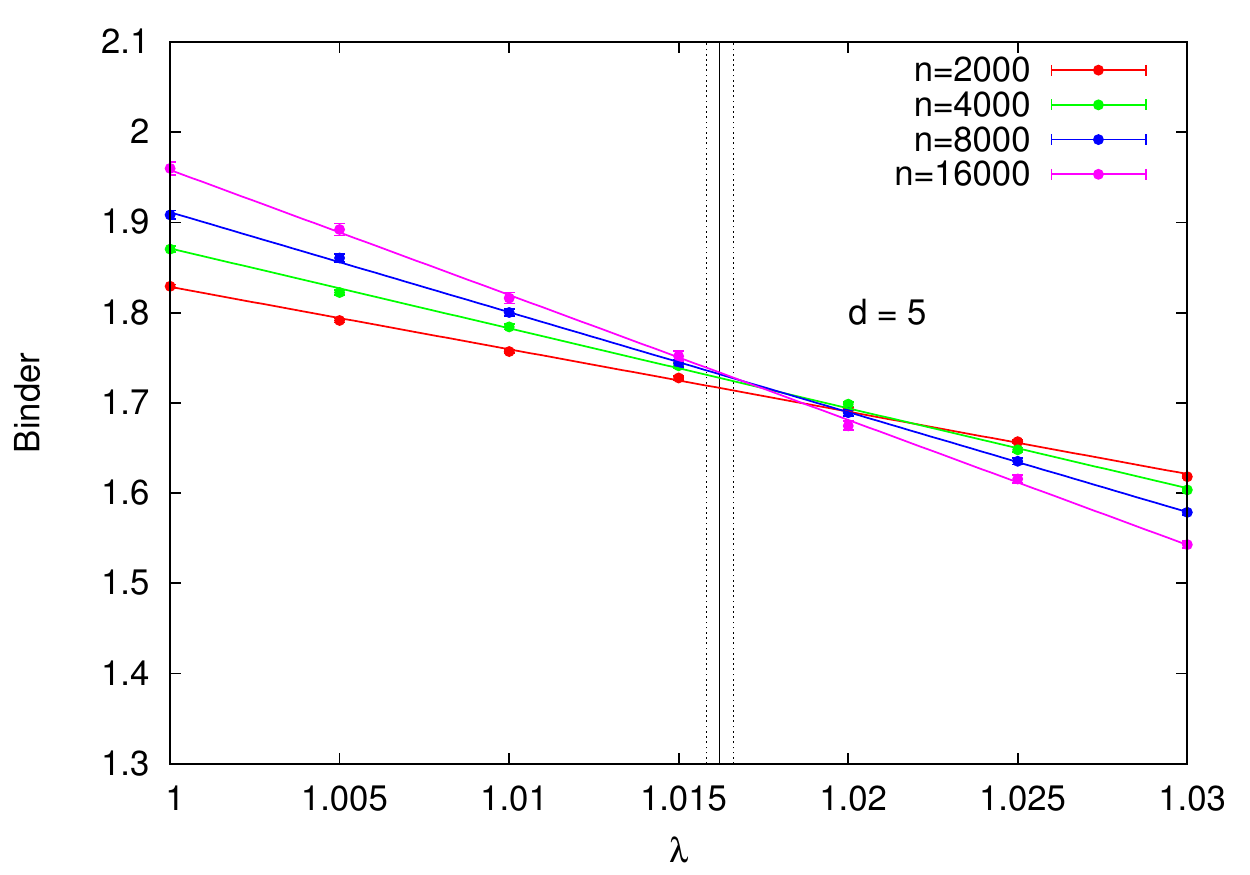}
\includegraphics[width=0.6\textwidth]{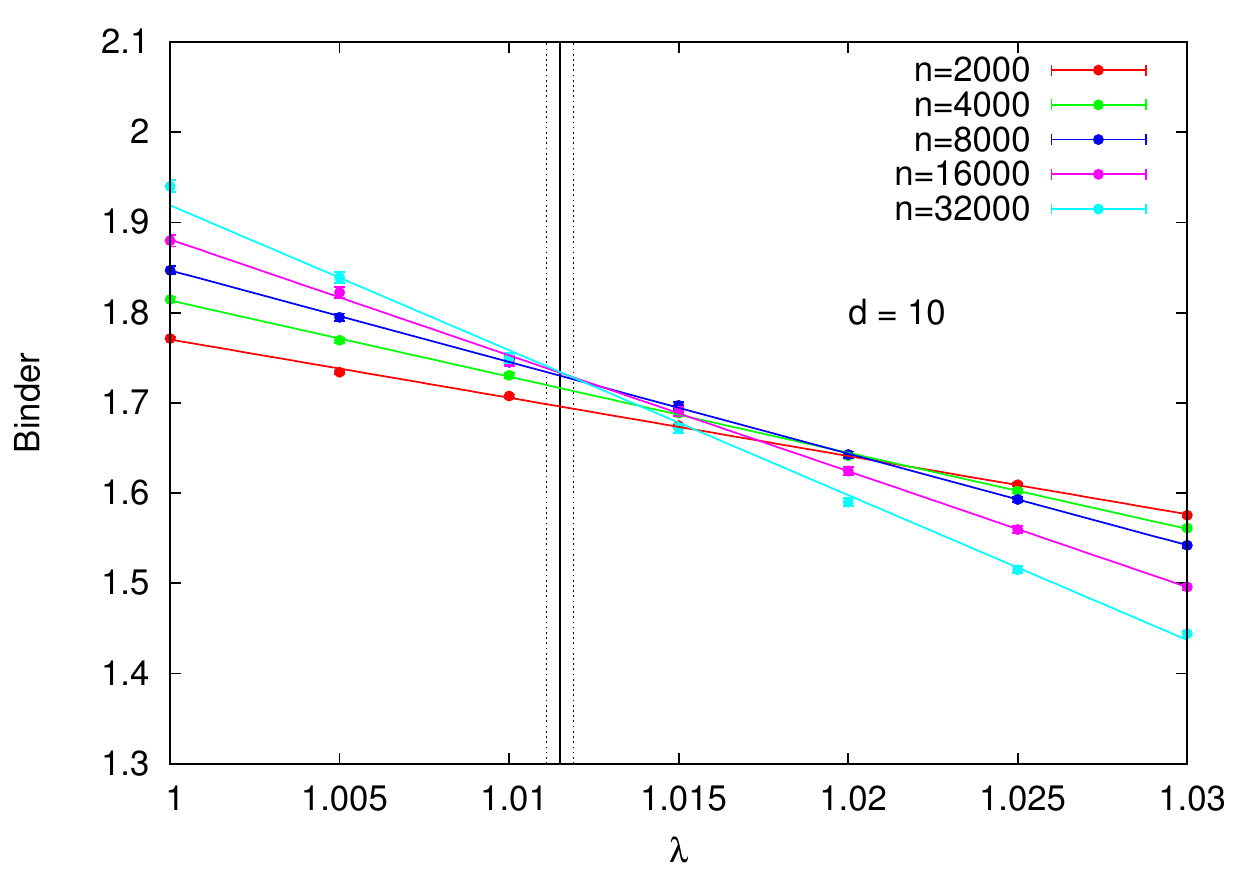}
\caption{Crossings of the Binder parameters mark the critical value of
  $\lambda$. Vertical lines are the analytical estimates for the
  critical point $\tlsdp(d)$ (with dashed lines indicating the
  uncertainty in this estimate, due to numerical solution of the
  recursive distributional equation).}
\label{fig:binder}
\end{figure}

In particular, finite size scaling suggests to estimate $\lsdp$ by the
value of $\lambda$ at which the curves
$\lambda\mapsto\Bind(n,\lambda,d)$, corresponding to different values
of $n$, intersect. In Figure \ref{fig:binder} we report our data for
$d=2$, $5$, $10$, focusing on a small window around the crossing
point. Continuous lines are linear fit to the data, and vertical lines
correspond to the analytical estimates of Section \ref{sec:numLinear}.

We observe that, for large $n$, the crossing point is roughly
independent of the the value of $n$, in agreement with the finite-size
scaling ansatz. 
As a nominal estimate for the critical point, we use
the crossing point $\lambda_{\#}(d)$ of the two Binder cumulant curves corresponding to
the two largest values of $n$, see Fig.~\ref{fig:binder}. These are
$n=32,000$ and $64,000$ for $d=2$, and  $n=16,000$ and $32,000$ for
$d=5$, $10$. We obtain
\begin{align}
\lambda_{\#}(d=2) & = 1.010\, ,\\
\lambda_{\#}(d=5) & = 1.016\, ,\\
\lambda_{\#}(d=10) & = 1.012\, .
\end{align}
These values are broadly consistent with our analytical prediction for
$\tlsdp(d)$.
There appear to be some discrepancy, especially for $d=2$. This might be due to the
graph-size being still too small for extrapolating to $n\to\infty$,
or to the inaccuracy of our calculation based on the vectorial ansatz.

\subsection{Improving numerical results by restricting to the 2-core}

In order to accelerate our numerical experiments presented in
Section~\ref{sec:numAdel} and \ref{sec:numGreedy}, we preprocessed the
graph $G$ by reducing it to its $2$-core. Recall that 
the $k$-core of a graph $G$ is the largest subgraph of $G$, with
minimum degree at least $k$. It can be constructed in linear time by
recursively removing vertices with degree at most $(k-1)$.

In numerical experiments we first generated $G_0$ according to the
model (\ref{eq:SBMDefinition}), then reduced $G_0$ to its $2$-core
$G$, and finally solved the SDP (\ref{eq:SDP}) on $G$. If $G_0$ has size
$n$, and $d>1$, the size of $G$ is still of order $n$ albeit somewhat
smaller \cite{pittel1996sudden}. 

The pruned graph $G\setminus G_0$ is formed with high probability by a
collection of trees with size of order $1$. It is not hard to see that
the SDP estimator can achieve strictly positive overlap on $G_0$ (as
$n\to\infty$) if and only if it does on $G$. Hence, this reduction does
not change the phase transition location. We confirmed numerically
this argument as well.

\end{document}